\documentclass[aps,prd,twocolumn,superscriptaddress,showpacs,longbibliography]{revtex4-1}
\usepackage{amsmath,amsfonts,amsthm,amssymb,xspace,bm,bbold, bbm}
\usepackage{algorithm,algorithmic}

\usepackage{multirow}
\usepackage{xfrac}
\usepackage{epstopdf}
\theoremstyle{plain}
\newtheorem{theorem}{Theorem}%[section]
\newtheorem{lemma}[theorem]{Lemma}
\newtheorem*{lemma*}{Lemma}

\newtheorem{corollary}[theorem]{Corollary}

\newtheorem{definition}[theorem]{Definition}

\theoremstyle{definition}

\newtheorem{assumption}{Assumption}

\newcommand{\R}{\mathbb{R}}

\newcommand{\tr}{{\rm Tr}}
\DeclareMathOperator*{\argmin}{argmin}

\newcommand{\ip}[2]{\left\langle #1, #2 \right\rangle}

\newcommand{\norm}[1]{\left \Vert #1\right \Vert}
\newcommand{\dist}{{\rm{\textsc{Dist}}}}

\newcommand{\obs}{y}
\newcommand{\linmap}{\mathcal{M}}

\newcommand{\C}{\mathcal{C}}
\newcommand{\X}{\rho}

\newcommand{\U}{A}
\newcommand{\V}{B}
\newcommand{\E}{\mathcal{E}}
\newcommand{\gradf}{\nabla f}
\newcommand{\gradg}{\nabla g}
\newcommand{\Xo}{\rho_\star}
\newcommand{\Uo}{A_\star}
\newcommand{\weta}{\widetilde{\eta}}
\newcommand{\Rus}{R_{A_t}^\star}
\newcommand{\Uw}{\widetilde{A}}

\newcommand{\trace}{\text{Tr}}
\newcommand{\A}{A}
\newcommand{\B}{B}
\newcommand{\Ao}{A_\star}
\newcommand{\rhoo}{\rho_{\star}}
\usepackage{amsmath,amsfonts,amsbsy,amssymb,amsthm}
\usepackage{bbold,bbm,bm}
\usepackage{graphicx}
\usepackage[colorlinks=true,citecolor=black]{hyperref}
\usepackage{booktabs}
\usepackage[framed,numbered,autolinebreaks,useliterate]{mcode}

\newcommand{\DM}{\mathcal{S}}

\AtBeginDocument{
\heavyrulewidth=.08em
\lightrulewidth=.05em
\cmidrulewidth=.03em
\belowrulesep=.65ex
\belowbottomsep=0pt
\aboverulesep=.4ex
\abovetopsep=0pt
\cmidrulesep=\doublerulesep
\cmidrulekern=.5em
\defaultaddspace=.5em
}

\begin{document}

\title{Provable quantum state tomography via non-convex methods}
\author{Anastasios Kyrillidis}
\thanks{Author to whom correspondence should be addressed}
\email{anastasios.kyrillidis@ibm.com}
\affiliation{IBM T. J. Watson Research Center}
\author{Amir Kalev}
\email{amirk@umd.edu}
\affiliation{University of Maryland}
\author{Dohuyng Park}
\email{dhpark@utexas.edu}
\affiliation{Facebook}
\author{Srinadh Bhojanapalli}
\email{srinadh@ttic.edu}
\affiliation{Toyota Technological Institute at Chicago}
\author{Constantine Caramanis}
\email{constantine@utexas.edu}
\affiliation{University of Texas at Austin}
\author{Sujay Sanghavi}
\email{sanghavi@mail.utexas.edu}
\affiliation{University of Texas at Austin}

\begin{abstract}
With nowadays steadily growing quantum processors, it is required to develop new quantum tomography tools that are tailored for high-dimensional systems. 
In this work, we describe such a computational tool, based on recent ideas from non-convex optimization. 
The algorithm excels in the compressed-sensing-like setting, where only a few data points are measured from a low-rank or highly-pure quantum state of a  high-dimensional system. 
We show that the algorithm can practically be used in quantum tomography problems that are beyond the reach of convex solvers, and, moreover, is faster than other state-of-the-art non-convex approaches. 
Crucially, we prove that, despite being a non-convex program, under mild conditions, the algorithm is guaranteed to converge to the global minimum of the problem; thus, it constitutes a provable quantum state tomography protocol.
\end{abstract}

%\pacs{}
%\keywords{} 
\maketitle
\section{Introduction}\label{sec:intro}
Like any other processor, the behavior of a quantum information processor must be characterized, verified, and certified. 
Quantum state tomography (QST) is one of the main tools for that purpose \citep{altepeter2005photonic}.
Yet, it is generally an inefficient procedure, since the number of parameters that specify quantum states, grows exponentially with the number of sub-systems. 
This inefficiency has two practical manifestations:  
$(i)$ without any prior information, a vast number of data points needs to be collected \citep{altepeter2005photonic}; 
$(ii)$ once the data is gathered, a numerical procedure should be executed on an exponentially-high dimensional space, in order to infer the quantum state that is most consistent with the observations. 
Thus, to perform QST on nowadays steadily growing quantum processors~\citep{zhang2017observation, IBMQ}, we must introduce novel, more efficient, techniques for its completion.

Since often the aim in quantum information processing is to coherently manipulate \emph{pure} quantum states (\emph{i.e.}, states that can be equivalently described with rank-1, positive semi-definite (PSD) density matrices), the use of such prior information is the \emph{modus operandi} towards making QST manageable, with respect to the  amount of data required~\citep{flammia2005minimal, gross2010quantum, heinosaari2013quantum, baldwin2016strictly}. 
Compressed sensing (CS) \citep{donoho2006compressed, baraniuk2007compressive} --and its extension to guaranteed low-rank approximation \citep{recht2010guaranteed, candes2011tight, candes2009exact}-- has been applied to QST \citep{gross2010quantum, kalev2015quantum} within this context. 
In particular, it has been proven \citep{gross2010quantum,flammia2011direct, liu2011universal, kalev2015quantum} that convex programming  guarantees robust estimation of pure $n$-qubit states from much less information than common wisdom dictates, with overwhelming probability.
 
These advances, however, leave open the question of how efficiently one can estimate exponentially large-sized quantum states, from a limited set of observations. 
Since convex programming is susceptible of \emph{provable performance}, typical QST protocols rely on convex programs~\cite{gross2010quantum, kalev2015quantum, baldwin2016strictly}. 
Nevertheless, their Achilles' heel remains the high computational and storage complexity. 
In particular, due to the PSD nature of density matrices, a key step is the repetitive application of Hermitian eigenproblem solvers. 
Such solvers include the well-established family of Lanczos methods \citep{kokiopoulou2004computing, baglama2005augmented, baglama2006restarted, cullum1983lanczos}, the Jacobi-Davinson SVD type of methods \citep{hochstenbach2001jacobi}, as well as preconditioned hybrid schemes \citep{wu2015preconditioned}, among others; see also the recent article in \citep{stathopoulos2017extended} for a more complete overview. 
Since --at least once per iteration-- a full eigenvalue decomposition is required in such convex programs, 
these eigensolvers contribute a $\mathcal{O}((2^n)^3)$ computational complexity, where $n$ is the number of qubits of the quantum processor. 
It is obvious that the recurrent application of such eigensolvers makes convex programs impractical, even for quantum systems with a relatively small number $n$ of qubits \cite{haffner2006scalable, gross2010quantum}. 

Ergo, to improve the efficiency of QST, we need to complement it with numerical algorithms that can efficiently handle large search spaces using limited amount of data, while having rigorous performance guarantees. 
This is the purpose of this work. 
Inspired by the recent advances on finding the global minimum in non-convex problems  \citep{sun2015guaranteed, zhao2015nonconvex, chen2015fast, jain2015computing, tu2015low, bhojanapalli2016dropping, park2016non, ge2016matrix, park2016provable, park2016finding, li2016recovery, li2016symmetry, tran2016extended, wang2017universal, ge2017no}, we propose the application of alternating gradient descent in QST, that operates directly on the assumed low-rank structure of the density matrix.
The algorithm --named Projected Factored Gradient Decent (\texttt{ProjFGD}) and described below in detail-- is based on the recently analyzed non-convex method in \citep{bhojanapalli2016dropping} for PSD matrix factorization problems.
The added twist is the inclusion of further constraints in the optimization program, that makes it applicable for tasks such as QST.  

In general, finding the global minimum in non-convex problems is a hard problem. 
However, our approach assumes certain regularity conditions --that are, however, satisfied by common CS-inspired protocols in practice~\cite{gross2010quantum, kalev2015quantum, baldwin2016strictly}-- and a good initialization --which we make explicit in the text; both lead to a fast and \emph{provable} estimation of the state of the system, even with limited amount of data. 
Our numerical experiments show that our scheme outperforms in practice state-of-the-art approaches for QST.

Apart from the QST application, our aim is to broaden the results on efficient, non-convex recovery within the set of constrained low-rank matrix problems. 
Our developments maintain a connection with analogous results in convex optimization, where standard assumptions are made.
However, this work goes beyond convexity, in an attempt to justify recent findings that non-convex methods show significant acceleration, as compared to state-of-the-art convex analogs.

\section{Quantum state tomography setup}{\label{sec:QST}}
We begin by describing  the problem of QST. 
We are focusing here on  QST of a highly-pure $n$-qubit state from Pauli measurements. 
In particular, let $y \in \mathbb{R}^m$ be the measurement vector with elements $y_i = \tfrac{2^n}{\sqrt{m}}\text{Tr}(P_i \cdot \rho_\star)+e_i,~i = 1, \dots, m$, for some measurement error $e_i$. %, following i.i.d. Gaussian statistics. 
Here, $\rho_\star$ denotes the unknown $n$-qubit density matrix, associated with the pure quantum state;  
$P_i \in \mathbb{C}^{2^n \times 2^n}$ is a randomly chosen Pauli observable; and the normalization $\tfrac{2^n}{\sqrt{m}}$ is chosen to follow the results of~\cite{liu2011universal}.
For brevity, we denote $\linmap : \mathbb{C}^{2^n \times 2^n} \rightarrow \mathbb{R}^m$ as the linear ``sensing" map, such that $(\linmap(\rho))_i = \tfrac{2^n}{\sqrt{m}} \text{Tr}(P_i \cdot \rho)$, for $i = 1, \dots, m$.

For the above setting, we assume that the data is given in the form of expectation values of $n$-qubit Pauli observables. 
An $n$-qubit Pauli observable is given by $P=\otimes_{j=1}^n s_j$ where $s_j\in\{\mathbb{1},\sigma_x,\sigma_y,\sigma_z\}$.
There are $4^n$ such observables in total. 
In general, one needs to have the expectation values of all $4^n$ Pauli observables to uniquely reconstruct $\rho_\star$; \emph{i.e.}, performing for many repetitions/shots the same experiment and then taking the expectation of the results (\emph{e.g.}, counts of qubit registers).
Since $\rho_\star$ is a highly-pure density matrix, we apply the CS result on Pauli measurements \citep{gross2010quantum, liu2011universal},  that guarantees robust estimation, with high probability, from just $m={\cal O}(r 2^n n^6)$ randomly chosen Pauli observables (in expectation). 
Key property to achieve this is the \emph{restricted isometry property} \citep{liu2011universal}:
\begin{definition}[Restricted Isometry Property (RIP) for Pauli measurements]\label{def:RIP1}
Let $\linmap : \mathbb{C}^{2^n \times 2^n} \rightarrow \mathbb{R}^m$ be a linear map, such that $(\linmap(\rho))_i = \tfrac{2^n}{\sqrt{m}} \emph{\text{Tr}}(P_i \cdot \rho)$, for $i = 1, \dots, m$.
Then, with high probability over the choice of $m = \tfrac{c}{\delta_r^2}\cdot (r 2^n n^6)$ Pauli observables $P_i$, where $c > 0$ is an absolute constant, $\mathcal{M}$ satisfies the $r$-RIP with constant $\delta_r$, $ 0 \leq \delta_r < 1$; i.e.,
\begin{align*}
(1 - \delta_r)\|\rho\|_F^2 \leq \|\mathcal{M}(\rho)\|_2^2 \leq (1 + \delta_r)\|\rho\|_F^2,
\end{align*} 
where $\|\cdot\|_F$ denote the  Frobenius norm, is satisfied $\forall \rho \in \mathbb{C}^{2^n \times 2^n}$ such that ${\rm \text{rank}}(\rho) \leq r$. 
\end{definition} 

An accurate estimation of $\rhoo$ is obtained  by solving, essentially, a convex optimization problem constrained to the set of quantum states~\citep{kalev2015quantum}, consistent with the measured data.  
Two such convex program examples are:
\begin{equation}
	\begin{aligned}
		& \underset{\rho \in \mathbb{C}^{2^n \times 2^n}}{\text{minimize}}
		& & \text{Tr}(\rho) \\
		& \text{subject to}
		& & \rho \succeq 0, \\
		& & & \|y - \mathcal{M}(\rho)\|_2 \leq \epsilon,
	\end{aligned} \label{eq:CVX1}
\end{equation}
and,
\begin{equation}
	\begin{aligned}
		& \underset{\rho \in \mathbb{C}^{2^n \times 2^n}}{\text{minimize}}
		& & \tfrac{1}{2} \cdot \|y - \mathcal{M}(\rho)\|_2^2 \\
		& \text{subject to}
		& & \rho \succeq 0, \\
		& & & \text{Tr}(\rho) \leq 1,
	\end{aligned} \label{eq:CVX2}
\end{equation}
where $\rho \succeq 0$ captures the positive semi-definite assumption, $\|\cdot\|_2$ is the vector Euclidean $\ell_2$-norm, and $\epsilon >0$ is a parameter related to the error level in the model.
Key in both programs is the combination of the PSD and the trace constraints: combined, they constitute the tightest convex relaxation to the low-rank, PSD structure of the unknown $\rho_\star$; see also \citep{recht2010guaranteed}. 
%We note that if $\mathcal{M}$ corresponds to a positive-operator valued measure (POVM), or includes the identity operator, then the explicit trace constraint is redundant.

As was discussed in the introduction, the problem with convex  programs, such as~\eqref{eq:CVX1} and~\eqref{eq:CVX2}, is their inefficiency when applied in high-dimensional systems: most practical solvers for \eqref{eq:CVX1}-\eqref{eq:CVX2} are iterative and handling PSD constraints adds an immense complexity overhead per iteration, especially when $n$ is large; see also Section \ref{sec:experiments}.

In this work, we propose to use non-convex programming for QST of low-rank density matrices, which leads to higher efficiency than typical convex programs. We achieve this by restricting the optimization over the intrinsic non-convex structure of rank-$r$ PSD matrices. 
This allow us to ``describe" an $2^n \times 2^n$ PSD matrix with only $\mathcal{O}(2^n r)$ space, as opposed to the $\mathcal{O}\left((2^n)^2\right)$ ambient space.
Even more substantially, our program has theoretical guarantees of global convergence,  similar to the guarantees of convex programming, while maintaining faster preformace than the latter.
These properties make our scheme ideal to complement the CS methodology for QST in practice.

\section{Projected Factored Gradient Decent algorithm}

\emph{Optimization criterion recast:} 
At its basis, the Projected Factored Gradient Decent (\texttt{ProjFGD}) algorithm transforms convex programs, such as in \eqref{eq:CVX1}-\eqref{eq:CVX2}, by imposing the factorization of a $d\times d$ PSD matrix $\rho$ such that $\rho = A A^\dagger$. 
This factorization, popularized by Burer and Monteiro \cite{burer2003nonlinear, burer2005local} for solving semi-definite convex programming instances, naturally encodes the PSD constraint, removing the expensive eigen-decomposition projection step. 
For concreteness, we focus here on the  convex program \eqref{eq:CVX2}, where $d=2^n$. 
In order to encode the trace constraint, \texttt{ProjFGD} enforces additional constraints on $A$. In particular, the requirement that $\tr(\rho) \leq 1$ is translated to the \emph{convex} constraint $\|A\|_F^2 \leq 1$, where $\|\cdot\|_F$ is the Frobenius norm. 
The above recast the program~\eqref{eq:CVX2} as a non-convex program: 
\begin{equation}
	\begin{aligned}
		& \underset{A \in {\mathbb C}^{d \times r}}{\text{minimize}}
		& & \tfrac{1}{2} \cdot \|y - \mathcal{M}(AA^\dagger)\|_2^2 \\
		& \text{subject to}
		& & \|A\|_F^2 \leq 1.
	\end{aligned} \label{eq:nonCVX}
\end{equation}

Observe that, while the constraint set is convex, the objective is no longer convex due to the bilinear transformation of the parameter space $\rho = AA^\dagger$.
Such criteria have been studied recently in machine learning and signal processing applications \citep{sun2015guaranteed, zhao2015nonconvex, chen2015fast, jain2015computing, tu2015low, bhojanapalli2016dropping, park2016non, ge2016matrix, park2016provable, park2016finding, li2016recovery, li2016symmetry, tran2016extended, wang2017universal, ge2017no}.
Here, the added twist is the inclusion of further matrix norm constraints, that makes it proper for tasks such as QST; as we show in  Appendices~\ref{sec:proof} and~\ref{sec:init}, such addition complicates the algorithmic analysis.

The prior knowledge that $\text{rank}(\rho_\star) \leq r_\star$ is imposed in the program by setting $A \in {\mathbb C}^{d \times r_\star}$. 
In real experiments, the state of the system, $\rho_\star$, could be full rank, but often is highly-pure with only few dominant eigenvalues. 
In this case, $\rho_\star$ is well-approximated by a low-rank matrix of rank $r$, which can be much smaller than $r_\star$, similar to the CS methodology. 
Therefore in the \texttt{ProjFGD} protocol we set  $A \in {\mathbb C}^{d\times r}$. 
In this form, $A$ contains much less variables to maintain and optimize than a $d \times d$ PSD matrix, and thus it is easier to update and to store its iterates.

An important issue in optimizing \eqref{eq:nonCVX} over the factored space is the existence of non-unique possible factorizations for a given $\X$. 
To see this, if $\X = AA^\dagger$, then for any unitary matrix $R \in \mathbb{C}^{r \times r}$ such that $RR^\dagger  = I$, we have $\X = \widehat{A} \widehat{A}^\dagger$, where $\widehat{A} = AR$.
Since we are interested in obtaining a low-rank solution in the original space, we need a notion of distance to the low-rank solution $\Xo$ over the factors. 
We use the following unitary-invariant distance metric:
\begin{definition}{\label{def:metric}}
Let matrices $\A, \Ao \in \mathbb{C}^{d \times r}$. Define:
\begin{align*}
\dist\left(A, \Ao\right) :=\min_{R: R \in \mathcal{U}} \norm{A - \Ao R}_F, 
\end{align*} where $\mathcal{U}$ is the set of $r \times r$ unitary matrices.
\end{definition}

\emph{The \texttt{ProjFGD} algorithm:} 
At heart, \texttt{ProjFGD} is a projected gradient descent algorithm over  the variable $A$.
The pseudocode is provided in  Algorithm~\ref{algo:projFGD}.

\begin{algorithm}[H]
	\caption{\texttt{ProjFGD} pseudocode for~\eqref{eq:nonCVX}}
\label{algo:projFGD}
	\begin{algorithmic}[1]		
		\STATE \textbf{Input:} Function $f$, target rank $r$, \# iterations $T$. 
		\STATE \textbf{Output:} $\rho = A_T A_T^\dagger$. 	\\ \vspace{-0.2cm}
		\hrulefill
		\STATE \quad Initialize $\rho_0$ randomly or set $\rho_0 := \sfrac{2}{\widehat{L}} \cdot \Pi_{\mathcal{C}'}\left(\mathcal{M}^\dagger\left(y\right) \right)$.
		\STATE \quad Set $\A_0 \in {\mathbb C}^{d \times r}$ such that $\rho_0 = \A_0 \A_0^\dagger$.
		\STATE \quad Set step size $\eta$ as in~\eqref{eq:step_size}.
		\STATE \quad \textbf{for} {$t=0$ to $T-1$} \textbf{do}
			\STATE \quad \quad $A_{t+1} = \Pi_{\mathcal{C}}\left(A_t - \eta \gradf(A_t A_t^\dagger)  \cdot A_t\right)$.
		\STATE \quad \textbf{end}
	\end{algorithmic}
\end{algorithm}

First, some properties of the objective in \eqref{eq:nonCVX}.
Denote $g(A) = \tfrac{1}{2} \cdot \|y - \mathcal{M}(AA^\dagger)\|_2^2$ and $f(\rho) = \tfrac{1}{2} \cdot \|y - \mathcal{M}(\rho)\|_2^2$. 
Due to the symmetry of $f$, \emph{i.e.}, $f(\rho) = f(\rho^\dagger)$, the gradient of $g(A)$ with respect to $A$ variable is given by 
$$\gradg(A)  = \left(\gradf(\rho) + \gradf(\rho)^\dagger\right) \cdot A = 2 \gradf(\rho) \cdot A,$$ where $\gradf(\rho) = -2 \mathcal{M}^\dagger \left(y - \mathcal{M}(\rho)\right)$, and $\mathcal{M}^\dagger$ is the adjoint operator for $\mathcal{M}$. 
For the Pauli measurements case we consider in this paper, the adjoint operator for an input vector $b \in \R^{m}$ is $\mathcal{M}^\dagger(b) = \tfrac{2^n}{\sqrt{m}}\sum_{i = 1}^m b_i P_i$.

Let $\Pi_\mathcal{C}(B)$ denote the projection a matrix $B \in \mathbb{C}^{d \times r}$ onto the set $\mathcal{C} = \left\{ A : A \in \mathbb{C}^{d \times r}, ~\|A\|_F^2 \leq 1\right\}$, as in \eqref{eq:nonCVX}.
For this particular $\mathcal{C}$, $\Pi_\mathcal{C}(B) = \xi(B) \cdot B$, where $\xi(\cdot) \in (0,1]$; in the case where $\xi(\cdot) = 1$, $B \in \mathcal{C}$ already.
As an initialization, we compute $\Pi_{\mathcal{C}'}(\cdot)$, which denotes the projection onto the set of PSD matrices with trace bound $\text{Tr}(\rho) \leq 1$; we discuss later in the text how to complete this step in practice.

The main iteration of \texttt{ProjFGD} is in line 7 of Algorithm \ref{algo:projFGD}, where it applies a simple update rule over the factors:
\begin{equation}
A_{t+1} = \Pi_{\mathcal{C}}\left(A_t - \eta \gradf(A_t A_t^\dagger)  \cdot A_t\right) \nonumber,
\end{equation} 
Observe that the input argument in $\Pi_{\C}(\cdot)$ is:
$$A_t - \eta \gradf(A_t A_t^\dagger)  \cdot A_t = A_t - \eta \gradg(A_t);$$ 
\emph{i.e.}, it performs gradient descent over $A$ variable, with step size $\eta$.
Any constants are ``absorbed" in the step size selection, for clarity.

Two vital components of our algorithm are: $(i)$ the initialization step and, $(ii)$ step size selection.

\subsection{Initialization $\rho_0$}\label{subsec:initialization}
Due to the bilinear structure in \eqref{eq:nonCVX}, at first glance it is not clear whether the factorization $\rho = AA^\dagger$ introduces \emph{spurious} local minima, \emph{i.e.}, local minima that do not exist in \eqref{eq:CVX1}-\eqref{eq:CVX2}, but are ``created'' after the substitution $\rho = AA^\dagger$.
This necessitates careful initialization, in order to obtain the global minimum.

Before describing our initialization procedure, we find it helpful to first discuss an initialization procedure for an altered version of  \eqref{eq:nonCVX} where  trace constraints are excluded. 
In this case, \eqref{eq:nonCVX} transforms to:
\begin{equation}
	\begin{aligned}
		& \underset{A \in {\mathbb C}^{d \times r}}{\text{minimize}}
		& & \tfrac{1}{2} \cdot \|y - \mathcal{M}(AA^\dagger)\|_2^2.
	\end{aligned} \label{eq:nonCVX1}
\end{equation}
Under this setting, the following theory stems from \citep{park2016non}:
\begin{theorem}{\label{thm:nospurious1}}
Suppose the unknown $\rho_\star$ is a rank-$r$ density matrix with a non-unique factorization $\rho_\star = \Ao \Ao^{\dagger}$, for $\Ao \in \mathbb{C}^{d \times r}$.
Under the noiseless model, the observations satisfy $y = \mathcal{M}(\rhoo)$.
Assuming that the linear map $\mathcal{M}$ satisfies the restricted isometry property in Definition \ref{def:RIP1}, with constant $\delta_{4r} \lessapprox 0.0363$, any critical point $A$ satisfying first- and second-order optimality conditions is a global minimum.
\end{theorem}

\begin{corollary}{\label{thm:nospurious2}}
Suppose the unknown $\rho_\star$ is a full rank density matrix and let $\rho_{\star, r}$ denote its best rank-$r$ approximation, in the Eckart-Young-Minsky-Steward sense \citep{eckart1936approximation, stewart1993early}. 
Let $\rho_{\star, r}$ have a non-unique factorization $\rho_{\star, r} = \Ao \Ao^{\dagger}$, for $\Ao \in \mathbb{C}^{d \times r}$.
Under the noiseless model, the observations satisfy $y = \mathcal{M}(\rhoo)$.
Assuming that the linear map $\mathcal{M}$ satisfies the restricted isometry property in Definition \ref{def:RIP1}, with constant $\delta_{4r} \leq \sfrac{1}{200}$, any critical point $A$ satisfying first- and second-order optimality conditions satisfy:
\begin{align*}
\dist(\A, \Ao) \leq \tfrac{1250}{3 \sigma_r(\rhoo)} \cdot \left\| \mathcal{M}\left(\rhoo - \rho_{\star, r} \right)\right\|_2.
\end{align*}
\end{corollary}

In plain words, under the noiseless model and with high probability (that depends on the random structure of the sensing map $\mathcal{M}$), Theorem \ref{thm:nospurious1} states that the non-convex change of variables $\rho = AA^\dagger$ does not introduce any spurious local minima in the low-rank $\rhoo$ case, and random initialization is sufficient for Algorithm \ref{algo:projFGD} to find the global minimum, assuming a proper step size selection.
Further, when $\rhoo$ is not low-rank, there are low-rank solutions $\rho_{\star, r}$, close to $\rhoo$; how much close is a function of the spectrum of $\rhoo$ and its best rank-$r$ approximation residual (Corollary \ref{thm:nospurious2}).
In these cases, Step 3 in Algorithm boils down to random initialization.

The above cases hold when $\mathcal{M}$ corresponds to a POVM; then the explicit trace constraint is redundant.
In the more general case where the trace constraint is present, a different approach is followed. 
In that case, the initial point $\rho_0$ is set as $\rho_0 := \sfrac{1}{\widehat{L}} \cdot \Pi_{\C'}\left(-\nabla f(0) \right)=\sfrac{2}{\widehat{L}} \cdot \Pi_{\C'}\left(\mathcal{M}^\dagger\left(y\right) \right)$, where $\Pi_{\C'}(\cdot)$ denotes the projection onto the set of PSD matrices $\rho$ that satisfy $\text{Tr}(\rho) \leq 1$.
Here, $\widehat{L}$ represents an approximation of $L$, where $L$ is such that for all rank-$r$ matrices $\rho, \zeta$:
\begin{equation}{\label{eq:lipschitz}}
\norm{\gradf\left(\rho\right) - \gradf\left(\zeta\right)}_F \leq L \cdot \norm{\rho - \zeta}_F.
\end{equation}
(This also means that $f$ is \emph{restricted gradient Lipschitz continuous} with parameter $L$. We defer the reader to the Appendix~\ref{sec:proof} for more information). 
In practice, we set $\widehat{L} \in (1,2)$.

This is the only place in the algorithm where eigenvalue-type calculation is required. 
The projection $\Pi_{\C'}(\cdot)$ is given in \citep{gonccalves2016projected}.
Given $\mathcal{M}^\dagger\left(y\right)$, it is described with the following criterion:
\begin{equation}
	\begin{aligned}
		& \underset{\rho_0 \in \mathbb{C}^{d \times d}}{\text{minimize}}
		& & \tfrac{1}{2} \cdot \|\rho_0 - \mathcal{M}^\dagger\left(y\right)\|_F^2 \\
		& \text{subject to}
		& & \rho_0 \succeq 0, \\
		& & & \text{Tr}(\rho_0) \leq 1,
	\end{aligned} \label{eq:projection}
\end{equation}
To solve this problem, we first compute its eigen-decomposition $\mathcal{M}^\dagger\left(y\right) := \Phi \Lambda \Phi^\dagger$, where $\Phi$ is a unitary matrix containing the eigenvectors of the input matrix.
Due to the fact that the Frobenius norm is invariant under unitary transformations, \citep{gonccalves2016projected} proves that $\rho_0 = \Phi \widehat{\Lambda} \Phi^\top$, where $\widehat{\Lambda}$ is a diagonal matrix, computed via:
\begin{equation}
	\begin{aligned}
		& \underset{\widehat{\Lambda}}{\text{minimize}}
		& & \tfrac{1}{2} \cdot \|\widehat{\Lambda} - \Lambda\|_F^2 \\
		& \text{subject to}
		& & \sum_i \widehat{\Lambda}_{ii} \leq 1, \\
		& & & \widehat{\Lambda} \succeq 0.
	\end{aligned} \label{eq:projection2}
\end{equation}
The last part can be easily solved using the projection onto the unit simplex \citep{michelot1986finite, duchi2008efficient, kyrillidis2013sparse}.

Alternatively, in practice, we could just use a standard projection onto the set of PSD matrices $\rho_0 :=\sfrac{2}{\widehat{L}} \cdot \Pi_{+}\left(\mathcal{M}^\dagger\left(y\right) \right)$; our experiments show that it is sufficient and can be implemented by any off-the-shelf eigenvalue solver.
In that case, the algorithm generates an initial matrix $A_0 \in \mathbb{C}^{d \times r}$ by truncating the computed eigen-decomposition, followed by a projection onto the convex set,  $\C$, defined by set of constraints in the program, $A_0 = \Pi_{\C}( \widetilde{A}_0 )$.  
In our case $\C=\{A \in {\mathbb C}^{d \times r}:\|A\|_F^2 \leq 1\}$. 
Note again that the projection operation is a simple \emph{entry-wise scaling}, for $A\notin\C$, $\Pi_{\C}(A)= \xi(A) \cdot A$, where $\xi(A) = \|A\|_F^{-1}$.

Apart from the procedure mentioned above, we could also use more specialized spectral methods for initialization~\cite{chen2015fast, zheng2015convergent} or, alternatively, run convex algorithms, such as~\eqref{eq:CVX2} for only a few iterations. 
However, this choice often leads to an excessive number of full or truncated eigenvalue decompositions~\cite{tu2015low}, which constitutes it a non-practical approach. 

The discussion regarding the step size and what type of guarantees we obtain is discussed next. 

\subsection{Step size selection and theoretical guarantees}

Focusing on~\eqref{eq:nonCVX}, we provide theoretical guarantees for \texttt{ProjFGD}. Our theory dictates a specific \emph{constant} step size selection that guarantees convergence to the global minimum, assuming a satisfactory initial point $\rho_0$ is provided.

Let us first describe the local convergence rate guarantees of \texttt{ProjFGD}.
\begin{theorem}[Local convergence rate for QST]\label{thm:main2}
Let $\rhoo$ be a rank-$r$ quantum state density matrix of an $n$-qubit system with a non-unique factorization $\rho_\star = \Ao \Ao^{\dagger}$, for $\Ao \in \mathbb{C}^{2^n \times r}$. Let $y\in{\mathbb R}^m$ be the measurement vector of $m={\cal O}(r n^6 2^n)$ random $n$-qubit Pauli observables, and $\mathcal{M}$ be the corresponding sensing map, such that $y_i = \left(\mathcal{M}(\rhoo)\right)_i + e_i, ~\forall i = 1, \dots, m$. Let the step $\eta$ in \texttt{ProjFGD} satisfy:
\begin{align}\label{eq:step_size}
\eta \leq \tfrac{1}{128\left(\widehat{L} \sigma_1(\rho_0) + \sigma_1(\gradf(\rho_0))\right)},
\end{align}
where $\sigma_1(\rho)$ denotes the leading singular value of $\rho$.
Here, $\widehat{L} \in (1,2)$ and $\rho_0 = \A_0 \A_0^\dagger$ is the initial point such that:
\begin{align*}
\dist(\U_0, \Uo) \leq \gamma' \sigma_{r}(\Uo),
\end{align*}
for $\gamma' := c \cdot \tfrac{(1-\delta_{4r})}{(1+\delta_{4r})} \cdot \tfrac{\sigma_r(\Xo)}{\sigma_1(\Xo)}, ~c \leq \tfrac{1}{200}$, where $\delta_{4r}$ is the RIP constant. 
Let $\U_t$ be the estimate of \texttt{ProjFGD} at the $t$-th iteration.; then, the new estimate $\U_{t+1}$ satisfies
\begin{equation}
\dist(\U_{t+1}, \Uo)^2 \leq \alpha \cdot \dist(\U_t, \Uo)^2, \label{conv:eq_01}
\end{equation}
where $\alpha := 1 - \frac{(1 - \delta_{4r}) \cdot \sigma_r(\Xo)}{550((1+\delta_{4r}) \sigma_1(\Xo) + \|e\|_2)} < 1$. 
Further, $\U_{t+1}$ satisfies $ \dist(\U_{t+1}, \Uo) \leq \gamma' \sigma_{r}(\Uo)$, $\forall t$.
\end{theorem}

The above theorem provides a \emph{local} convergence guarantee: 
given an initialization point $\rho_0 = \A_0 \A_0^\dagger$ close enough to the optimal solution --in particular, where $\dist(\U_0, \Uo) \leq \gamma' \sigma_{r}(\Uo)$ is satisfied-- our algorithm converges locally with linear rate.
In particular, in order to obtain $\dist(A_T, \Uo)^2 \leq \varepsilon$, \texttt{ProjFGD} requires $T = \mathcal{O}\left( \log \tfrac{\gamma' \cdot \sigma_r(\Uo)}{\varepsilon} \right)$ number of iterations.
We conjecture that this further translates into linear convergence in the infidelity metric, $1-\tr(\sqrt{\sqrt{\rho_T}\rhoo\sqrt{\rho_T}})^2$.

The per-iteration complexity of \texttt{ProjFGD} is dominated by the application of the linear map $\mathcal{M}$ and by matrix-matrix multiplications.
We note that, while both eigenvalue decomposition and matrix multiplication are known to have $\mathcal{O}((2^n)^2 r)$ complexity in Big-Oh notation, the latter is at least two-orders of magnitude faster than the former on dense matrices \citep{park2016finding}.

The proof of Theorem~\ref{thm:main2} is provided in the Appendix~\ref{sec:proof}. 
We believe that our result, as stated in its most generality, complements recent results from the machine learning and optimization communities, where different assumptions were made~\cite{chen2015fast}, or where constraints on $A$ cannot be accommodated~\cite{bhojanapalli2016dropping}.

So far, we assumed $\rho_0$ is provided such that $\dist(\U_0, \Uo) \leq \gamma' \sigma_{r}(\Uo)$.
The next theorem shows that our initialization could achieve this guarantee (under assumptions) and turn the above local convergence guarantees to convergence to the global minimum.

\begin{lemma}{\label{lem:init}}
Let $\U_0$ be such that $\X_0 = \U_0\U_0^\dagger = \Pi_{\C'}\left(\tfrac{-1}{L} \cdot \nabla f(0) \right)$. 
Consider the problem in \eqref{eq:nonCVX} where $\mathcal{M}$ satisfies the RIP property for some constant $\delta_{4r} \in (0, 1)$.
Further, assume the optimum point $\rhoo$ satisfies $\text{rank}(\rhoo) = r$. 
Then, $\U_0$ computed as above satisfies:
\begin{align*}
\dist(\U_0, \Uo) \leq \gamma' \cdot \sigma_r(\Uo),
\end{align*} where $\gamma' = \sqrt{\tfrac{1 - \tfrac{1- \delta_{4r}}{1 + \delta_{4r}}}{2(\sqrt{2}-1)}} \cdot \tau(\rhoo) \cdot \sqrt{\texttt{srank}(\rhoo)}$ and $\texttt{srank}(\X) = \tfrac{\|\X\|_F}{\sigma_1(\X)}$.
\end{lemma}

This initialization  introduces further restrictions on the condition number of $\rhoo$, $\tau(\rhoo) = \frac{\sigma_1(\rhoo)}{\sigma_r(\rhoo)}$, and the condition number of the objective function, which is proportional to $\propto \frac{1+\delta_{4r}}{1 - \delta_{4r}}$.
In particular, the initialization assumptions in Theorem \ref{thm:main2} are satisfied by Lemma \ref{lem:init} if and only if $\mathcal{M}$ satisfies RIP with constant $\delta_{4r}$ fulfilling the following expression:
\begin{align*}
\frac{1 + \delta_{4r}}{1 - \delta_{4r}} \cdot \sqrt{1 - \tfrac{1 - \delta_{4r}}{1 + \delta_{4r}}} \leq \tfrac{\sqrt{2 (\sqrt{2} - 1)}}{200} \cdot \frac{1}{\sqrt{r} \cdot \tau^2(\rhoo)}.
\end{align*} 
While such conditions are hard to check a priori, our experiments showed that our initialization, as well as the random initialization, work well in practice, and this behavior has been observed repeatedly in all the experiments we conducted.
Thus, the method returns the exact solution of the convex programming problem, while being orders of magnitude faster.

\section{Related work}{\label{sec:related}}
We focus on efficient methods for QST;
for a broader set of citations that go beyond QST, we defer the reader to \citep{park2016finding} and references therein.

The use of non-convex algorithms in QST is not new~\citep{rehacek2007diluted, siah2013informationally}, and dates before the introduction of the CS protocol in QST settings \citep{gross2010quantum}.
Assuming a multinomial distribution, \citep{rehacek2007diluted} focus on the normalized negative log-likelihood objective (see Eq.~(2) in \citep{rehacek2007diluted}) and propose an \emph{diluted} non-convex iterative algorithm for its solution.
The suggested algorithm exhibits good convergence and monotonic increase of the likelihood objective in practice; despite its success, there is no argument that guarantees its performance, neither a provable setup for its execution.

\citep{banaszek1999maximum, paris2001maximum} use the reparameterization $\rho = AA^\dagger$ in a Lagrange augmented maximum log-likelihood (ML) objective (see Eq.~(3) in \citep{banaszek1999maximum} and Eq.~(9) in \citep{paris2001maximum}), under the multinomial distribution assumption.
The authors state that such a problem can be solved by standard numerical procedures for searching the maximum of the ML objective \citep{banaszek1999maximum}, and use the downhill simplex method for its solution, over the parameters of the matrix $A$ \citep{nelder1965simplex}.
Albeit \citep{banaszek1999maximum, paris2001maximum} rely on the uniqueness of the ML solution before the reformulation $\rho = AA^\dagger$ (due to the convexity of the original problem), there are no theoretical results on the non-convex nature of the transformed objective (\emph{e.g.}, the presence of spurious local minima).

\citep{smolin2012efficient} consider the case of maximum likelihood quantum state tomography, under additive Gaussian noise, in the informationally complete case.
Assuming the measurement operators are traceless, simple linear inversion techniques are shown to work accurately to infer the constrained ML state in a single projection step, from the unconstrained ML state. 
As an extension, \citep{hou2016full} present a GPU implementation of the algorithm that recovers a simulated 14-qubit density matrix within four hours; however,  implementations of a linear system inversion could increase dramatically the computational and storage complexity, as the dimension of the problem grows.

Based on the extremal equations for the multinomial ML objective, \citep{siah2013informationally} propose a fixed-point iteration steepest-ascent method on $A$ (with user-defined hyperparameters, such as the step size of the ascent). 
How many iterations required and how to set up initial conditions are heuristically defined. 
Typically these methods, as discussed in~\cite{shang2017superfast}, lead to ill-conditioned optimization problems, resulting in slow convergence.  

\citep{shang2017superfast} propose a hybrid algorithm that $(i)$ starts with a conjugate-gradient (CG) algorithm in the $A$ space, in order to get initial rapid descent, and $(ii)$ switch over to accelerated first-order methods in the original $\rho$ space, provided one can determine the switchover point cheaply.
Under the multinomial ML objective, in the initial CG phase, the Hessian of the objective is computed per iteration (\emph{i.e.}, a $d^2 \times d^2$ matrix), along with its eigenvalue decomposition.
Such an operation is costly, even for moderate values of $d$, and heuristics are proposed for its completion.
In the later phase, the authors exploit ``momentum'' techniques from convex optimization, that lead to provable acceleration when the objective is convex; as we state in the Conclusions section, such acceleration techniques have not considered in the factored space $A$, and constitute an interesting research direction.
From a theoretical perspective, \citep{shang2017superfast} provide no convergence or convergence rate guarantees.

\citep{teo2013informationally} use in practice the general parameterization of density matrices, $\rho = \tfrac{AA^\dagger}{\text{Tr}(AA^\dagger)}$, that ensures jointly positive definiteness and unity in the trace.
There, in order to attain the maximum value of the log-likehoood objective, a steepest ascent method is proposed over $A$ variables, where the step size $\eta$ is an arbitrarily selected but sufficient small parameter.
There is no discussion regarding convergence and convergence rate guarantees, as well as any specific set up of the algorithm (step size, initialization, etc.).

\citep{gonccalves2016projected} study the QST problem in the original parameter space, and propose a projected gradient descent algorithm.
The proposed algorithm applies both in convex and non-convex objectives, and convergence only to stationary points could be expected. 
\citep{bolduc2017projected} extend the work in \citep{gonccalves2016projected} with two first-order variants, using momentum motions, similar to the techniques proposed by Polyak and Nesterov for faster convergence in convex optimization \citep{nesterov1983method}.
The above algorithms operate in the informationally complete case. 
Similar ideas in the informationally incomplete case can be found in \citep{kyrillidis2014matrix, becker2013randomized}.

Very recently, \citep{riofrio2017experimental} presented an experimental implementation of CS tomography of a $n = 7$ qubit system, where only $127$ Pauli basis measurements are available. 
To achieve recovery in practice --within a reasonable time frame over hundreds of problem instances-- the authors proposed a computationally efficient estimator, based on the factorization $\rho = AA^\dagger$.
The resulting method resembles the gradient descent on the factors $A$, as the one presented in this paper.
However, the authors focus only on the experimental efficiency of the method and provide no specific results on the optimization efficiency of the algorithm, what are its theoretical guarantees, and how its components (such as initialization and step size) affect its performance (\emph{e.g.}, the step size is set to a sufficiently small constant).

One of the first provable algorithmic solutions for the QST problem was through convex approximations \cite{recht2010guaranteed}:
this includes nuclear norm minimization approaches \cite{gross2010quantum}, as well as proximal variants, as the one that follows:
\begin{equation} \label{eq:conventional}
\begin{aligned}
	& \underset{\X \succeq 0}{\text{minimize}}
	& & \|\linmap(\X) - \obs\|_F^2 + \lambda \trace(\X).
\end{aligned}
\end{equation}
See also \cite{gross2010quantum} for the theoretical analysis. 
Within this context, we mention the work of \cite{yurtsever2015universal}: there,  the \texttt{AccUniPDGrad} algorithm is proposed --a universal primal-dual convex framework with sharp operators, in lieu of proximal low-rank operators-- where QST is considered as an application. 
\texttt{AccUniPDGrad} combines the flexibility of proximal primal-dual methods with the computational advantages of conditional gradient (Frank-Wolfe-like) methods.
We will use this algorithm for comparisons in the experimental section. 

\cite{hazan2008sparse} presents \texttt{SparseApproxSDP} algorithm that solves the QST problem in \eqref{eq:CVX2},
when the objective is a generic gradient Lipschitz smooth function, by updating a putative low-rank solution with rank-1 refinements, coming from the gradient. 
This way, \texttt{SparseApproxSDP} avoids computationally expensive operations per iteration, such as full eigen-decompositions.
In theory, at the $r$-th iteration, \texttt{SparseApproxSDP}
is guaranteed to compute a $\tfrac{1}{r}$-approximate solution, with rank at most $r$, \textit{i.e.}, achieves a sublinear 
$O\left(\tfrac{1}{\varepsilon}\right)$ convergence rate. 
However, depending on $\varepsilon$, \texttt{SparseApproxSDP} might not return a low rank solution. 

Finally, \cite{becker2013randomized} propose Randomized Singular Value Projection (\texttt{RSVP}), a projected gradient descent algorithm for QST, which merges gradient calculations with truncated eigen-decompositions, via randomized approximations for computational efficiency. 

Overall, our program is tailored for tomography of highly-pure quantum states, by incorporating this constraint into the structure of $A$.  
This has two advantages. 
First, it results in a faster algorithm that enables us to deal many-qubit state reconstruction in a reasonable time; and, second, it allows us prove the accuracy of the \texttt{ProjFGD} estimator under model errors and experimental noise, similar to the CS results.

\section{Numerical experiments}{\label{sec:experiments}}

\begin{table*}[!ht]
\centering
\begin{tabular}{c c c c c c c c c c c c c c c c c c c c c}
  \toprule
& & \multicolumn{9}{c}{$d = 2^7$} & & \multicolumn{9}{c}{$d = 2^{13}$} \\ 
\cmidrule{3-11} \cmidrule{13-21}
  & & \multicolumn{3}{c}{$\sigma = 0$} & & \multicolumn{5}{c}{$\sigma = 0.05$} & & \multicolumn{3}{c}{$\sigma = 0$} & & \multicolumn{5}{c}{$\sigma = 0.05$} \\
  \cmidrule{3-5} \cmidrule{7-11} \cmidrule{13-15} \cmidrule{17-21}
  Algorithm   & & Time [s]& & $\tfrac{\|\widehat{\rho} - \rhoo\|_F}{\|\rhoo\|_F}$ & & Time [s] & & $\tfrac{\|\widehat{\rho} - \rhoo\|_F}{\|\rhoo\|_F}$ & & Infidelity & & Time [s] & & $\tfrac{\|\widehat{\rho} - \rhoo\|_F}{\|\rhoo\|_F}$  & & Time [s] & & $\tfrac{\|\widehat{\rho} - \rhoo\|_F}{\|\rhoo\|_F}$ & & Infidelity  \\
  \cmidrule{1-1}   \cmidrule{3-3} \cmidrule{5-5} \cmidrule{7-7} \cmidrule{9-9} \cmidrule{11-11} \cmidrule{13-13} \cmidrule{15-15} \cmidrule{17-17} \cmidrule{19-19} \cmidrule{21-21}
  \eqref{eq:CVX1}  & & 46.01 & &  5.3538e-07 & &  58.48 & &  6.0405e-02 & &  3.0394e-02  & &  - & &  - & &  - & &  - & & -\\ 
  \eqref{eq:CVX2} & & 77.12 & &  3.0645e-04  & &  65.53 & &  6.1407e-02 & &  3.0559e-02  & &  - & &  - & &  - & &  - & & -\\ 
  \texttt{ProjFGD} & & 0.28 & &  3.2224e-08  & &  0.30 & &  2.3540e-02 & &  1.3820e-04  & &  1314.01 & &  6.8469e-08  && 1487.22 && 3.1104e-02 && 1.9831e-03\\
  \bottomrule
\end{tabular}
\caption{All values are \emph{median} values over 10 independent Monte Carlo iterations.} \label{table1} 
\end{table*}

We conducted experiments in an Matlab environment, installed in a Linux-based system with 256 GB of RAM, and equipped with two Intel Xeon E5-2699 v3 2.3GHz, 45M Cache, 9.60GT/s. 
In all the experiments, the error reported in the Frobenius metric,  $\|\widehat{\rho} - \rhoo\|_F / \|\rhoo\|_F$, where $\widehat{\rho}$ is the estimation of the true state $\rhoo$. 
Note that for a pure state $\rho$, $\|\rho\|_F=1$. For some experiments we also report the infidelity metric $1 - \text{Tr}\left(\sqrt{\sqrt{\rhoo} \widehat{\rho} \sqrt{\rhoo}}\right)^2$. 
We will also use $\DM_{d}$ to denote the set of $d\times d$ density matrices $\DM_{d}=\{\varrho: \varrho\in \mathbb{C}^{d \times d}, \varrho \succeq 0, \tr(\varrho)=1 \}$.

\subsection{Comparison of \texttt{ProjFGD} with second-order methods}

As a first set of experiments, we compare the efficiency of \texttt{ProjFGD} with \emph{second-order} cone convex programs.
State of the art solvers within this class of solvers are the SeDuMi \citep{sturm1999using} and SDPT3 \citep{tutuncu2003solving} methods; for their use, we rely on the off-the-shelf Matlab wrapper \texttt{CVX} \citep{cvx}.
In our experiments, we observed that SDPT3 was faster and we select it for our comparison.

The setting is as described in Section \ref{sec:QST}: 
we consider rank-1 normalized density matrices $\rhoo \in \DM_{2^n}$, from we which we obtain Pauli measurements such that 
$y_i = \tfrac{2^n}{\sqrt{m}}\text{Tr}(P_i \cdot \rho_\star)+e_i,~i = 1, \dots, m$, for some i.i.d. Gaussian measurement error $e_i$, with variance $\sigma$, \emph{i.e.}, $\sim \mathcal{CN}(0, \sigma \cdot I)$.
We consider both convex formulations \eqref{eq:CVX1}-\eqref{eq:CVX2} and compare it to the  \texttt{ProjFGD} estimator with $r=1$; in figures we use the notation \texttt{CVX 1} and \texttt{CVX 2} for simplicity.

We consider two cases: $(i)$ $n = 7$, and $(ii)$ $n = 13$.
Table \ref{table1} shows median values of 10 independent experimental realizations for $m = \tfrac{7}{3} rd \log d$; this selection of $m$ was made so that all algorithms return a solution close to the optimum $\rhoo$.
Empirically, we have observed that $\texttt{ProjFGD}$ succeeds even for cases $m = \mathcal{O}(rd)$.
We consider both noiseless $\sigma = 0$ and noisy $\sigma = 0.05$ settings.

In order to accelerate the execution of convex programs, we set the solvers in \texttt{CVX} to low precision.
From Table \ref{table1}, we observe that our method is two orders of magnitude faster than second-order methods for $n = 7$: \texttt{ProjFGD} achieves better performance (in both error metrics, and in both noisy/noiseless cases), faster. 
For the higher $n = 13$ qubit case, we could not complete the experiments for \eqref{eq:CVX1}-\eqref{eq:CVX2} due to system crash (RAM overflow). 
Contrariwise, our method was able to complete the task with success within about $22$ minutes of CPU time.

Figures \ref{fig:time_vs_measurements}-\ref{fig:time_vs_dimensions} show graphically how second-order convex vs. our first-order non-convex schemes scale, as a function of time.
In Figure \ref{fig:time_vs_measurements}, we fix the dimension to $d = 2^7$ and study how increasing the number of observations $m$ affects the performance of the algorithms.
We observe that, while in the \texttt{ProjFGD}, more observations lead to faster convergence \citep{chandrasekaran2013computational}, the same does not hold for the second-order cone programs.
In Figure \ref{fig:time_vs_dimensions}, we fix the number data points to $m = \tfrac{7}{3}rd \log d$, and we scale the dimension $d$. 
It is obvious that the convex solvers do not scale easily beyond $n = 7$, whereas our method handles cases up to $n = 13$, within reasonable time.

\begin{figure}[!h]
\centering
\includegraphics[width=1\columnwidth]{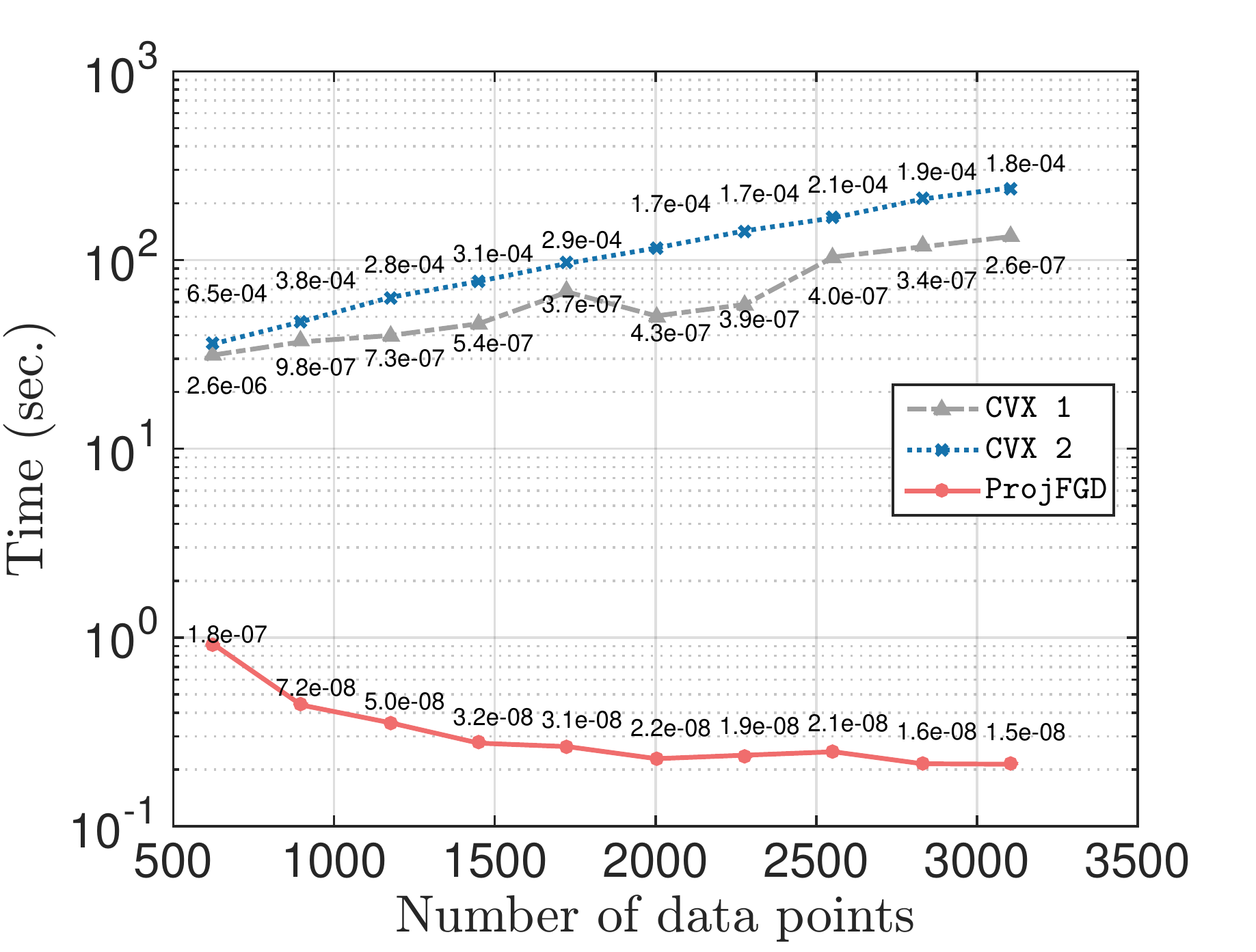}
\caption{Dimension fixed to $d = 2^7$ with $\texttt{rank}(\rhoo) = 1$. The figure depicts the noiseless setting, $\sigma=0$. Numbers within figure are the error in Frobenius norm achieved (median values). }\label{fig:time_vs_measurements}
\end{figure}

\begin{figure}[!h]
\centering
\includegraphics[width=1\columnwidth]{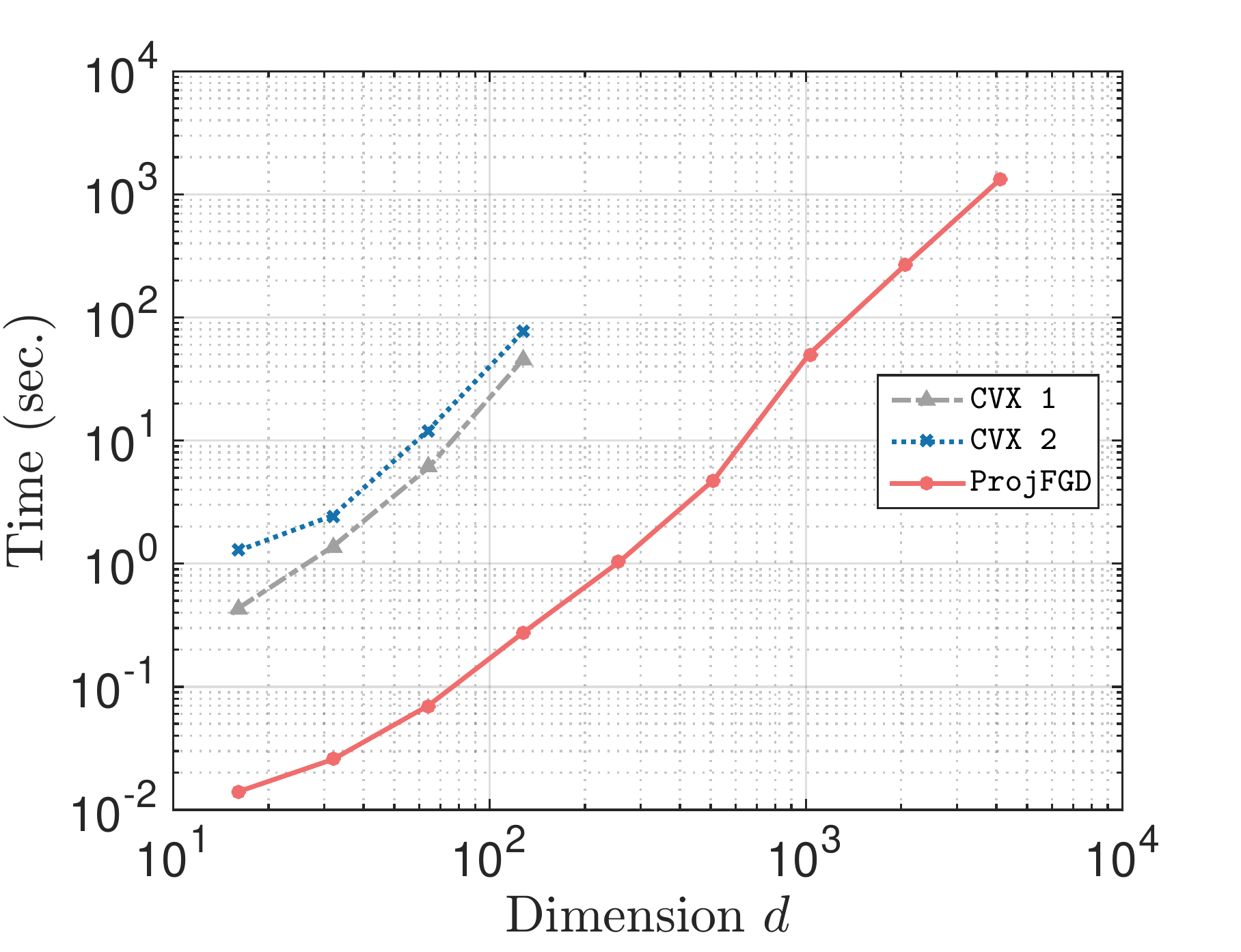}
\caption{Number of data points set to $m = \tfrac{7}{3} r d \log d$. Rank of optimum point is set to $\texttt{rank}(\rhoo) = 1$. The figure depicts the noiseless setting. }\label{fig:time_vs_dimensions}
\end{figure}

\subsection{Comparison of \texttt{ProjFGD} with first-order methods}

\begin{figure*}[t!]
\centering
\includegraphics[width=0.33\textwidth]{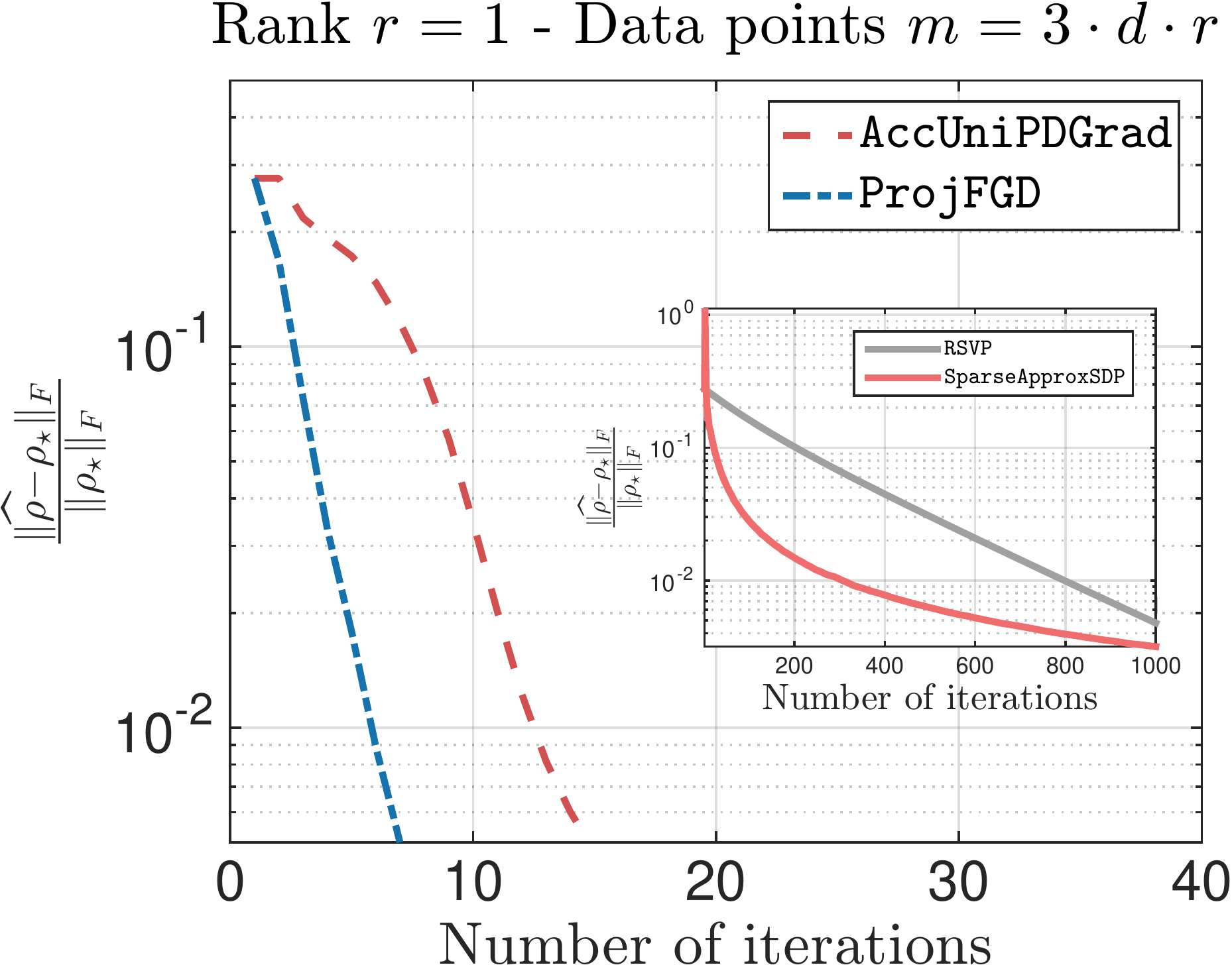} \hspace{-0.2cm}
\includegraphics[width=0.33\textwidth]{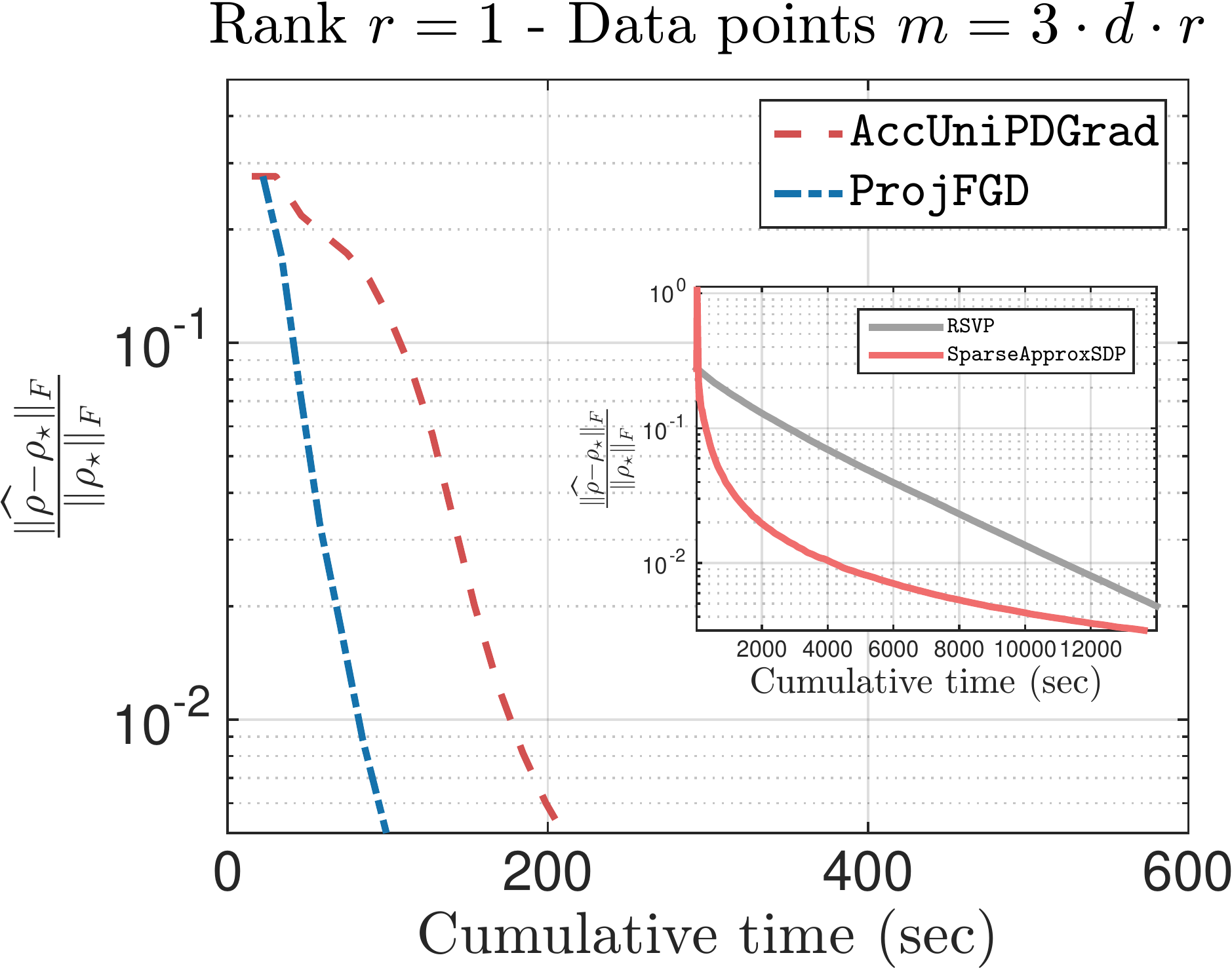} \hspace{-0.2cm}
\includegraphics[width=0.33\textwidth]{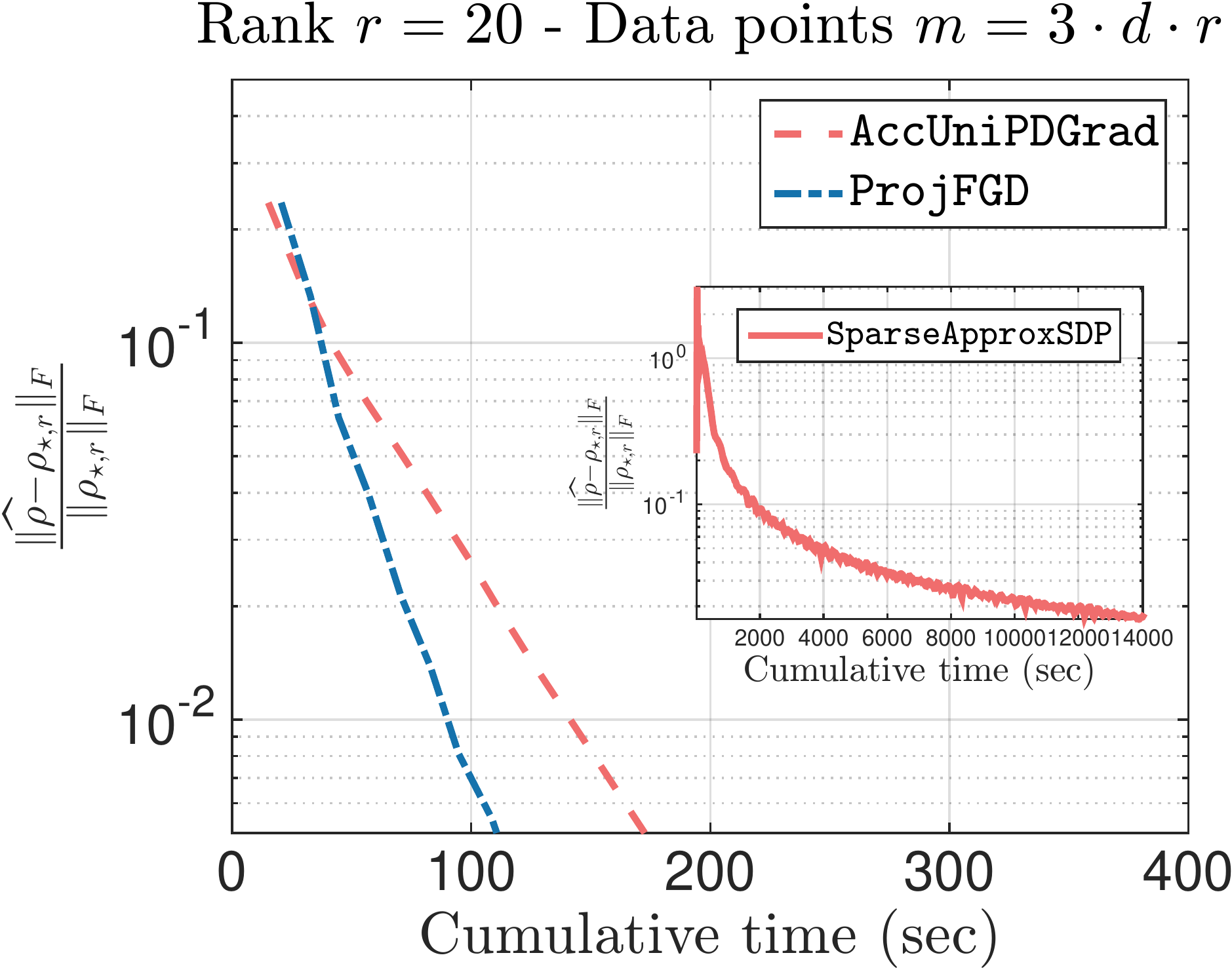} 
\caption{\small{\textbf{Left and middle panels}: Convergence performance of algorithms under comparison w.r.t. error in Frobenius norm vs. $(i)$ the total number of iterations (left) and $(ii)$ the total execution time. Both cases correspond to $C_{\rm sam} = 3$, $r = 1$ (pure state) and $n = 12$ (\textit{i.e.}, $d = 4096$). \textbf{Right panel}: Nearly low-rank state case---we approximate $\rhoo$ with $\rho_{\star, r}$, with $r = 20$. In this setting, $n = 12$ (\textit{i.e.}, $d = 4096$) and $C_{\rm sam} = 3$. }
}
\label{fig:exp1}
\end{figure*}

Here, we compare our method with more efficient first-order methods, both convex (\texttt{AccUniPDGrad}~\cite{yurtsever2015universal}) and non-convex (\texttt{SparseApproxSDP}~\cite{hazan2008sparse} and \texttt{RSVP}~\cite{becker2013randomized}). 

We consider two settings: $\rhoo \in \DM_{d}$ is $(i)$ a pure state (\emph{i.e.}, $\text{rank}(\rhoo) = 1$) and, 
$(ii)$ a nearly low-rank state.
In the latter case, we construct $\rhoo = \rho_{\star, r} + \zeta$, where $\rho_{\star, r}$ is a rank-deficient PSD satisfying $\text{rank}(\rho_{\star, r}) = r$, and $\zeta \in \mathbb{C}^{d \times d}$ is a full-rank PSD noise term with a fast decaying eigen-spectrum, significantly smaller than the leading eigenvalues of $\rho_{\star, r}$.
In other words, we can well-approximate $\rhoo$ with $\rho_{\star, r}$.
For all cases, we model the measurement vector as $\obs = \linmap(\rhoo) + e$; here, the noise is such that $\|e\| = 10^{-3}$.
The number of data points $m$ satisfy $m = C_{\rm sam} \cdot r d $, for various values of $C_{\rm sam} > 0$.

For all algorithms, we assumed $r=\text{rank}(\rho_{\star, r})$ is known and use it to reconstruct a rank-$r$ approximation of $\rhoo$. 
All methods that require an SVD routine use \texttt{lansvd}$(\cdot)$ from the \texttt{PROPACK} software package.   
Experiments and algorithms are implemented in a \textsc{Matlab} environment; we used non-specialized and non-\texttt{mex}ified code parts for all algorithms. 
For initialization, we use the same starting point for all algorithms, which is either specific (Section \ref{subsec:initialization}) or random. 
As a stopping criterion, we use $\tfrac{\|\rho_{t+1} - \rho_t\|_F}{\|\rho_{t+1}\|_F} \leq \texttt{tol}$; we set the tolerance parameter to $\texttt{tol} := 5\cdot 10^{-6}$.

\emph{Convergence plots.} 
Figure \ref{fig:exp1} (two-leftmost plots) illustrates the iteration and timing complexities of each algorithm under comparison, for a pure state recovery setting ($r = 1$) of a highly-pure $\rhoo$. 
Here, $n = 12$ which corresponds to a ${d^2} = 16,777,216$ dimensional problem; moreover, we assume $C_{\rm sam} = 3$ and thus the number of data points are $m = 12,288$. 
For initialization, we use the proposed initialization in Section \ref{subsec:initialization} for all algorithms: we compute $-\mathcal{M}^\dagger(y)$, extract factor $\U_0$ as the best-$r$ PSD approximation of $-\mathcal{M}^\dagger(y)$, and project $\U_0$ onto $\C$. 

It is apparent that \texttt{ProjFGD} converges faster to a vicinity of $\rhoo$, compared to the rest of the algorithms; observe also the sublinear rate of \texttt{SparseApproxSDP} in the inner plots, as reported in \cite{hazan2008sparse}.

Table \ref{tbl:small_Comp} contains recovery error and execution time results for the case $n = 13$ ($d = 8096$); in this case, we solve a $d^2 = 67,108,864$ dimensional problem. 
For this case, \texttt{RSVP} and \texttt{SparseApproxSDP} algorithms were excluded from the comparison, due to excessive execution time.
Appendix~\ref{sec:add_exp} provides extensive results, where similar performance is observed for other values of $d=2^n$ and $C_{\rm sam}$.

\begin{table}[!ht]
	\centering
		\begin{tabular}{c c c c c} \toprule
			& \phantom{a} & \multicolumn{3}{c}{} \\
			\cmidrule {3-5}
			Algorithm & \phantom{a} & $\tfrac{\|\widehat{\rho} - \rhoo \|_F}{\|\rhoo\|_F}$ & \phantom{a}  & Time [s]\\
			\cmidrule{1-1} \cmidrule {3-3} \cmidrule{5-5} 
			\texttt{AccUniPDGrad} & & 7.4151e-02 & & 2354.4552 \\ 
			\texttt{ProjFGD} & & 8.6309e-03 & & 1214.0654 \\ 
			\bottomrule
		\end{tabular}
	\caption{\small{Comparison results for reconstruction and efficiency, for $n = 13$ qubits and  $C_{\rm sam} = 3$.}} \label{tbl:small_Comp}
\end{table}

Figure \ref{fig:exp1} (rightmost plot) considers the more general case where $\rhoo$ is nearly low-rank: \emph{i.e.}, it can be well-approximated by a density matrix $\rho_{\star, r}$ where $r=20$ (low-rank density matrix).
In this case, $n = 12$, $m = 245,760$ for $C_{\rm sam} = 3$. 
As the rank in the model, $r$, increases, algorithms that utilize an SVD routine spend more CPU time on singular value/vector calculations. 
Certainly, the same applies for matrix-matrix multiplications; however, in the latter case, the complexity scale is milder than that of the SVD calculations.  Further metadata are also provided in Table \ref{tbl:small_Comp2}. 

\begin{table}[!ht]
	\centering
		\begin{tabular}{c c c c c} \toprule
		& \multicolumn{2}{c}{Setting: $r = 5$.}  & \multicolumn{2}{c}{Setting: $r = 20$.} \\
		\cmidrule {2-3} \cmidrule{4-5}
		Algorithm & $\tfrac{\|\widehat{\rho} - \rho_{\star, r} \|_F}{\|\rho_{\star, r}\|_F}$ & Time [s] & $\tfrac{\|\widehat{\rho} - \rho_{\star, r} \|_F}{\|\rho_{\star, r}\|_F}$ & Time [s] \\
		\midrule
     \texttt{SparseApproxSDP} & 3.17e-02 & 3.74 & 5.49e-02 & 4.38 \\ 
		\texttt{RSVP} & 5.15e-02 &  0.78 & 1.71e-02 & 0.38  \\ 
		\texttt{AccUniPDGrad} & 2.01e-02 & 0.36 & 1.54e-02 & 0.33\\ 
		\texttt{ProjFGD} & 1.20e-02 & 0.06 & 7.12e-03 & 0.04 \\ 
		\bottomrule
	\end{tabular}
	\caption{\small{Results for reconstruction and efficiency. Time reported is in seconds. For all cases, $C_{\rm sam} = 3$ and $n = 10$.}} \label{tbl:small_Comp2}
\end{table}
For completeness, in Appendix~\ref{sec:add_exp} we provide results that illustrate the effect of random initialization:
Similar to above, \texttt{ProjFGD} shows competitive behavior by finding a better solution faster, irrespective of initialization point. 

\emph{Timing evaluation (total and per iteration)}. 
Figure \ref{fig:exp2} highlights the efficiency of our algorithm in terms of time complexity, for various problem configurations.
Our algorithm has fairly low per iteration complexity (where the most expensive operation for this problem is matrix-matrix and matrix-vector multiplications). 
Since our algorithm shows also fast convergence in terms of the number of iterations, this overall results into faster convergence towards a good approximation of $\rhoo$, even as the dimension increases. 
Figure \ref{fig:exp2}  shows how the total execution time scales with parameters $n$ and $r$. 
\begin{figure}[t!]
\centering
\includegraphics[width=1\columnwidth]{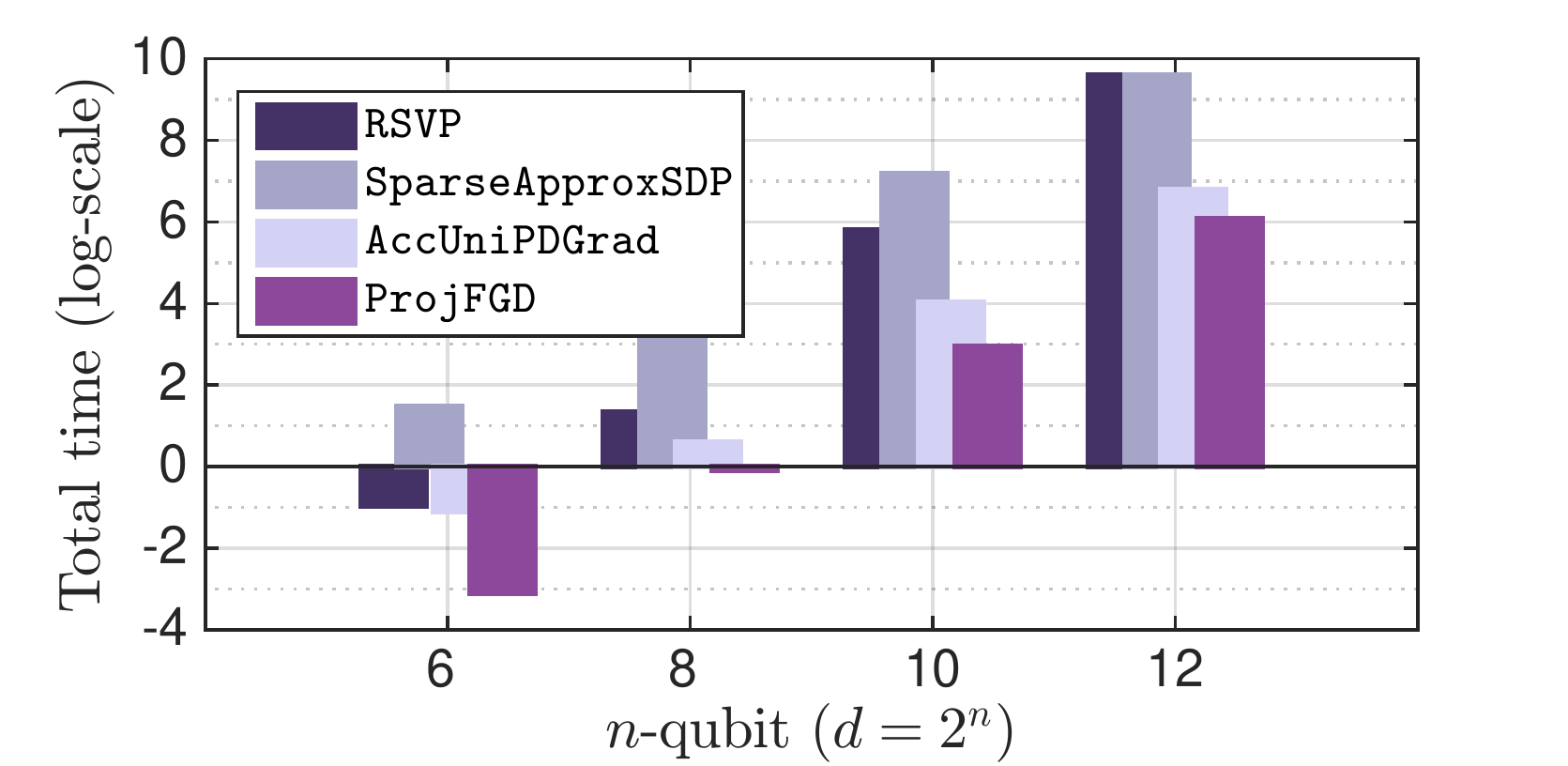} 
\includegraphics[width=1\columnwidth]{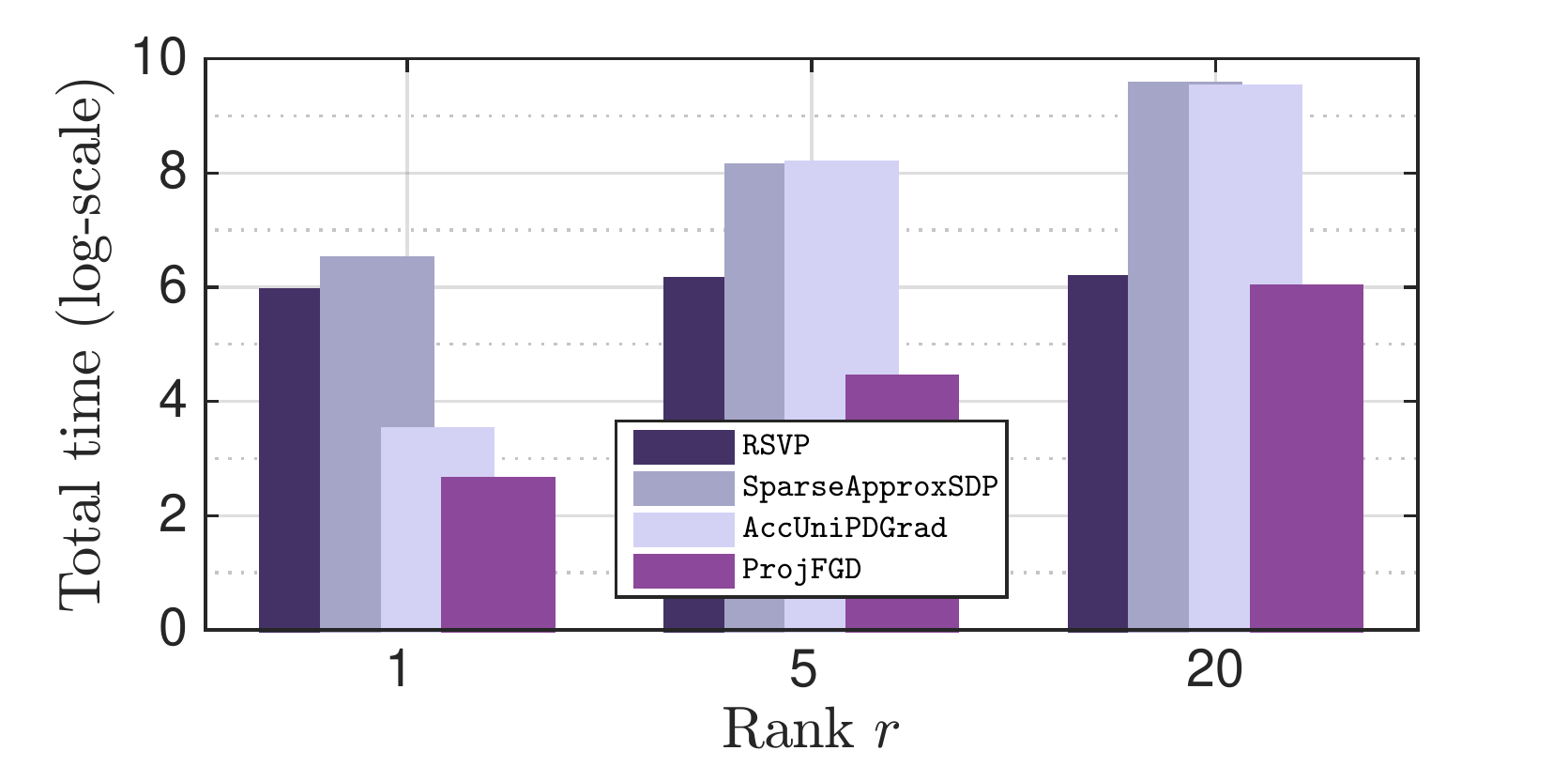}
\caption{\small{Timing bar plot.$y$-axis shows total execution time in ($\log_{10}$-scale) while $x$-axis corresponds to different $n$ values. Top panel corresponds to $r = 1$ and $C_{\rm sam} = 6$; bottom panel corresponds to $n = 10$ and $C_{\rm sam} = 3$. All cases are noiseless.}
}
\label{fig:exp2}
\end{figure}

\emph{Overall performance}. 
\texttt{ProjFGD} shows a substantial improvement in performance, as compared to the state-of-the-art algorithms; we would like to emphasize also that projected gradient descent schemes, such as in \cite{becker2013randomized}, are also efficient in small- to medium-sized problems, due to their
fast convergence rate. 
Further, convex approaches might show better sampling complexity performance (\emph{i.e.}, as $C_{\rm sam}$ decreases).
Nevertheless, one can perform accurate MLE reconstruction for larger systems in the same amount of time using our methods for such small- to medium-sized problems.
We defer the reader to Appendix~\ref{sec:add_exp}, due to space restrictions.

\section{Summary and conclusions}

In this work, we propose a non-convex algorithm, dubbed as \texttt{ProjFGD}, for estimating a highly-pure quantum state, in a high-dimensional Hilbert space, from relatively small number of data points. 
We showed empirically that \texttt{ProjFGD} is orders of magnitude faster than state-of-the-art convex and non-convex programs, such as \citep{yurtsever2015universal},\citep{hazan2008sparse}, and \citep{becker2013randomized}. 
More importantly, we prove that under proper initialization and step-size, the \texttt{ProjFGD} is guaranteed to converge to the global minimum of the problem, thus  ensuring a provable tomography procedure; see Theorem~\ref{thm:main2} and Lemma~\ref{lem:init}.

In our setting, we model the state as a low-rank PSD matrix. 
This, in turn, means that the estimator is biased towards low-rank states. 
However, such bias is inherent to all CS-like QST protocols by the imposition of the positivity constraint \citep{kalev2015quantum}.

Our techniques and proofs can be applied --in some cases under proper modifications-- to scenaria beyond the ones considered in this work. 
We restricted our discussions to a measurement model of random Pauli observables, that satisfies RIP.
We conjecture that our results apply for other ``sensing'' settings, that are informationally complete for low-rank states; see \emph{e.g.}, \citep{baldwin2016strictly}. 
The results presented here are independent of the noise model and could be applied for non-Gaussian noise models, such as those stemming from finite counting statistics.  
Lastly, while here we focus on state tomography, it would be interesting to explore similar techniques for the problem of process tomography.
 
We conclude with a short list of interesting future research directions. 
Our immediate goal is the application of \texttt{ProjFGD} in real-world scenaria; this could be completed by utilizing the infrastructure at the IBM T.J. Watson Research Center \citep{IBMQ}.
This could complement the results found in \citep{riofrio2017experimental} for a different quantum system.

Beyond this practical implementation, we identify the following interesting open questions.
First, the ML estimator is one of the most frequently-used methods for QST experiments. 
Beyond its use as point estimator, it is also used as a basis for inference around the point estimate, via confidence intervals~\citep{christandl2012reliable} and credible regions~\citep{shang13optimal}. 
However, there is still no rigorous analysis when the factorization $\rho = AA^\dagger$ is used. 

The work of \citep{shang2017superfast} considers accelerated gradient descent methods for QST in the original parameter space $\rho$: 
Based on the seminal work of Polyak and Nesterov \citep{nesterov1983method} on convex optimization first-order methods, one can achieve orders of magnitude acceleration (both in theory and practice), by exploiting the \emph{momentum} from previous iterates.
It remains an open question how our approach could exploit acceleration techniques that lead to faster convergence in practice, along with rigorous approximation and convergence guarantees. 
Further, distributed/parallel implementations, like the one in \citep{hou2016full}, remain widely open using our approach, in order to accelerate further the execution of the algorithm.
Research along these directions is very interesting and is left for future work.

Finally, while we saw numerically that a random initialization under noisy and constrained settings works well, a careful theoretical treatment for this case is an open problem.

\begin{acknowledgments}
Anastasios Kyrillidis is supported by the IBM Goldstine Fellowship and Amir Kalev is supported by the Department of Defense.
\end{acknowledgments}

\appendix
\small
%!TEX root = QST_101417.tex
\section{Theory}{\label{sec:proof}}

\noindent \textbf{Notation.} For matrices $\rho, \zeta \in \mathbb{C}^{d \times d}$, $\ip{\rho}{\zeta} = \trace\left(\rho^\dagger \zeta \right)$ represents their inner product. 
We use $\norm{\rho}_F$ and $\sigma_1(\rho)$ for the Frobenius and spectral norms of a matrix, respectively. 
We denote as $\sigma_i(\rho)$ the $i$-th singular value of $\rho$.
$\rho_r$ denotes the best rank-$r$ approximation of $\rho$. 

\subsection{Problem generalization, notation and definitions}
To expand the generality of our scheme, we re-state the problem setting for a broader set of objectives.
The following arguments hold for both real and complex matrices. 
We consider criteria of the following form:
\begin{equation}{\label{eq:appe_00}}
\begin{aligned}
	& \underset{\rho \in \mathbb{C}^{d \times d}}{\text{minimize}}
	& & f(\rho) \quad \quad \text{subject to} \quad \rho \succeq 0, ~\rho \in \C'.
\end{aligned}
\end{equation} 
To make connection with the QST objective, set $f(\rho) = \tfrac{1}{2} \cdot \|y - \mathcal{M}(\rho)\|_2^2$, and $\rho \in \C' ~\Leftrightarrow~ \trace(\rho) \leq 1$.
Apart from the least-squares objective in QST, our theory extends to applications that can be described by \emph{strongly} convex functions $f$ with \emph{gradient Lipschitz continuity}.
Further, our ideas can be applied in a similar fashion to the case of restricted smoothness and restricted strong convexity \cite{agarwal2010fast}. 
We state these standard definitions below for the square case.

\begin{definition}{\label{prelim:def_00}}
Let $f: \mathbb{C}^{d \times d} \rightarrow \R$ be convex and differentiable. 
$f$ is restricted $\mu$-strongly convex if for all rank-$r$ matrices $\rho, \zeta \in  \mathbb{C}^{d \times d}$, 
\begin{equation}\label{eq:sc}
f(\zeta) \geq f(\rho) + \ip{\gradf\left(\rho\right)}{\zeta - \rho} + \tfrac{\mu}{2} \norm{\zeta - \rho}_F^2.
\end{equation}
\end{definition}

\begin{definition}{\label{prelim:def_01}}
Let $f: \mathbb{C}^{d \times d} \rightarrow \R$ be a convex differentiable function. 
$f$ is restricted gradient Lipschitz continuous with parameter $L$ (or $L$-smooth) if for all rank-$r$ matrices $\rho, \zeta \in  \mathbb{C}^{d \times d}$, 
\begin{equation}
\norm{\gradf\left(\rho\right) - \gradf\left(\zeta\right)}_F \leq L \cdot \norm{\rho - \zeta}_F.
\end{equation}
\end{definition} 

To shed some light on the notions of (restricted) strong convexity and smoothness and how they relate to the QST objective, consider the \emph{restricted isometry property}, which holds with high probability under Pauli measurements for low rank $\rhoo$ \cite{candes2011tight, liu2011universal}; here, we present a simplified version of the definition in the main text:
\begin{definition}[Restricted Isometry Property (RIP)]\label{def:RIP}
A linear map $\mathcal{M}$ satisfies the $r$-RIP with constant $\delta_r$, if
\begin{align*}
(1 - \delta_r)\|\rho\|_F^2 \leq \|\mathcal{M}(\rho)\|_2^2 \leq (1 + \delta_r)\|\rho\|_F^2,
\end{align*} is satisfied for all matrices $\rho \in \mathbb{C}^{d \times d}$ such that ${\rm \text{rank}}(\rho) \leq r$. 
\end{definition} 

According to the quadratic loss function in QST:
\begin{align*}
f(\rho) = \tfrac{1}{2} \|y - \mathcal{M}(\rho)\|_2^2,
\end{align*} 
its Hessian is given by $\mathcal{M}^\dagger\mathcal{M}(\cdot)$. 
Then, (restricted) strong convexity suggests that: 
\begin{align*}
\|\mathcal{M}(\rho)\|_2^2\geq C \cdot \|\rho\|_F^2, \quad \rho \in \mathbb{C}^{d \times d},
\end{align*} 
for a restricted set of directions $\rho$, where $C > 0 $ is a small constant. 
Then, the correspondence of restricted strong convexity and smoothness with the RIP is obvious: both lower and upper bound the quantity $\|\mathcal{M}(\rho)\|_2^2$, where $\rho$ is drawn from a restricted (low-rank) set. 
It turns out that linear maps that satisfy the RIP for low rank matrices, also satisfy the restricted strong convexity; see Theorem 2 in \cite{chen2010general}.

By assuming RIP, the condition number $\kappa$ of $f$ depends on the RIP constants of the linear map $\mathcal{M}$; in particular, one can show that $\kappa = \tfrac{L}{\mu} \propto \tfrac{1 + \delta}{1-\delta}$, since the eigenvalues of $\mathcal{M}^\dagger\mathcal{M}$ lie between $1-\delta$ and $1 + \delta$, when restricted to low-rank matrices. 
Observe that for $\delta$ sufficiently small and dimension $d$ sufficiently large, $\kappa \approx 1$, with high probability.

We assume the optimum $\rho_\star$ satisfies $\text{rank}(\rho_\star) = r_\star$.
For our analysis, we further assume we know $r_\star$ and set $r_\star \equiv r$. 

As suggested in the main text, we solve \eqref{eq:appe_00} in the factored space, as follows:
\begin{equation}{\label{eq:appe_01}}
\begin{aligned}
	& \underset{\A \in \mathbb{C}^{d \times r}}{\text{minimize}}
	& & f(\A\A^\dagger) \quad \quad \text{subject to} \quad A \in \C.
\end{aligned}
\end{equation} 
In the QST setting,  $\A \in \C ~\Leftrightarrow~ \|\A\|_F^2 \leq 1$.

In our theory we mostly focus on sets $\C$ that satisfy the following assumptions.
\begin{assumption}{\label{ass:00}}
For $\rho \succeq 0$, there is $\A \in \mathbb{C}^{d \times r}$ and $r \leq d$ such that $\rho = \A\A^\dagger$. 
Then, $\C' \subseteq \mathbb{C}^{d \times d}$ is endowed with a constraint set $\C \subseteq \mathbb{C}^{d \times r}$ that $(i)$ for each $\rho \in \C'$, there is an subset in $\C$ where each $\A \in \C$ satisfies $\rho = \A\A^\dagger$ 
and $(ii)$ its projection operator, say $\Pi_{\C}(B) = \argmin_{\A \in \C} \tfrac{1}{2} \|\A - B\|_F^2$ for $B \in \mathbb{C}^{d \times r}$, is an entrywise scaling operation on the input $B$.
\end{assumption}

We also require the following \emph{faithfulness} assumption  \cite{chen2015fast}:
\begin{assumption}{\label{prelim:def_02}}
Let $\E$ denote the set of equivalent factorizations that lead to a rank-$r$ matrix $\rhoo \in \mathbb{C}^{d \times d}$; 
i.e.,
$\E := \left\{ \Ao \in \mathbb{C}^{d \times r}~:~ \rhoo = \Ao \Ao^\dagger \right\}.$
Then, we assume $\E \subseteq \C$, i.e., the resulting convex set $\C$ in \eqref{eq:appe_01} (from $\C'$ in \eqref{eq:appe_00}) \emph{respects the structure of $\E$}.
\end{assumption}

Summarizing, by faithfulness of $\C$ (Assumption \ref{prelim:def_02}), we assume that $\E \subseteq \C$. 
This means that the feasible set $\C$ in \eqref{eq:appe_01} contains all matrices $\Ao$ that lead to $\rhoo = \Ao \Ao^\dagger $ in \eqref{eq:appe_00}.
Moreover, we assume both $\C, ~\C'$ are convex sets and there exists a ``mapping" from $\C'$ to $\C$, such that the two constraints are ``equivalent": \emph{i.e.}, $\forall \A \in \C$, we are guaranteed that $\rho = \A\A^\dagger \in \C'$. 
We restrict our discussion on norm-based sets for $\C$	such that Assumption \ref{ass:00} is satisfied.
As a representative example, consider the QST case where, for any $\rho = \A\A^\dagger$, $\trace(\rho) \leq 1 \Leftrightarrow \|\A\|_F^2 \leq 1$. 

For our analysis, we will use the following step sizes: 
\begin{small}
\begin{align*}
\widehat{\eta} &= \frac{1}{128 \left(L \sigma_1(\rho_t) + \sigma_1\left(Q_{\A_t} Q_{\A_t}^\dagger\gradf(\rho_t)\right)\right)}, \\
\eta_\star &= \frac{1}{128 \left( L \sigma_1(\rhoo) + \sigma_1(\gradf(\rhoo))\right)}.
\end{align*}
\end{small}
Here, $L$ is the Lipschitz constant in \eqref{eq:lipschitz} and $Q_{\A_t}Q_{\A_t}^\dagger$ represents the projection onto the column space of $\A_t$. 
In our algorithm, as described in the main text, we use the following step size:
\begin{align*}
\eta \leq \frac{1}{128\left(L \sigma_1(\rho_0) + \sigma_1(\gradf(\rho_0) )\right)}  \quad \text{for given initial } \rho_0.
\end{align*}
While different, by Lemma A.5 in \citep{bhojanapalli2016dropping}, we know that $\widehat{\eta} \geq \tfrac{5}{6} \eta$ and $\tfrac{10}{11}\eta_\star \leq \eta \leq \tfrac{11}{10} \eta_\star$. 
Thus, in our proof, we will work with step size $\widehat{\eta}$, which is equivalent --up to constants-- to the original step size $\eta$ in the proposed algorithm. 

For ease of exposition, we re-define the sequence of updates:
$\A_t $ is the current estimate in the factored space, 
$\widetilde{\A}_{t+1} = \A_t - \widehat{\eta} \gradf(\rho_t)\A_t$ is the putative solution after the gradient step (observe that $\widetilde{\A}_{t+1}$ might not belong in $\C$),
and $\A_{t+1} = \Pi_{\C}(\widetilde{A}_{t+1})$ is the projection step onto $\C$. 
Observe that for the constraint cases we consider in this paper, $\A_{t+1} = \Pi_{\C}(\widetilde{\A}_{t+1}) = \xi_t(\widetilde{\A}_{t+1}) \cdot \widetilde{A}_{t+1}$, where $\xi_t(\cdot) \in (0, 1)$; in the case $\xi_t(\cdot) = 1$, the algorithm simplifies to the algorithm in \citep{bhojanapalli2016dropping}.
For simplicity, we drop the subscript and the parenthesis of the $\xi$ parameter; these values are apparent from the context.

An important issue in optimizing $f$ over the factored space is the existence of non-unique possible factorizations. 
We use the following rotation invariant distance metric:
\begin{definition}{\label{prelim:def_04}}
Let matrices $\A, \B \in \mathbb{C}^{d \times r}$. Define:
\begin{align*}
\dist\left(A, B\right) :=\min_{R: R \in \mathcal{U}} \norm{A - B R}_F, 
\end{align*} where $\mathcal{U}$ is the set of $r \times r$ unitary matrices $R$.
\end{definition}

%\subsection{Assumptions}
We assume that \texttt{ProjFGD} is initialized with a ``good'' starting point $\rho_0 = \A_0 \A_0^\dagger$, such that:
\begin{itemize}
\item [$(A1)$] \quad $\A_0 \in \C$ \quad and \quad $ \dist(\A_0, \Ao) \leq \gamma' \sigma_{r}(\Ao)$ ~~\text{for } $\gamma' := c \cdot \tfrac{\mu}{L} \cdot \tfrac{\sigma_r(\rhoo)}{\sigma_1(\rhoo)}  $, where $c \leq \tfrac{1}{200}$.
\end{itemize}
Here, $\sigma_r(\cdot)$ denotes the $r$-singular value of the input matrix, in descending order.
Later in the text, we present an initialization that, under assumptions, leads further to global convergence results.

\subsection{Generalized theorem}
Next, we present the full proof of the following generalization of Theorem \ref{thm:main2}:

\begin{theorem}{\label{thm:main3}}
Let $\C \subseteq \mathbb{C}^{d \times r}$ be a convex, compact, and faithful set, with projection operator satisfying the assumptions described above.
Let $f$ be a convex function satisfying Definitions \ref{prelim:def_00} and \ref{prelim:def_01}.

Let $\U_t \in \C$ be the current estimate and $\X_t = \U_t\U_t^\dagger$.
Assume current point $\U_t$ satisfies $ \dist(\U_t, \Uo) \leq \gamma' \sigma_{r}(\Uo)$, for $\gamma' := c \cdot \tfrac{\mu}{L} \cdot \tfrac{\sigma_r(\X_\star)}{\sigma_1(\X_\star)}, ~c \leq \tfrac{1}{200}$, and given $\xi_t(\cdot) \gtrsim 0.78 $ per iteration, the new estimate of \texttt{ProjFGD}, $\U_{t+1} = \Pi_{\mathcal{C}}\left(\U_t - \widehat{\eta} \gradf(\U_t\U_t^\dagger)  \cdot \U_t\right) = \xi_t \cdot \left(\U_t - \widehat{\eta} \gradf(\U_t\U_t^\dagger)  \cdot \U_t\right)$ satisfies
\begin{equation}
\dist(\U_{t+1}, \Uo)^2 \leq \alpha \cdot \dist(\U_t, \Uo)^2, \label{conv:eq_01}
\end{equation}
where $\alpha := 1 - \frac{\mu \cdot \sigma_r(\Xo)}{550(L \sigma_1(\Xo) + \sigma_1(\gradf(\Xo)))} < 1$. 
Further, $\U_{t+1}$ satisfies $ \dist(\U_{t+1}, \Uo) \leq \gamma' \sigma_{r}(\Uo). $
\end{theorem}

When applied to the QST setting, we obtain the following variation of the above theorem:
\begin{theorem}[Local convergence rate for QST]
Let $\rho_\star$ be the quantum state of an $n$-qubit system, $y\in{\mathbb R}^m$ be the measurement vector of $m={\cal O}(r n^6 2^n)$ random $n$-qubit Pauli observables, and $\mathcal{M}$ be the corresponding sensing map, such that $y_i = \left(\mathcal{M}(\rhoo)\right)_i + e_i, ~\forall i = 1, \dots, m$.

Let $\U_t$ be the current estimate of \texttt{ProjFGD}.
Assume $\U_t$ satisfies $ \dist(\U_t, \Uo) \leq \gamma' \sigma_{r}(\Uo)$, for $\gamma' := c \cdot \tfrac{(1-\delta_{4r})}{(1+\delta_{4r})} \cdot \tfrac{\sigma_r(\X_\star)}{\sigma_1(\X_\star)}, ~c \leq \tfrac{1}{200}$, where $\delta_{4r}$ is the RIP constant. 
Then, the new estimate $\U_{t+1}$ satisfies
\begin{equation*}
\dist(\U_{t+1}, \Uo)^2 \leq \alpha \cdot \dist(\U_t, \Uo)^2, 
\end{equation*}
where $\alpha := 1 - \frac{(1 - \delta_{4r}) \cdot \sigma_r(\Xo)}{550((1+\delta_{4r}) \sigma_1(\Xo) + \|e\|_2)} < 1$. 
Further, $\U_{t+1}$ satisfies $ \dist(\U_{t+1}, \Uo) \leq \gamma' \sigma_{r}(\Uo)$.
\end{theorem}

\subsection{Proof of Theorem \ref{thm:main3}}
For our analysis, we make use of the following lemma \cite[Chapter 3]{bubeck2015convex}, which characterizes the effect of projections onto convex sets w.r.t. to inner products, as well as provides 
a type-of triangle inequality for such projections; see also Figure \ref{fig:proj} for a simple illustration.
\begin{lemma}{\label{lem:proj}}
Let $U \in \C \subseteq \mathbb{C}^{d \times r}$ and $V \in \mathbb{C}^{d \times r}$ where $V \notin \C$. Then,
\begin{align}\label{eq:proj_00}
\left\langle \Pi_\C(V) - U, V - \Pi_\C(V) \right\rangle \geq 0.
\end{align} 
\end{lemma}

\begin{figure}[h]
	\centering
	\includegraphics[width=0.5\columnwidth]{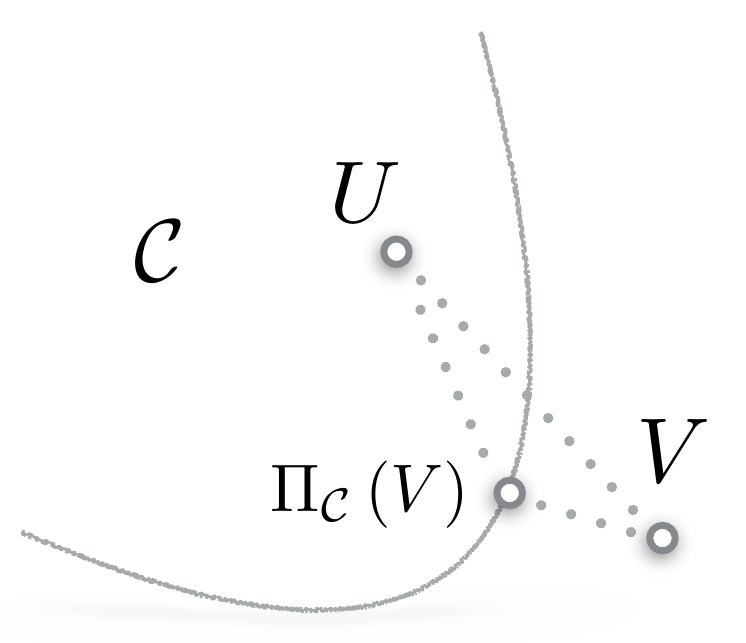}
	\caption{Illustration of Lemma \ref{lem:proj}}{\label{fig:proj}}
\end{figure}

We start with the following series of (in)equalities:
\begin{align*}
\dist\left(\U_{t+1}, ~\Uo\right)^2 \nonumber &= \min_{R \in \mathcal{U}} \|\U_{t+1} - \Uo R\|_F^2 \nonumber \\ 
											  &\stackrel{(i)}{\leq} \|\U_{t+1} - \Uo \Rus\|_F^2 \\
											  &\stackrel{(ii)}{=} \|\U_{t+1} - \Uw_{t+1} + \Uw_{t+1} - \Uo \Rus\|_F^2 \\
											  &= \|\U_{t+1} - \Uw_{t+1}\|_F^2 + \|\Uw_{t+1} - \Uo \Rus\|_F^2 \\
											  &+ 2\left\langle \U_{t+1} - \Uw_{t+1}, ~\Uw_{t+1} - \Uo \Rus \right\rangle,
\end{align*} where $(i)$ is due to the fact $\Rus := \argmin_{R \in \mathcal{U}} \|\U_{t} - \Uo R\|_F^2$, $(ii)$ is obtained by adding and subtracting $\Uw_{t+1}$.

Focusing on the second term of the right hand side, we substitute $\Uw_{t+1}$ to obtain:
\begin{align*}
\|\Uw_{t+1} - \Uo \Rus\|_F^2 \nonumber &= \|\U_t - \weta \gradf\left(\U_t \U_t^\dagger\right)\U_t - \Uo \Rus\|_F^2 \\
										 &= \|\U_t - \Uo \Rus\|_F^2 + \weta^2 \|\gradf \left(\U_t \U_t^\dagger\right) \U_t\|_F^2 \nonumber \\ 
										 &- 2\weta\left\langle \gradf \left(\U_t \U_t^\dagger\right) \U_t,~\U_t - \Uo \Rus \right\rangle
\end{align*} 
Then, our initial equation transforms into:
\begin{align*}
\dist\left(\U_{t+1}, ~\Uo\right)^2 \nonumber &\leq \|\U_{t+1} - \Uw_{t+1}\|_F^2 \\
&+ \dist\left(\U_t, ~\Uo\right)^2 + \weta^2 \|\gradf \left(\U_t \U_t^\dagger\right) \U_t\|_F^2 \nonumber \\ 
&- 2\weta \left\langle \gradf \left(\U_t \U_t^\dagger\right) \U_t,~\U_t - \Uo \Rus \right\rangle \nonumber \\ 
&+ 2\left\langle \U_{t+1} - \Uw_{t+1}, ~\Uw_{t+1} - \Uo \Rus \right\rangle
\end{align*}
Focusing further on the last term of the expression above, we obtain:
\begin{align*}
&\left\langle \U_{t+1} - \Uw_{t+1}, ~\Uw_{t+1} - \Uo \Rus \right\rangle \\ 
&= \left\langle \U_{t+1} - \Uw_{t+1}, ~\Uw_{t+1} - \U_{t+1} + \U_{t+1} - \Uo \Rus \right\rangle \\
&= \left\langle \U_{t+1} - \Uw_{t+1}, ~\Uw_{t+1} - \U_{t+1} \right\rangle \nonumber \\ 
&\quad \quad \quad \quad \quad \quad \quad+ \left\langle \U_{t+1} - \Uw_{t+1}, ~\U_{t+1} - \Uo \Rus \right\rangle
\end{align*} 
Observe that, in the special case where $\Uw_{t+1} \equiv \U_{t+1}$ for all $t$, \emph{i.e.}, the iterates are always within $\C$ before the projection step, the above equation equals to zero and the recursion is identical to that of \cite{bhojanapalli2016dropping}[Proof of Theorem 4.2]. 
Here, we are more interested in the case where $\Uw_{t+1} \not\equiv \U_{t+1}$ for some $t$---thus $\Uw_{t+1} \not\in \C$. 
By faithfulness (Assumption \ref{prelim:def_02}), observe that $\Uo \Rus \in \C$ and $\Xo = \Uo \Rus \left(\Uo \Rus\right)^\dagger = \Uo \Uo^\dagger$.
Moreover, $\U_{t+1} = \Pi_{\C}(\Uw_{t+1})$:
Then, according to Lemma \ref{lem:proj} and focusing on eq. \eqref{eq:proj_00}, for $U := \Uo \Rus$ and $V := \Uw_{t+1}$, the last term in the above equation satisfies: 
\begin{align*}
\left\langle \U_{t+1} - \Uw_{t+1}, ~\U_{t+1} - \Uo \Rus \right\rangle \leq 0,
\end{align*} 
and, thus, the expression above becomes:
\begin{align*}
\left\langle \U_{t+1} - \Uw_{t+1}, ~\Uw_{t+1} - \Uo \Rus \right\rangle \leq - \|\U_{t+1} - \Uw_{t+1}\|_F^2.
\end{align*}
Therefore, going back to the original recursive expression, we obtain:
\begin{align*}
\dist\left(\U_{t+1}, ~\Uo\right)^2 &\leq - \|\U_{t+1} - \Uw_{t+1}\|_F^2 \\
&+ \dist\left(\U_t, ~\Uo\right)^2 + \weta^2 \|\gradf \left(\U_t \U_t^\dagger\right) \U_t\|_F^2 \nonumber \\ 
&- 2\weta\left\langle \gradf \left(\U_t \U_t^\dagger\right) \U_t, \U_t - \Uo \Rus \right\rangle 
\end{align*} 

For the last term, we use the following descent lemma~\ref{lem:gradU,U-U_r_ bound}; its proof is provided in Section \ref{descent_lemma_proof}.
\begin{lemma}[Descent lemma]\label{lem:gradU,U-U_r_ bound}
Let $\widetilde{\U}_{t+1} = \U_t - \widehat{\eta} \gradf(\X_t)  \cdot \U_t$.
For $f$ restricted $L$-smooth and $\mu$-strongly convex, and under the same assumptions with Theorem \ref{thm:main3}, the following inequality holds true:
\begin{align*}
&2 \weta \big \langle \gradf(\U_t\U_t^\dagger) \cdot \U_t, ~\U_t - \Uo R_{\U_t}^\star \big \rangle +  \|\U_{t+1} - \widetilde{\U}_{t+1}\|_F^2 \nonumber \\ 
&\quad \quad \quad \geq \weta^2\|\gradf(\U_t\U_t^\dagger) \U_t\|_F^2 +  \tfrac{3\weta \mu}{10} \cdot \sigma_r(\Xo) \cdot \dist(\U_t, \Uo)^2.
\end{align*}
\end{lemma}

Using the above lemma in our expression, we get:
\begin{align*}
\dist\left(\U_{t+1}, ~\Uo\right)^2 &\leq \left(1 - \tfrac{3\weta \mu}{10} \cdot \sigma_r(\Xo)\right) \cdot \dist(\U_t, \Uo)^2.
\end{align*}

The expression for $\alpha$ is obtained by observing $\weta \geq \tfrac{5}{6} \eta$ and $\tfrac{10}{11}\eta_\star \leq \eta \leq \tfrac{11}{10}\eta_\star$, from Lemma 20 in \cite{bhojanapalli2016dropping}. 
Then, for $\eta_\star \leq \frac{C}{L \sigma_1(\X_\star) + \sigma_1(\gradf(X_\star) )}$ and $C = \sfrac{1}{128}$, we have:
\begin{align*}
1 - \frac{3\weta \mu}{10} \cdot \sigma_r(\Xo) &\leq 1 - \frac{3 \cdot \tfrac{10}{11} \cdot \tfrac{5}{6} \eta_\star \mu}{10} \cdot \sigma_r(\Xo) \nonumber \\ 
															&= 1 - \frac{15}{66} \eta_\star \mu \cdot \sigma_r(\Xo) \nonumber \\ 
															&= 1 - \frac{15}{66} \frac{\mu \cdot \sigma_r(\Xo)}{128(L \sigma_1(\Xo) + \sigma_1(\gradf(\Xo)))} \nonumber \\ 
															&\leq 1 - \frac{\mu \cdot \sigma_r(\Xo)}{550(L \sigma_1(\Xo) + \sigma_1(\gradf(\Xo)))} =: \alpha
\end{align*} where $\alpha < 1$.

Concluding the proof, the condition $\dist(\U_{t+1}, \Uo)^2 \leq \gamma' \sigma_{r}(\Uo)$ is naturally satisfied, since $\alpha < 1$.

\subsection{Proof of Lemma \ref{lem:gradU,U-U_r_ bound}} {\label{descent_lemma_proof}}
Recall $\Uw_{t+1} = \U_t - \weta\gradf(\X_t)\U_t$ and define $\Delta := \U_t - \Uo R_{\U_t}^\star$. 
Before presenting the proof, we need the following lemma that bounds one of the error terms arising in the proof of Lemma~\ref{lem:gradU,U-U_r_ bound}. 
This is a variation of Lemma 6.3 in \cite{bhojanapalli2016dropping}. 
The proof is presented in Section \ref{sec:DD_proof}.

\begin{lemma}\label{lem:DD_bound_sc}
Let $f$ be $L$-smooth and $\mu$-restricted strongly convex. 
Then, under the assumptions of Theorem \ref{thm:main3} and assuming step size $\widehat{\eta} = \tfrac{1}{128(L \sigma_1(\X_t) + \sigma_1(\gradf(\X_t)Q_{\U_t} Q_{\U_t}^\dagger))}$,
the following bound holds true:
\begin{small}
\begin{align}{\label{eq:appe_combine1}}
\ip{\gradf(\X_t) }{ \Delta \Delta^\dagger} \geq - \tfrac{\weta}{5} \|\gradf(\X_t) \U_t \|_F^2 - \tfrac{\mu\sigma_{r}(\Xo)}{10} \cdot \dist(\U_t, \Uo)^2.
\end{align} 
\end{small}
\end{lemma} 

Now we are ready to present the proof of Lemma~\ref{lem:gradU,U-U_r_ bound}.
\begin{proof}[Proof of Lemma \ref{lem:gradU,U-U_r_ bound}]
First we rewrite the inner product as shown below.
\begin{align}
&\ip{\gradf(\X_t)\U_t}{\U_t - \Uo R_{\U_t}^\star} \nonumber \\ 
&= \ip{\gradf(\X_t) }{\X_t - \Uo R_{\U_t}^\star \U_t^\dagger} \nonumber \\
&=\tfrac{1}{2}\ip{\gradf(\X_t) }{\X_t -\Xo}  + \ip{\gradf(\X_t) }{\tfrac{1}{2}(\X_t + \Xo) - \Uo R_{\U_t}^\star \U_t^\dagger} \nonumber \\
&=\tfrac{1}{2}\ip{\gradf(\X_t) }{\X_t -\Xo}  + \tfrac{1}{2} \ip{\gradf(\X_t) }{\Delta \Delta^\dagger}, \label{proofsr1:eq_09}
\end{align} 
which follows by adding and subtracting $\tfrac{1}{2}\Xo$.

Let us focus on bounding the first term on the right hand side of \eqref{proofsr1:eq_09}. 
Consider points $\X_t = \U_t\U_t^\dagger$ and $\X_{t+1} = \U_{t+1} \U_{t+1}^\dagger$; by assumption, both $\X_t$ and $\X_{t+1}$ are feasible points in \eqref{eq:appe_01}.
By smoothness of $f$, we get:
\begin{align}
f(\X_t) &\geq f(\X_{t+1}) -\ip{\gradf(\X_t)}{\X_{t+1} -\X_t} - \tfrac{L}{2} \norm{\X_{t+1} -\X_t}_F^2 \nonumber \\ 
&\stackrel{(i)}{\geq} f(\Xo) -\ip{\gradf(\X_t)}{\X_{t+1} -\X_t} - \tfrac{L}{2} \norm{\X_{t+1} -\X_t}_F^2 , \label{eq:gradX,X-X_r_00}
 \end{align}
where $(i)$ follows from optimality of $\Xo$ and since $\X_{t+1}$ is a feasible point $(\X_{t+1} \succeq 0, ~\Pi_{\C'}(\X_{t+1}) = \X_{t+1})$ for problem~\eqref{eq:appe_00}.

Moreover, by the restricted strong convexity of $f$, we get, 
\begin{align} 
f(\Xo) \geq f(\X_t) +\ip{\gradf(\X_t)}{\Xo-\X_t} +\tfrac{\mu}{2}\norm{\Xo -\X_t}_F^2 .\label{eq:gradX,X-X_r_02}
\end{align}
Combining equations~\eqref{eq:gradX,X-X_r_00}, and~\eqref{eq:gradX,X-X_r_02}, we obtain: 
\begin{align} 
\ip{\gradf(\X_t)}{\X_t-\Xo} &\geq  \ip{\gradf(\X_t)}{\X_t -\X_{t+1}} \nonumber \\ &-\tfrac{L}{2} \norm{\X_{t+1} -\X_t}_F^2+\tfrac{\mu}{2}\norm{\Xo -\X_t}_F^2  \label{eq:gradX,X-X_r_03} 
\end{align}
By the nature of the projection $\Pi_\C(\cdot)$ step, it is easy to verify that $$\X_{t+1} =  \xi^2 \cdot \left(\X_t - \weta \gradf(\X_t) \X_t \Lambda_t - \weta \Lambda_t^\dagger \X_t^\dagger \gradf(\X_t)^\dagger\right),$$ where $\Lambda_t = I - \tfrac{\weta}{2} Q_{\U_t} Q_{\U_t}^\dagger \gradf(X_t) \in \mathbb{C}^{d \times d}$ and $Q_{\U_t} Q_{\U_t}^\dagger$ denoting the projection onto the column space of $\U_t$.
Notice that, for step size $\weta$, we have 
\begin{align*}
\Lambda_t \succ 0, \quad  \sigma_1\left(\Lambda_t\right) \leq  1+ \sfrac{1}{256} \quad \text{and} \quad \sigma_{d}(\Lambda_t) \geq  1- \sfrac{1}{256}.
\end{align*} 

Using the above $\X_{t+1}$ characterization in~\eqref{eq:gradX,X-X_r_03}, we obtain:
\begin{footnotesize}
\begin{align}\label{eq:appe_000}
& \ip{\gradf(\X_t)}{\X_t-\Xo} - \tfrac{\mu}{2}\norm{\Xo -\X_t}_F^2 + \tfrac{L}{2} \norm{\X_t -\X_{t+1}}_F^2 \nonumber \\
 								 &\stackrel{(i)}{\geq} \ip{\gradf(\X_t)}{\left(1 - \xi^2\right)\X_t} + 
 								 2\weta \cdot \xi^2 \cdot \ip{\gradf(\X_t)}{\gradf(\X_t)X_t \Lambda_t} \nonumber \\
 								 &\stackrel{(ii)}{\geq}\left(1 - \xi^2\right) \cdot \ip{\gradf(\X_t)\U_t}{\U_t} + 2\weta \cdot \xi^2 \cdot \trace( \gradf(\X_t) \gradf(\X_t)\X_t) \cdot \sigma_{d}(\Lambda_t) \nonumber \\
 								 &\geq \left(1 - \xi^2\right) \cdot \ip{\gradf(\X_t)\U_t}{\U_t} + \tfrac{255\cdot\weta \cdot \xi^2}{128} \|\gradf(\X_t) \U_t\|_F^2,
\end{align} 
\end{footnotesize}
where: $(i)$ follows from symmetry of $\gradf(\X_t)$ and $\X_t$ and, $(ii)$ follows from the sequence equalities and inequalites:
\begin{footnotesize}
\begin{align*}
&\trace(\gradf(\X_t)\gradf(\X_t)\X_t\Lambda_t) \nonumber \\ &= \trace(\gradf(\X_t)\gradf(\X_t)\U_t\U_t^\dagger) - \tfrac{\weta}{2}\trace(\gradf(\X_t)\gradf(\X_t)\U_t\U_t^\dagger \gradf(\X_t)) \nonumber \\
&\geq \left(1 -\tfrac{\weta}{2} \|Q_{\U_t} Q_{\U_t}^\dagger\gradf(\X_t)\|_2\right) \|\gradf(\X_t) \U_t\|_F^2 \nonumber \\ 
&\geq \left(1- \sfrac{1}{256} \right) \|\gradf(\X_t) \U_t\|_F^2.
\end{align*} 
\end{footnotesize}

 Combining the above in the expression we want to lower bound: $2 \weta \left \langle \gradf(\X_t) \cdot \U_t, \U_t - \Uo R_{\U_t}^\star \right \rangle + \|\U_{t+1} - \widetilde{\U}_{t+1}\|_F^2$, we obtain:
\begin{align}
&2 \weta \left \langle \gradf(\X_t) \cdot \U_t, \U_t - \Uo R_{\U_t}^\star \right \rangle + \|\U_{t+1} - \widetilde{\U}_{t+1}\|_F^2 \nonumber \\ 
&=\weta \ip{\gradf(\X_t) }{\X_t -\Xo}  + \weta \ip{\gradf(\X_t) }{\Delta \Delta^\dagger} + \|\U_{t+1} - \widetilde{\U}_{t+1}\|_F^2 \nonumber \\ 
&\geq \left(1 - \xi^2\right) \cdot \weta \ip{\gradf(\X_t)\U_t}{\U_t} + \tfrac{255 \cdot \weta^2 \cdot \xi^2}{128} \|\gradf(\X_t) \U_t\|_F^2 \nonumber \\
&\quad \quad \quad \quad \quad + \tfrac{\weta \mu}{2}\norm{\Xo -\X_t}_F^2 - \tfrac{\weta L}{2} \norm{\X_t -\X_{t+1}}_F^2 \nonumber \\
&\quad \quad \quad \quad \quad - \tfrac{\weta^2}{5} \|\nabla f(\X_t) \U_t\|_F^2 - \tfrac{\weta \mu \sigma_r(\Xo)}{10} \cdot \dist(\U_t, \Uo)^2 \nonumber \\
&\quad \quad \quad \quad \quad + \|\U_{t+1} - \widetilde{\U}_{t+1}\|_F^2 \label{eq:new_p_00}
\end{align} 
 
For the last term in the above expression and given $\U_{t+1} = \Pi_\C\left(\Uw_{t+1}\right) = \xi \cdot \Uw_{t+1}$ for some $\xi \in (0, 1)$, we further observe:
\begin{footnotesize}
\begin{align*}
\|\U_{t+1} - \widetilde{\U}_{t+1}\|_F^2 &= \|\xi \cdot \widetilde{\U}_{t+1} - \widetilde{\U}_{t+1}\|_F^2 \nonumber \\
&=\left(1 - \xi\right)^2 \cdot \|\U_t\|_F^2 + \left(1 - \xi\right)^2 \weta^2 \cdot \|\gradf(\X_t) \U_t\|_F^2  \nonumber \\ 
&- 2\left(1 - \xi\right)^2 \cdot \weta \cdot \ip{\gradf(\X_t)\U_t}{\U_t}
\end{align*} 
\end{footnotesize}

Combining the above equality with the first term on the right hand side in \eqref{eq:new_p_00}, we obtain the expression in \eqref{eq:figure1}.
\begin{figure*}
\begin{align}
\left(1 - \xi^2\right) \cdot \weta \ip{\gradf(\X_t)\U_t}{\U_t} \nonumber &+ \left(1 - \xi\right)^2 \cdot \|\U_t\|_F^2 + \left(1 - \xi\right)^2 \weta^2 \cdot \|\gradf(\X_t) \U_t\|_F^2  \\&- 2\left(1 - \xi\right)^2 \cdot \weta \cdot \ip{\gradf(\X_t)\U_t}{\U_t} = \nonumber \\
\left[\left(1 - \xi^2\right) - 2\left(1 - \xi\right)^2 \right] \cdot \weta \ip{\gradf(\X_t)\U_t}{\U_t} &+ \left(1 - \xi\right)^2 \cdot \|\U_t\|_F^2 + \left(1 - \xi\right)^2 \weta^2 \cdot \|\gradf(\X_t) \U_t\|_F^2 \nonumber = \\
\left(3\xi - 1\right) \left(1 - \xi\right) \cdot \weta \ip{\gradf(\X_t)\U_t}{\U_t} &+ \left(1 - \xi\right)^2 \cdot \|\U_t\|_F^2 + \left(1 - \xi\right)^2 \weta^2 \cdot \|\gradf(\X_t) \U_t\|_F^2 \nonumber = \\
\left\| \tfrac{3\xi - 1}{2} \cdot \U_t + (1 - \xi) \cdot \weta \gradf(\X_t) \cdot \U_t\right\|_F^2 &+ \left(\left(1 - \xi\right)^2 - \tfrac{(3\xi - 1)^2}{4}\right) \|\U_t\|_F^2. \label{eq:figure1} 
\end{align}
\hrulefill
\end{figure*}
Focusing on the first term, let $\Theta_t := I + \tfrac{2(1-\xi)}{3\xi - 1} \cdot \weta \cdot \gradf(\X_t)Q_{\U_t} Q_{\U_t}^\dagger$; then, $\sigma_d(\Theta_t) \geq 1 - \tfrac{2(1-\xi)}{3\xi - 1} \cdot \tfrac{1}{128} $, by the definition of $\weta$ and the fact that $\weta \leq \tfrac{1}{128 \sigma_1\left(\gradf(\X_t)Q_{\U_t} Q_{\U_t}^\dagger\right)}$. 
Then:
\begin{align*}
&\left\| \tfrac{3\xi - 1}{2} \cdot \U_t + (1 - \xi) \cdot \weta \gradf(\X_t) \cdot \U_t\right\|_F^2 = \left\| \tfrac{3\xi - 1}{2} \Theta_t \cdot \U_t \right\|_F^2 \nonumber \\ 
&\geq \tfrac{(3\xi - 1)^2}{4} \cdot \|\U_t\|_F^2 \cdot \sigma_d(\Theta_t)^2 \nonumber \\
&\geq \tfrac{(3\xi - 1)^2}{4} \cdot \left(1 - \tfrac{2(1-\xi)}{3\xi - 1} \cdot \tfrac{1}{128} \right)^2 \cdot \|\U_t\|_F^2 \nonumber
\end{align*}
Combining the above, we obtain the following bound:
\begin{align*}
&\left(1 - \xi^2\right) \cdot \weta \ip{\gradf(\X_t)\U_t}{\U_t} + \|\U_{t+1} - \widetilde{\U}_{t+1}\|_F^2 \nonumber \\ &\geq \left((1 - \xi)^2 - \tfrac{(3\xi - 1)^2}{4} \cdot \left(1 - \left(1 - \tfrac{2(1-\xi)}{3\xi - 1} \cdot \tfrac{1}{128} \right)^2\right)\right) \cdot \|\U_t\|_F^2
\end{align*}
The above transform \eqref{eq:new_p_00} as follows:
\begin{footnotesize}
\begin{align}
2 \weta \big \langle &\gradf(\X_t) \cdot \U_t, ~\U_t - \Uo R_{\U_t}^\star \big \rangle +  \|\U_{t+1} - \widetilde{\U}_{t+1}\|_F^2 \nonumber \\ 
&\geq  \left(\tfrac{255 \cdot \xi^2}{128} - \tfrac{1}{5}\right)\cdot \weta^2\|\gradf(X_t) U_t\|_F^2 + \tfrac{\weta \mu}{2}\norm{\Xo -\X_t}_F^2 \nonumber \\ 
&- \tfrac{\weta \mu \sigma_r(\Xo)}{10} \cdot \dist(U_t, \Uo)^2 \nonumber \\
&+\left((1 - \xi)^2 - \tfrac{(3\xi - 1)^2}{4} \cdot \left(1 - \left(1 - \tfrac{2(1-\xi)}{3\xi - 1} \cdot \tfrac{1}{128} \right)^2\right)\right) \cdot \|\U_t\|_F^2 \nonumber \\ 
&- \tfrac{\weta L}{2} \norm{\X_t -\X_{t+1}}_F^2 \label{eq:new_p_01}
\end{align} 
\end{footnotesize}
 
Let us focus on the term $\tfrac{\weta L}{2} \norm{\X_t - \X_{t+1}}_F^2$; this can be bounded as follows:
\begin{footnotesize}
\begin{align*}
&\tfrac{\weta L}{2} \norm{\X_t - \X_{t+1}}_F^2 = \tfrac{\weta L}{2} \|\U_t\U_t^\dagger - \U_{t+1}\U_{t+1}^\dagger\|_F^2 \\ 
&\quad \quad \quad = \tfrac{\weta L}{2} \|\U_t\U_t^\dagger - \U_t\U_{t+1}^\dagger + \U_t\U_{t+1}^\dagger - \U_{t+1}\U_{t+1}^\dagger\|_F^2 \\
&\quad \quad \quad = \tfrac{\weta L}{2} \|\U_t\left(\U_t - \U_{t+1}\right)^\dagger + \left(\U_t - \U_{t+1}\right)\U_{t+1}^\dagger\|_F^2 \nonumber \\ 
&\quad \quad \quad \stackrel{(i)}{\leq} \weta L \cdot \left(\|\U_t\left(\U_t - \U_{t+1}\right)^\dagger\|_F^2 + \|\left(\U_t - \U_{t+1}\right)\U_{t+1}^\dagger\|_F^2\right) \\
&\quad \quad \quad \stackrel{(ii)}{\leq} \weta L \left(\sigma_1\left(\U_{t+1}\right)^2 + \sigma_1\left(\U_t\right)^2\right) \cdot \|\U_{t+1} - \U_t\|_F^2.
\end{align*} 
\end{footnotesize}
where $(i)$ is due to the identity $\|A + B\|_F^2 \leq 2\|A\|_F^2 + 2\|B\|_F^2$ and $(ii)$ is due to the Cauchy-Schwarz inequality. 
By definition of $\U_{t+1}$, we observe that:
\begin{small}
\begin{align*}
\sigma_1\left(\U_{t+1}\right)^2 &= \sigma_1\left(\xi \cdot \left(\U_t - \weta\nabla f(\X_t)\U_t\right)\right)^2 \\
&\stackrel{(i)}{\leq} \xi^2 \cdot \sigma_1\left(\U_t\right)^2 \cdot \sigma_1\left(I - \weta  \nabla f(\X_t)Q_{\U_t} Q_{\U_t}^\dagger\right)^2 \\
&\stackrel{(ii)}{\leq} \left(1 + \tfrac{1}{128}\right)^2 \cdot \sigma_1\left(\U_t\right)^2.
\end{align*}
\end{small} 
where $(i)$ is due to Cauchy-Schwarz and $(ii)$ is obtained by substituting $\weta \leq \tfrac{1}{128 \sigma_1\left(\nabla f(\X_t)Q_{\U_t} Q_{\U_t}^\dagger\right)}$ and since $\xi \in (0, 1)$. 
Thus, $\tfrac{\weta L}{2} \norm{\X_t - \X_{t+1}}_F^2$ can be further bounded as follows:
\begin{footnotesize}
\begin{align*}
\tfrac{\weta L}{2} \norm{\X_t - \X_{t+1}}_F^2 &\leq \weta L \left( \left(1 + \tfrac{1}{128}\right)^2 + 1\right) \cdot \sigma_1\left(\U_t\right)^2 \cdot \|\U_{t+1} - \U_t\|_F^2 \\ 
&=  \weta L \left( \left(1 + \tfrac{1}{128}\right)^2 + 1\right) \cdot \sigma_1\left(\X_t\right) \cdot \|\U_{t+1} - \U_t\|_F^2 \\ 
&\leq \tfrac{ \left(1 + \tfrac{1}{128}\right)^2 + 1}{128 } \cdot \|\U_{t+1} - \U_t\|_F^2 \\
&= \tfrac{ \left(1 + \tfrac{1}{128}\right)^2 + 1}{128 } \cdot \|\xi \cdot \widetilde{\U}_{t+1} - \U_t\|_F^2 \\
&= \tfrac{ \left(1 + \tfrac{1}{128}\right)^2 + 1}{128 } \cdot \|(\xi -1)\U_t - \xi \cdot \weta \gradf(\X_t) \cdot \U_t\|_F^2 \\
&\leq  (1 - \xi)^2 \cdot \tfrac{ \left(1 + \tfrac{1}{128}\right)^2 + 1}{64 } \cdot \|\U_t\|_F^2 \\ &\quad \quad \quad \quad + \tfrac{ \left(1 + \tfrac{1}{128}\right)^2 + 1}{64} \cdot \xi^2 \cdot \weta^2 \cdot  \|\gradf(\X_t) \cdot \U_t\|_F^2
\end{align*} 
\end{footnotesize}
where in the last inequality we substitute $\weta$; observe that $\weta \leq \tfrac{1}{128 L \sigma_1\left(\X_t\right)}$.
Combining this result with \eqref{eq:new_p_01}, we obtain:
\begin{footnotesize}
\begin{align}
2 \weta \big \langle &\gradf(\X_t) \cdot \U_t, ~\U_t - \Uo R_{\U_t}^\star \big \rangle +  \|\U_{t+1} - \widetilde{\U}_{t+1}\|_F^2 \nonumber \\ 
&\geq  \left(\tfrac{255 \cdot \xi^2}{128} - \tfrac{1}{5} -  \tfrac{ \left(1 + \tfrac{1}{128}\right)^2 + 1}{64} \cdot \xi^2\right)\cdot \weta^2\|\gradf(\X_t) \U_t\|_F^2 \nonumber \\ 
&\quad \quad + \tfrac{\weta \mu}{2}\norm{\Xo -\X_t}_F^2 - \tfrac{\weta \mu \sigma_r(\Xo)}{10} \cdot \dist(\U_t, \Uo)^2 \nonumber \\
&\quad \quad+\Bigg((1 - \xi)^2\cdot \left(1 - \tfrac{ \left(1 + \tfrac{1}{128}\right)^2 + 1}{64} \right) \nonumber \\ 
&\quad \quad- \tfrac{(3\xi - 1)^2}{4} \cdot \left(1 - \left(1 - \tfrac{2(1-\xi)}{3\xi - 1} \cdot \tfrac{1}{128} \right)^2\right)\Bigg) \cdot \|\U_t\|_F^2 \nonumber \\ 
&\stackrel{(i)}{\geq} \weta^2\|\gradf(\X_t) \U_t\|_F^2 + \tfrac{\weta \mu}{2}\norm{\Xo -\X_t}_F^2 - \tfrac{\weta \mu \sigma_r(\Xo)}{10} \cdot \dist(\U_t, \Uo)^2 \nonumber \\
&\quad \quad~+\left((1 - \xi)^2\cdot \Bigg(1 - \tfrac{ \left(1 + \tfrac{1}{128}\right)^2 + 1}{64} \right) \nonumber \\
&\quad \quad~- \tfrac{(3\xi - 1)^2}{4} \cdot \left(1 - \left(1 - \tfrac{2(1-\xi)}{3\xi - 1} \cdot \tfrac{1}{128} \right)^2\right)\Bigg) \cdot \|\U_t\|_F^2 \nonumber \\ 
&\stackrel{(ii)}{\geq} \weta^2\|\gradf(\X_t) \U_t\|_F^2 + \tfrac{\weta \mu}{2}\norm{\Xo -\X_t}_F^2 - \tfrac{\weta \mu \sigma_r(\Xo)}{10} \cdot \dist(\U_t, \Uo)^2
\label{eq:new_p_02}
\end{align} 
\end{footnotesize}
where $(i)$ is due to the assumption $\xi \gtrsim 0.78$ and thus $\big(\tfrac{255 \cdot \xi^2}{128} - \tfrac{1}{5} -  \tfrac{ \left(1 + \tfrac{1}{128}\right)^2 + 1}{64} \cdot \xi^2\big) \geq 1$; see also Figure \ref{fig:appendix_xi} (top panel), 
and $(ii)$ is due to the non-negativity of the constant in front of $\|\U_t\|_F^2$; see also Figure \ref{fig:appendix_xi} (bottom panel).

In the special case where $\C = \left\{\A \in \mathbb{C}^{d \times r} : \|\A\|_F^2 \leq 1\right\}$, as in QST, the assumption  $\xi \gtrsim 0.78$ is always satisfied, according to the following Corollary; the proof is provided in Subsection \ref{sec:cor:xi}:

\begin{corollary}{\label{cor:projFGD_schatten}}
If $\C = \left\{\U \in \mathbb{C}^{d \times r}: ~\|\U\|_F \leq 1 \right\}$, 
then \texttt{ProjFGD} inherently satisfies $\tfrac{128}{129} \leq \xi_t(\cdot) \leq 1$, for every $t$. I.e., it guarantees \eqref{conv:eq_01} without assumptions on $\xi_t(\cdot)$.
\end{corollary}

We conjecture that the lower bound on $\xi_t(\cdot)$ for more generic cases of $\C$ could possibly be improved with a different analysis.

\begin{figure}[!ht]
	\centering
	\includegraphics[width=0.36\textwidth]{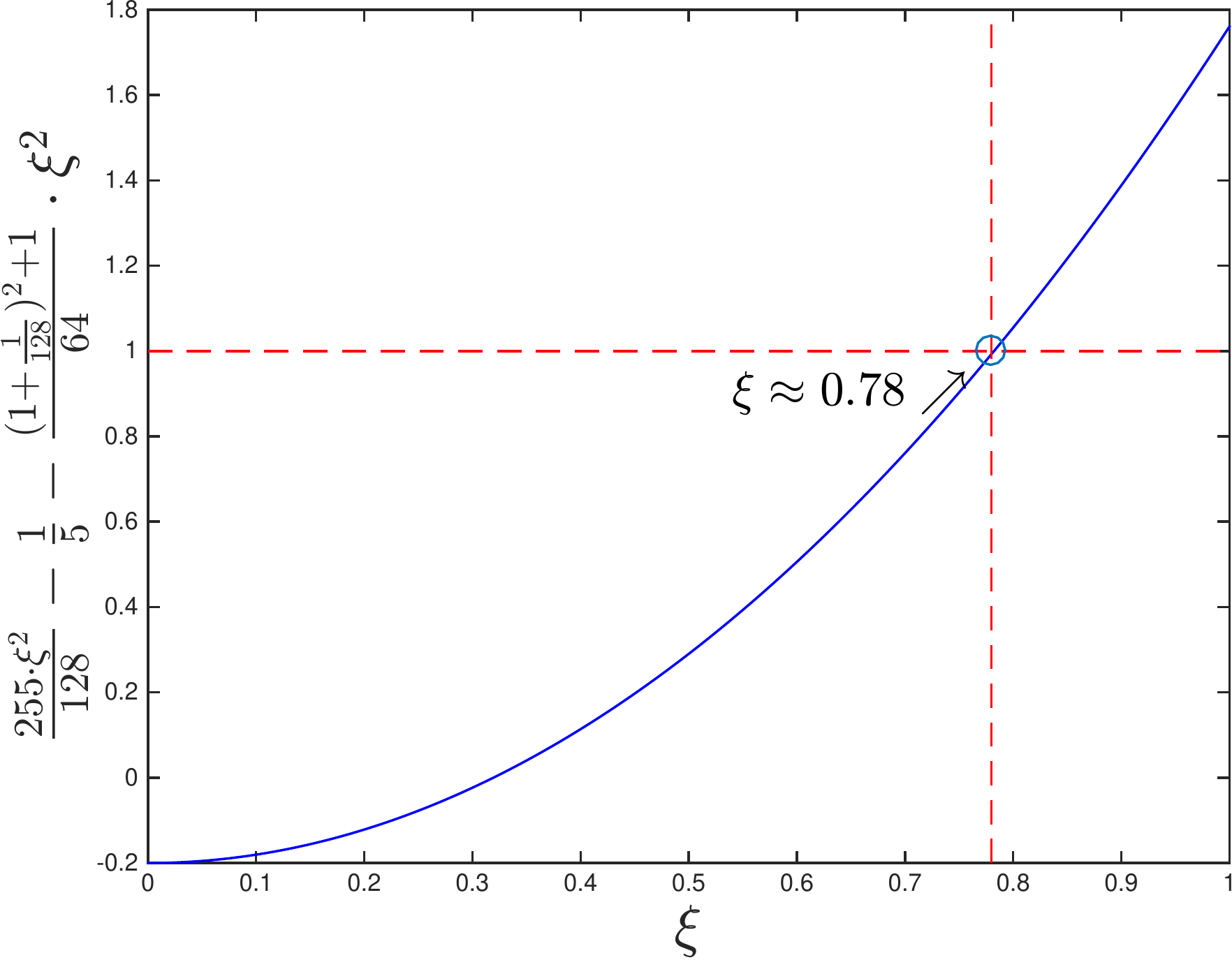} \hspace{0.3cm}
	\includegraphics[width=0.35\textwidth]{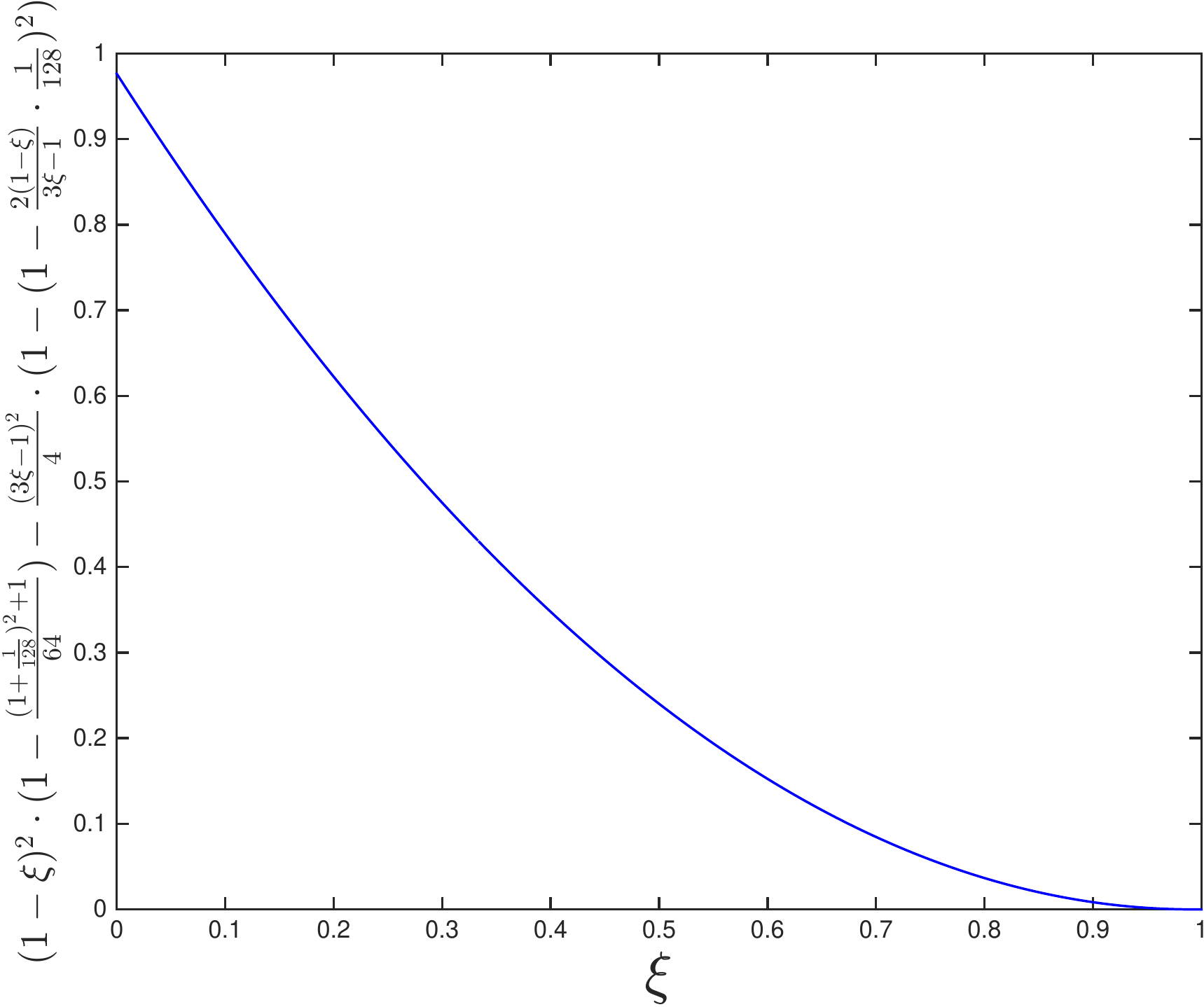}
	\caption{Behavior of constants, depending on $\xi$, in expression \eqref{eq:new_p_02}.} \label{fig:appendix_xi}
\end{figure}

Finally, we bound $\tfrac{\weta \mu}{2}\norm{\Xo -\X_t}_F^2$ using the following Lemma from \cite{tu2015low}:
\begin{lemma}{\label{lem:tu}}
For any $\U, \V \in \mathbb{C}^{d \times r}$, we have:
\begin{align*}
\|\U\U^\dagger - \V\V^\dagger\|_F^2 \geq 2 \cdot \left(\sqrt{2} - 1\right) \cdot \sigma_r(\U)^2 \cdot \dist(\U, \V)^2.
\end{align*} 
\end{lemma}

In our case, this translates to:
\begin{align*}
\tfrac{\weta \mu}{2}\norm{\Xo -\X_t}_F^2 \geq \weta \mu \cdot  \left(\sqrt{2} - 1\right) \cdot \sigma_r(\Xo) \cdot \dist(\U_t, \Uo)^2,
\end{align*} 
and we can thus conclude to:
\begin{footnotesize}
\begin{align*}
2 \weta \big \langle &\gradf(\X_t) \cdot \U_t, ~\U_t - \Uo R_{\U_t}^\star \big \rangle +  \|\U_{t+1} - \widetilde{\U}_{t+1}\|_F^2 \nonumber \\ 
&\geq \weta^2\|\gradf(\X_t) \U_t\|_F^2 +  \weta \mu  \left(\sqrt{2} - 1\right)  \sigma_r(\Xo) \cdot \dist(\U_t, \Uo)^2 \nonumber \\ 
&\quad \quad \quad - \tfrac{\weta \mu \sigma_r(\Xo)}{10} \cdot \dist(\U_t, \Uo)^2 \\
&= \weta^2\|\gradf(\X_t) \U_t\|_F^2 +  \left(\sqrt{2} - 1 - \tfrac{1}{10}\right)  \weta \mu \sigma_r(\Xo) \cdot \dist(\U_t, \Uo)^2 \\
&= \weta^2\|\gradf(\X_t) \U_t\|_F^2 +  \tfrac{3\weta \mu}{10} \cdot \sigma_r(\Xo) \cdot \dist(\U_t, \Uo)^2 
\end{align*}
\end{footnotesize}
This completes the proof.
\end{proof}

\subsection{Proof of Lemma \ref{lem:DD_bound_sc}}{\label{sec:DD_proof}}

\begin{proof}
We can lower bound $\ip{\gradf(\X_t) }{ \Delta \Delta^\dagger}$ as follows:
\begin{footnotesize}
\begin{align}
&\ip{\gradf(\X_t) }{ \Delta \Delta^\dagger} \stackrel{(i)}{=} \ip{Q_{\Delta} Q_{\Delta}^\dagger \gradf(\X_t) }{ \Delta \Delta^\dagger} \nonumber \\
&\geq -  \left|\trace\left(Q_{\Delta} Q_{\Delta}^\dagger \gradf(\X_t) \Delta \Delta^\dagger\right)\right| \nonumber \\ 
&\stackrel{(ii)}{\geq} - \sigma_1\left( Q_{\Delta} Q_{\Delta}^\dagger \gradf(\X_t) \right) \trace( \Delta \Delta^\dagger)  \nonumber \\
&\stackrel{(iii)}{\geq} - \Bigg( \sigma_1\left(Q_{\U_t} Q_{\U_t}^\dagger\gradf(\X_t)\right) \nonumber \\ 
&\quad \quad \quad \quad + \sigma_1\left(Q_{\Uo} Q_{\Uo}^\dagger\gradf(\X_t)\right) \Bigg)  \dist(\U_t, \Uo)^2 .   \label{proofsr1:eq_11}
\end{align}
\end{footnotesize}
Note that  $(i)$ follows from the fact $\Delta =Q_{\Delta} Q_{\Delta}^\dagger \Delta$ and 
$(ii)$ follows from $|\trace(AB)| \leq \sigma_1\left(A\right) \trace(B)$, for PSD matrix $B$ (Von Neumann's trace inequality~\cite{mirsky1975trace}). 
For the transformation in $(iii)$, we use that fact that the column space of $\Delta$, $\text{\textsc{Span}}(\Delta)$, is a subset of $\text{\textsc{Span}}(\U_t \cup \Uo)$, as $\Delta$ is a linear combination of $\U_t$ and $\Uo R_{\U_t}^\star$. 

For the second term in the parenthesis above, we first derive the following inequalities; their use is apparent later on:
\begin{footnotesize}
\begin{align*}
&\sigma_1\left(\gradf(\X_t)\Uo\right) \\
&\stackrel{(i)}{\leq} \sigma_1\left(\gradf(\X_t)\U_t\right) + \sigma_1\left(\gradf(\X_t)\Delta\right) \\
&\stackrel{(ii)}{\leq} \sigma_1\left(\gradf(\X_t)\U_t\right) + \sigma_1\left(\gradf(\X_t)Q_{\Delta}Q_{\Delta}^\dagger\right) \sigma_1\left(\Delta\right) \\
&\stackrel{(iii)}{\leq} \sigma_1\left(\gradf(\X_t)\U_t\right) +\left(\gradf(\X_t)Q_{\U_t}Q_{\U_t}^\dagger\right) \\ 
&\quad \quad \quad \quad \quad \quad \quad \quad ~+ \sigma_1\left(\gradf(\X_t)Q_{\Uo}Q_{\Uo}^\dagger\right) \sigma_1\left(\Delta\right) \\
&\stackrel{(iv)}{\leq} \sigma_1\left(\gradf(\X_t)\U_t\right) +\left(\gradf(\X_t)Q_{\U_t}Q_{\U_t}^\dagger\right) \\
&\quad \quad \quad \quad \quad \quad \quad \quad ~+ \sigma_1\left(\gradf(\X_t)Q_{\Uo}Q_{\Uo}^\dagger\right)\tfrac{1}{200} \sigma_{r}(\Uo)
\end{align*}
\end{footnotesize}
\begin{footnotesize}
\begin{align*}
&\stackrel{(v)}{\leq} \sigma_1\left(\gradf(\X_t)\U_t\right) +  \tfrac{1}{\left(1-\frac{1}{200}\right)} \cdot \tfrac{1}{200}\sigma_1\left(\gradf(\X_t)\U_t\right) \\ 
&\quad \quad \quad \quad \quad \quad \quad \quad + \tfrac{1}{200} \sigma_1\left(\gradf(\X_t)\Uo\right) \\
&\leq \tfrac{200}{199} \cdot \sigma_1\left(\gradf(\X_t)\U_t\right) + \tfrac{1}{200} \sigma_1\left(\gradf(\X_t)\Uo\right).
\end{align*} 
\end{footnotesize}
where 
$(i)$ is due to triangle inequality on $\Uo R_{\U_t}^\star = \U_t - \Delta$, 
$(ii)$ is due to generalized Cauchy-Schwarz inequality, 
$(iii)$ is due to triangle inequality and the fact that the column span of $\Delta$ can be decomposed into the column span of $\U_t$ and $\Uo$, by construction of $\Delta$, 
$(iv)$ is due to the assumption $ \dist(\U_t, \Uo) \leq \gamma' \cdot \sigma_r(\Uo)$ and 
$$\sigma_1\left(\Delta\right) \leq \dist(\U_t, \Uo) \leq \tfrac{1}{200} \tfrac{\sigma_r(\Xo)}{\sigma_1(\Xo)} \cdot \sigma_r(\Uo) \leq \tfrac{1}{200} \cdot \sigma_r(\Uo).$$ 
Finally, $(v)$ is due to the facts: 
\begin{align*}
\sigma_1\left(\gradf(\X_t) \Uo\right) &= \sigma_1\left(\gradf(\X_t) Q_{\Uo}Q_{\Uo}^\dagger \Uo\right) \\ 
												   &\geq \sigma_1\left(\gradf(\X_t) Q_{\Uo}Q_{\Uo}^\dagger\right) \cdot \sigma_r(\Uo)
\end{align*} 
and 
\begin{align*}
\sigma_1\left(\gradf(\X_t) \U_t\right) &= \sigma_1\left(\gradf(\X_t) Q_{\U_t}Q_{\U_t}^\dagger \U\right) \\ 
													&\geq \sigma_1\left(\gradf(\X_t) Q_{\U_t}Q_{\U_t}^\dagger\right) \cdot \sigma_r(\U_t) \nonumber \\ 
&\geq \sigma_1\left(\gradf(\X_t) Q_{\U_t}Q_{\U_t}^\dagger\right) \cdot \left(1 - \tfrac{1}{200}\right) \cdot \sigma_r(\Uo),
\end{align*} by the proof of (a variant of) Lemma A.3 in \cite{bhojanapalli2016dropping}. 
Thus, for the term $ \sigma_1\left(\gradf(\X_t)Q_{\Uo} Q_{\Uo}^\dagger\right)$, we have
\begin{align}\label{proofsjc:eq_067}
 \sigma_1\left(\gradf(\X_t)Q_{\Uo} Q_{\Uo}^\dagger\right)  &\leq \tfrac{1}{\sigma_r(\Uo)} \sigma_1\left(\gradf(\X_t)\Uo\right) \nonumber \\ 
 &\leq \tfrac{1}{\sigma_r(\Uo)}\tfrac{201}{199} \sigma_1\left(\gradf(\X_t)\U_t\right) \nonumber \\ 
 &\leq \tfrac{201 \sigma_1(\Uo)}{200 \sigma_r(\Uo)}\tfrac{201}{199} \sigma_1\left(\gradf(\X_t)Q_{\U_t} Q_{\U_t}^\dagger\right).
\end{align} 

Using \eqref{proofsjc:eq_067} in \eqref{proofsr1:eq_11}, we obtain:
\begin{footnotesize}
\begin{align}
&\ip{\gradf(\X_t) }{ \Delta \Delta^\dagger} \nonumber \\ 
&\geq - \Bigg( \sigma_1\left(Q_{\U_t} Q_{\U_t}^\dagger\gradf(\X_t)\right) \nonumber \\
&\quad \quad \quad \quad +  \tfrac{201 \sigma_1(\Uo)}{200 \sigma_r(\Uo)}\tfrac{201}{199}  \sigma_1\left(Q_{\U_t} Q_{\U_t}^\dagger\gradf(\X_t)\right) \Bigg) \cdot  \dist(\U_t, \Uo)^2 \nonumber \\ 
&\geq - \tfrac{21 \cdot \tau(\Uo)}{10} \|Q_{\U_t} Q_{\U_t}^\dagger\gradf(\X_t)\|_2 \cdot \dist(\U_t, \Uo)^2 \nonumber
\end{align}
\end{footnotesize} 
where $\tau(\Uo) := \tfrac{\sigma_1(\Uo)}{\sigma_r(\Uo)}$.

We remind that $\widehat{\eta} =  \tfrac{1}{128(L \sigma_1(\X_t) + \sigma_1(Q_{\U_t} Q_{\U_t}^\dagger\gradf(\X_t)))} $.
Then, we have:
\begin{footnotesize}
\begin{align}
&\tfrac{21 \cdot \tau(\Uo)}{10} \cdot \sigma_1\left(Q_{\U_t} Q_{\U_t}^\dagger\gradf(\X_t)\right) \cdot \dist(\U_t, \Uo)^2 \nonumber \\ 
 &\leq \tfrac{21 \cdot \tau(\Uo)}{10} \cdot \weta \cdot 128 L \cdot \sigma_1\left(\X_t\right) \sigma_1\left(Q_{\U_t} Q_{\U_t}^\dagger\gradf(\X_t)\right) \cdot \dist(\U_t, \Uo)^2 \nonumber \\ 
 &+ \tfrac{21 \cdot \tau(\Uo)}{10} \cdot \weta \cdot 128 \cdot \sigma_1\left(Q_{\U_t} Q_{\U_t}^\dagger\gradf(\X_t)\right)^2 \cdot  \dist(\U_t, \Uo)^2 \label{eq:dff_00}
\end{align} 
\end{footnotesize}
To bound the first term on the right hand side, we observe that 
\begin{footnotesize}
\begin{align*}
\sigma_1\left(Q_{\U_t} Q_{\U_t}^\dagger\gradf(\X_t)\right) \leq \tfrac{\mu \sigma_r(\X_t)}{\tfrac{21 \cdot \tau(\Uo)}{10} \cdot 10}
\end{align*}
\end{footnotesize}
or
\begin{footnotesize}
\begin{align*}
\sigma_1\left(Q_{\U_t} Q_{\U_t}^\dagger\gradf(\X_t)\right) \geq \tfrac{\mu \sigma_r(\X_t)}{\tfrac{21 \cdot \tau(\Uo)}{10} \cdot 10}.
\end{align*}
\end{footnotesize}
We use this trivial information to obtain \eqref{eq:figure2}
\begin{figure*}[!t]
\begin{footnotesize}
\begin{align}
&\tfrac{21 \cdot \tau(\Uo)}{10} \cdot \weta \cdot 128 L \cdot \sigma_1\left(\X_t\right) \sigma_1\left(Q_{\U_t} Q_{\U_t}^\dagger\gradf(\X_t)\right) \cdot \dist(\U_t, \Uo)^2 \nonumber \\ 
& \quad \quad \quad \leq \max \bigg\{\tfrac{\tfrac{21 \cdot \tau(\Uo)}{10} \cdot 128 \cdot \weta \cdot L \sigma_1\left(\X_t\right) \cdot \mu \sigma_r(\X_t)}{\tfrac{21 \cdot \tau(\Uo)}{10} \cdot 10} \cdot \dist(\U_t, \Uo)^2, ~~\weta \left(\tfrac{21 \cdot \tau(\Uo)}{10}\right)^2 \cdot 128 \cdot 10 \kappa \tau(\X_t) \sigma_1\left(Q_{\U_t} Q_{\U_t}^\dagger\gradf(\X_t)\right)^2 \cdot  \dist(\U_t, \Uo)^2 \bigg\} \nonumber \\
&\quad \quad \quad \leq  \tfrac{128 \cdot \weta \cdot L \sigma_1\left(\X_t\right) \cdot \mu \sigma_r(\X_t)}{10} \cdot \dist(\U_t, \Uo)^2 + \weta \left(\tfrac{21 \cdot \tau(\Uo)}{10}\right)^2 \cdot 128 \cdot 10 \kappa \tau(\X_t) \sigma_1\left(Q_{\U_t} Q_{\U_t}^\dagger\gradf(\X_t)\right)^2 \cdot  \dist(\U_t, \Uo)^2 \label{eq:figure2}
\end{align} 
\end{footnotesize}
\hrulefill
\end{figure*}
where $\kappa := \tfrac{L}{\mu}$ and $\tau(\X) := \tfrac{\sigma_1(\X)}{\sigma_r(\X)}$ for a rank-$r$ matrix $\X$.
Combining \eqref{eq:figure2} with \eqref{eq:dff_00}, we obtain the expression in \eqref{eq:figure3}
\begin{figure*}[!t]
\begin{footnotesize}
\begin{align}
&\tfrac{21 \cdot \tau(\Uo)}{10} \cdot \sigma_1\left(Q_{\U_t} Q_{\U_t}^\dagger\gradf(\X_t)\right) \cdot \dist(\U_t, \Uo)^2 \nonumber \\ 
&\quad \quad \quad \stackrel{(i)}{\leq}  \tfrac{\mu \sigma_{r}(\X_t)}{10} \cdot \dist(\U_t, \Uo)^2 + \left(10 \kappa \tau(\X_t) \cdot \tfrac{21 \cdot \tau(\Uo)}{10} +1 \right)\cdot \tfrac{21 \cdot \tau(\Uo)}{10} \cdot  128 \cdot \weta \sigma_1\left(Q_{\U_t} Q_{\U_t}^\dagger\gradf(\X_t)\right)^2 \cdot \dist(\U_t, \Uo)^2  \nonumber \\
&\quad \quad \quad  \stackrel{(ii)}{\leq}  \tfrac{\mu \sigma_{r}(\X_t)}{10} \cdot \dist(\U_t, \Uo)^2 + \left(11 \kappa \tau(\Xo) \cdot \tfrac{21 \cdot \tau(\Uo)}{10} +1 \right)\cdot \tfrac{21 \cdot \tau(\Uo)}{10} \cdot  128 \cdot \weta \sigma_1\left(Q_{\U_t} Q_{\U_t}^\dagger\gradf(\X_t)\right)^2 \cdot (\rho')^2 \sigma_{r}(\Xo)   \nonumber \\
&\quad \quad \quad \stackrel{(iii)}{\leq} \tfrac{\mu \sigma_{r}(\X_t)}{10} \cdot \dist(\U_t, \Uo)^2 + \tfrac{12 \cdot 21^2}{10^2} \cdot \kappa \cdot \tau(\Xo)^2 \cdot 128 \cdot \weta \sigma_1\left(\gradf(\X_t)\U_t\right)^2 \cdot \tfrac{11 \cdot (\rho')^2}{10} \nonumber \\
&\quad \quad \quad  \stackrel{(iii)}{\leq}  \tfrac{\mu \sigma_{r}(\X_t)}{10} \cdot \dist(\U_t, \Uo)^2+ \tfrac{\weta}{5} \sigma_1\left(\gradf(\X_t)\U_t\right)^2 \label{eq:figure3}
\end{align}
\end{footnotesize}
\hrulefill
\end{figure*}
where $(i)$ follows from $\weta \leq \tfrac{1}{128 L \sigma_1\left(\X_t\right)}$, 
$(ii)$ is due to Lemma A.3 in \cite{bhojanapalli2016dropping} and using the bound $\dist(\U_t, \Uo) \leq \gamma' \sigma_{r}(\Uo)$ by the hypothesis of the lemma, 
$(iii)$ is due to $\sigma_{r}(\Xo) \leq 1.1 \sigma_{r}(\X_t)$ by Lemma A.3 in \cite{bhojanapalli2016dropping}, due to the facts $\sigma_{r}(\X_t) \sigma_1\left(Q_{\U_t} Q_{\U_t}^\dagger\gradf(\X_t)\right)^2 \leq  \|\U_t^\dagger\gradf(X_t)\|_F^2$ and $(11 \kappa \tau(\Xo) \cdot \tfrac{21 \cdot \tau(\Uo)}{10} +1) \leq 12 \kappa \tau(\Xo) \cdot \tfrac{21 \cdot \tau(\Uo)}{10}$, and $\tau(\Uo)^2 = \tau(\Xo)$. 
Finally, $(iv)$ follows from substituting $\gamma' := c \cdot \tfrac{1}{\kappa} \cdot \tfrac{1}{\tau(\Xo)}$ for $c = \tfrac{1}{200}$ and using Lemma A.3 in \cite{bhojanapalli2016dropping} (due to the factor $\tfrac{1}{200}$, all constants above lead to bounding the term with the constant $\tfrac{1}{5}$).

Thus, we can conclude:
\begin{footnotesize}
\begin{align*}
\ip{\gradf(\X_t) }{ \Delta \Delta^\dagger} \geq - \left(\tfrac{\weta}{5} \|\gradf(\X_t) \U_t \|_F^2 + \tfrac{\mu \sigma_{r}(\Xo)}{10} \cdot \dist(\U_t, \Uo)^2\right).
\end{align*}
\end{footnotesize}
This completes the proof.
\end{proof}

\subsection{Proof of Corollary \ref{cor:projFGD_schatten}}{\label{sec:cor:xi}}

We have
\begin{align*}
\| \widetilde{\U}_{t+1} \|_F
&\le \| \U_t \|_F + \widehat{\eta} \cdot \| \gradf(\X_t) \U_t \|_F \\
&\le \| \U_t \|_F + \widehat{\eta} \cdot \sigma_1\left( \gradf(\X_t) Q_{\U_t} Q_{\U_t}^\dagger \right) \cdot \| \U_t \|_F \\
&\le \left(1 + \widehat{\eta} \cdot \sigma_1\left( \gradf(\X_t) Q_{\U_t} Q_{\U_t}^\dagger\right)\right) \\
&\le (1 + \tfrac{1}{128})
\end{align*}
where the first inequality follows from the triangle inequality, 
the second holds by the property $\|AB\|_F \le \sigma_1(A) \cdot \|B\|_F$, 
and the third follows because the step size is bounded above by $\widehat{\eta} \le \frac{1}{128 \sigma_1(\gradf(\X_t)Q_{\U_t} Q_{\U_t}^\dagger)}$. 
Hence, we get $\xi(\widetilde{\U}_{t+1}) = \frac{1}{\| \widetilde{\U}_{t+1} \|_F} \ge \frac{128}{129}$.

%\newpage
%!TEX root = QST_arxiv_v1.tex

\section{Initialization}{\label{sec:init}}
In this section, we present a specific initialization strategy for \texttt{ProjFGD}. 
For completeness, we repeat the definition of the optimization problem at hand, both in the original space:
\begin{equation}{\label{init:eq_01}}
\begin{aligned}
	& \underset{\X \in \mathbb{C}^{d \times d}}{\text{minimize}}
	& & f(\X) \quad \quad \text{subject to} \quad \X \in \C'.
\end{aligned}
\end{equation} 
and the factored space:
\begin{equation}{\label{init:eq_00}}
\begin{aligned}
	& \underset{\U \in \mathbb{C}^{d \times r}}{\text{minimize}}
	& & f(\U\U^\dagger) \quad \quad \text{subject to} \quad \U \in \C.
\end{aligned}
\end{equation} 
For our initialization, we restrict our attention to the full rank ($r = d$) case; the case of $r < d$ assumes that there is a projection step that projects at the same time onto the PSD cone and $\C$ at the same time.
In the full rank case, $\C'$ is a convex set and  includes the full-dimensional PSD cone, as well as other norm constraints, as described in the main text.
In the particular case of QST, we make no restrictions; \citep{gonccalves2016projected} provides an efficient projection procedure that satisfies the constraints $\C'$ and holds for any $r$.

Let us denote $\Pi_{\C'}(\cdot)$ the corresponding projection step, where all constraints are satisfied simultaneously.
Then, the initialization we propose follows similar motions with that in \cite{bhojanapalli2016dropping}:
We consider the projection of the weighted negative gradient at $0$, \emph{i.e.}, $-\tfrac{1}{L} \cdot \nabla f(0)$, onto $\C'$. \emph{I.e.}, 
\begin{equation}
\X_0 = \U_0 \U_0^\dagger = \Pi_{\C'}\left(\tfrac{-1}{L} \cdot \nabla f(0) \right).\label{eq:projinit}
\end{equation}
Assuming a first-order oracle model, where we access $f$ only though function evaluations and gradient calculations, \eqref{eq:projinit} provides a cheap way to find an initial point with some approximation guarantees as follows \footnote{As we show in the experiments section, a random initialization performs well in practice, without requiring the additional calculations involved in \eqref{eq:projinit}. However, a random initialization for the constraint case provides no guarantees whatsoever.}:

\begin{lemma}
Let $\U_0 \in \mathbb{C}^{d \times r}$ be such that $\X_0 = \U_0\U_0^\dagger = \Pi_{\C'}\left(\tfrac{-1}{L} \cdot \nabla f(0) \right)$. 
Consider the problem in \eqref{init:eq_00} where $f$ is assumed to be $L$-smooth and $\mu$-strongly convex, with optimum point $\X_\star$ such that $\text{rank}(\X_\star) = r$. 
We apply \texttt{ProjFGD} with $\U_0$ as the initial point. 
Then, in this generic case, $\U_0$ satisfies:
\begin{align*}
\dist(\U_0, \Uo) \leq \gamma' \cdot \sigma_r(\Uo),
\end{align*} where $\gamma' = \sqrt{\tfrac{1 - \sfrac{\mu}{L}}{2(\sqrt{2}-1)}} \tau^2(\Uo) \sqrt{\texttt{srank}(\Xo)}$, $\texttt{srank}(\X) = \tfrac{\|\X\|_F}{\sigma_1(\X)}$.
\end{lemma}

\begin{proof}
To show this, we start with:
\begin{align}
\|\X_0 -\Xo\|_F^2 = \|\Xo\|_F^2 + \|\X_0\|_F^2 -2\ip{\X_0}{\Xo}\label{eq:projinit0}.  
\end{align}
Recall that $\X_0 = \U_0\U_0^\dagger = \Pi_{\C'}\left(\tfrac{-1}{L} \cdot \nabla f(0) \right)$ by assumption, where $\Pi_{\C'}(\cdot)$ is a convex projection.
Then, by Lemma \ref{lem:proj}:
\begin{align}
\ip{\tfrac{-1}{L}\nabla f(0)}{\X_0 - \rhoo} \geq \ip{\X_0}{\X_0 - \Xo}. \label{eq:projinit1}
\end{align}

Observe that $0 \in \mathbb{C}^{d \times d}$ is a feasible point, since it is PSD and satisfy any common \emph{symmetric} norm constraints, as the ones considered in this paper. 
Hence, using strong convexity of $f$ around $0$, we get,
\begin{align}
    f(\Xo)-\frac{\mu}{2}\|\Xo\|_F^2 &\geq f(0)+\ip{\gradf(0)}{\Xo} \nonumber \\
    &\stackrel{(i)}{=} f(0)+\ip{\gradf(0)}{\X_0}+\ip{\gradf(0)}{\Xo - \X_0} \nonumber \\
    &\stackrel{(ii)}{\geq} f(0)+\ip{\gradf(0)}{\X_0}+\ip{L\cdot \X_0}{\X_0 - \Xo}.\label{eq:projinit2}
\end{align} 
where $(i)$ is by adding and subtracting $\ip{\nabla f(0)}{\X_0}$, and
$(ii)$ is due to \eqref{eq:projinit1}.
Further, using the smoothness of $f$ around $0$, we get:
\begin{align}
    f(\X_0) &\leq f(0) + \ip{\gradf(0)}{\X_0} + \tfrac{L}{2}\|\X_0\|_F^2 \nonumber\\
    &\stackrel{(i)}{\leq} f(\Xo)-\tfrac{\mu}{2}\|\Xo\|_F^2 + \ip{L \cdot \X_0}{\Xo} -\tfrac{L}{2}\|\X_0\|_F^2 \nonumber \\
    &\leq f(\X_0) - \tfrac{\mu}{2}\|\Xo\|_F^2 + \ip{L \cdot \X_0}{\Xo} - \tfrac{L}{2}\|\X_0\|_F^2. \nonumber
\end{align} where
$(i)$ follows from \eqref{eq:projinit2} by upper bounding the quantity $f(0) + \ip{\gradf(0)}{\X_0}$,
$(ii)$ follows from the assumption that $f(\Xo) \leq f(\X_0)$. 
Hence, rearranging the above terms, we get:
\begin{align*}
    \ip{\X_0}{\Xo} \geq \tfrac{1}{2}\|\X_0\|_F^2 + \tfrac{\mu}{2L}\|\Xo\|_F^2.
\end{align*}
Combining the above inequality with~\eqref{eq:projinit0}, we obtain,
\begin{align*}
    \| \X_0 - \Xo\|_F \leq  \sqrt{1-\tfrac{\mu}{L}} \cdot \|\Xo\|_F.
\end{align*}
Given, $\U_0$ such that $\X_0 = \U_0\U_0^\dagger$ and $\Uo$ such that $\Xo = \Uo \Uo^\dagger$, we use Lemma \ref{lem:tu} from \cite{tu2015low} to obtain:
\begin{align*}
\|\U_0\U_0^\dagger - \Uo\Uo^{\dagger}\|_F \geq \sqrt{2 (\sqrt{2} - 1)} \cdot \sigma_r(\Uo) \cdot \dist(\U_0, \Uo).
\end{align*} 
Thus: $\dist(\U_0, \Uo) \leq \tfrac{\|\X_0 - \Xo\|_F}{\sqrt{2(\sqrt{2}-1)} \cdot \sigma_r(\Uo)} \cdot \|\Xo\|_F \leq \gamma' \cdot \sigma_r(\Uo)$
where $\gamma' =  \sqrt{\tfrac{1 - \sfrac{\mu}{L}}{2(\sqrt{2}-1)}} \cdot \tau^2(\Uo) \cdot \sqrt{\texttt{srank}(\Xo)}$.
\end{proof}
Such initialization, while being simple, introduces further restrictions on the condition number $\tau(\Xo)$, and the condition number of function $f$. 
Finding such simple initializations with weaker restrictions remains an open problem; however, as shown in \cite{bhojanapalli2016dropping, tu2015low, chen2015fast}, one can devise specific deterministic initialization for a given application.

In practice, the projection $\Pi_{\C'}(\cdot)$ step might not be easy to compute, due to the joint involvement of convex sets. 
A practical solution would be to sequentially project $-\tfrac{1}{L} \cdot \nabla f(0)$ onto the individual constraint sets.
Let $\Pi_{+}(\cdot)$ denote the projection onto the PSD cone.
Then, we can consider the approximate point:
%\begin{align*}
$\widetilde{\X}_0 = \widetilde{\U}_0 \widetilde{\U}_0^\dagger = \Pi_{+} \left( \widetilde{\X}_0 \right)$.
%\end{align*} 
Given $\widetilde{\U}_0$, we can perform an additional step:
%$$ 
$\U_0 = \Pi_{\C}\left(\widetilde{\U}_0\right)$,
%$$
to guarantee that $\U_0 \in \C$.
In the special case of QST, one can use the procedure in \citep{gonccalves2016projected}.

\vspace{0.2cm}
\section{Additional experiments and pseudocode}{\label{sec:add_exp}}
Figures \ref{fig:app_exp3}-\ref{fig:app_exp4} show further results regarding the QST problem, where $r = 1$ and $n = 10, 12$, respectively. 
For each case, we present both the performance in terms of number of iterations needed, as well as what is the cumulative time required.
In the case of $n = 12$, for all algorithms, we use as initial point $U_0 = \Pi_{\C}(\widetilde{\U_0})$ such that $\X_0 = \widetilde{\U}_0 \widetilde{\U}_0^\dagger$ where $\X_0 = \Pi_{+}\left(-\mathcal{A}^*(y)\right)$ and $\Pi_{+}(\cdot)$ is the projection onto the PSD cone.
In the case of $n = 10$, we use random initialization.
Configurations are described in the caption of each figure. 
Table \ref{tbl:Comp} contains information regarding total time required for convergence and quality of solution for some of these cases.
Results on almost pure density states, \emph{i.e.}, $r > 1$, are provided in Figure \ref{fig:app_exp5}. 

Next, we also provide pseudocode for our approach. 
Input arguments for $\texttt{ProjFGD}$ is $(i)$ $\text{\texttt{f\_grad}}$ that specifies the gradient operator; in our case, we set 
\begin{lstlisting}[basicstyle=\fontsize{6}{6}\ttfamily]
f_grad = @(rho) -Mt(y - M(rho));
\end{lstlisting}
where \texttt{M} denotes the forward linear operator over Pauli observables, and \texttt{Mt} its adjoint.
And, $(ii)$ $\texttt{params}$, a Matlab structure that contains several hyperparameters.
Here, \texttt{params} contains 
\texttt{.d}, the dimension $d$; 
\texttt{.r}, the rank of the density matrix; 
\texttt{.init}, the choice between random or specific initialization; 
\texttt{.Ainit}, the random initial $A_0$ in case $\texttt{params.init} = 0$; 
\texttt{.stepsize} that selects a conservative (theory) versus a practical step size; 
\texttt{.iter}, the maximum number of iterations;  
\texttt{.tol}, the tolerance for the stopping criterion.
%\newpage

\begin{lstlisting}[basicstyle=\fontsize{6}{6}\ttfamily]
function [rhohat, Ahat] = ProjFGD(f_grad, params)

options.tol = 10^-3; d = params.d; r = params.r;

% Random initialization
if (params.init == 0) 
    Acur = params.Ainit; 
    Acur = Acur./norm(Acur, 'fro');
    rhocur = Acur * Acur';
    rhoprev = rhocur;
    
    % Use of PROPACK
    [~, S1, ~] = lansvd(Acur, 1, 'L', options);  
    norm_grad = f_grad(rhocur);
    [~, S2, ~] = lansvd(norm_grad, 1, 'L', options);  
    
    grad0 = f_grad(zeros(d, d));
    grad1 = f_grad(ones(d, d));
    L_est = 2*norm(grad0 - grad1, 'fro')/d;
% Our initialization    
elseif (params.init == 1)   
    % Compute 1/L * gradf(0)
    grad0 = f_grad(zeros(d, d));
    grad1 = f_grad(ones(d, d));
    L_est = 2*norm(grad0 - grad1, 'fro')/d;
    rhocur = -(1/L_est) * grad0;
    
    % P_r(gradf(0)) using PROPACK
    [Ar, Sr, ~] = lansvd(rhocur, r, 'L', options);
    Acur = Ar * sqrt(Sr); 
    Acur = Acur./norm(Acur, 'fro');
    rhocur = Acur * Acur';
    rhoprev = rhocur;
    
    S1 = Sr;
    norm_grad = f_grad(rhocur);
    [~, S2, ~] = svds(norm_grad, 1);      
end

if (params.stepsize == 0) % Theory
    eta = 1/( 128 * (L_est * S1(1,1) + S2(1,1)) );
elseif (params.stepsize == 1) % More practical
    eta = 1/( 10 * L_est * S1(1,1) + S2(1,1));
end;

i = 1;
while (i <= params.iters)    
    f_gradrho = f_grad(rhocur);
    gradA = f_gradrho * Acur;                
    Acur = Acur - eta * gradA;        
    
    norm_Acur = norm(Acur, 'fro')    
    if (norm_Acur > 1)
	    Acur = Acur./norm(Acur, 'fro');
    end;
    
    rhocur = Acur * Acur';
    
    % Test stopping criterion
    if ((i > 1) && (norm(rhocur - rhoprev, 'fro') < params.tol * norm(rhocur, 'fro')))
        break;
    end
    i = i + 1;   rhoprev = rhocur;        
end

rhohat = Acur * Acur';
Ahat = Acur;
\end{lstlisting}

\begin{figure*}[h!]
\centering
\includegraphics[width=0.32\textwidth]{./pureQST_q_12_C_3_init_ours_XcurvesIter} 
\includegraphics[width=0.32\textwidth]{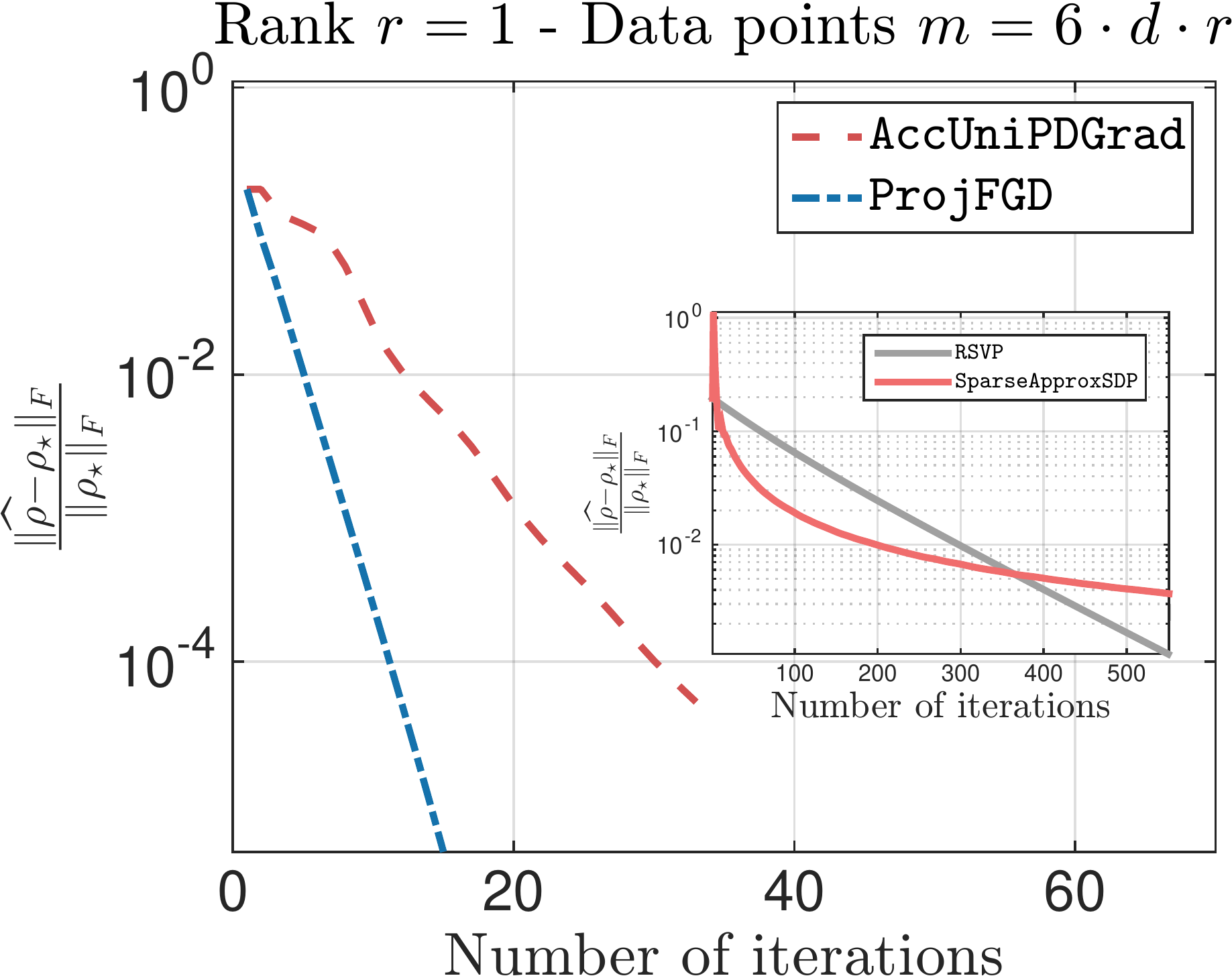} 
\includegraphics[width=0.32\textwidth]{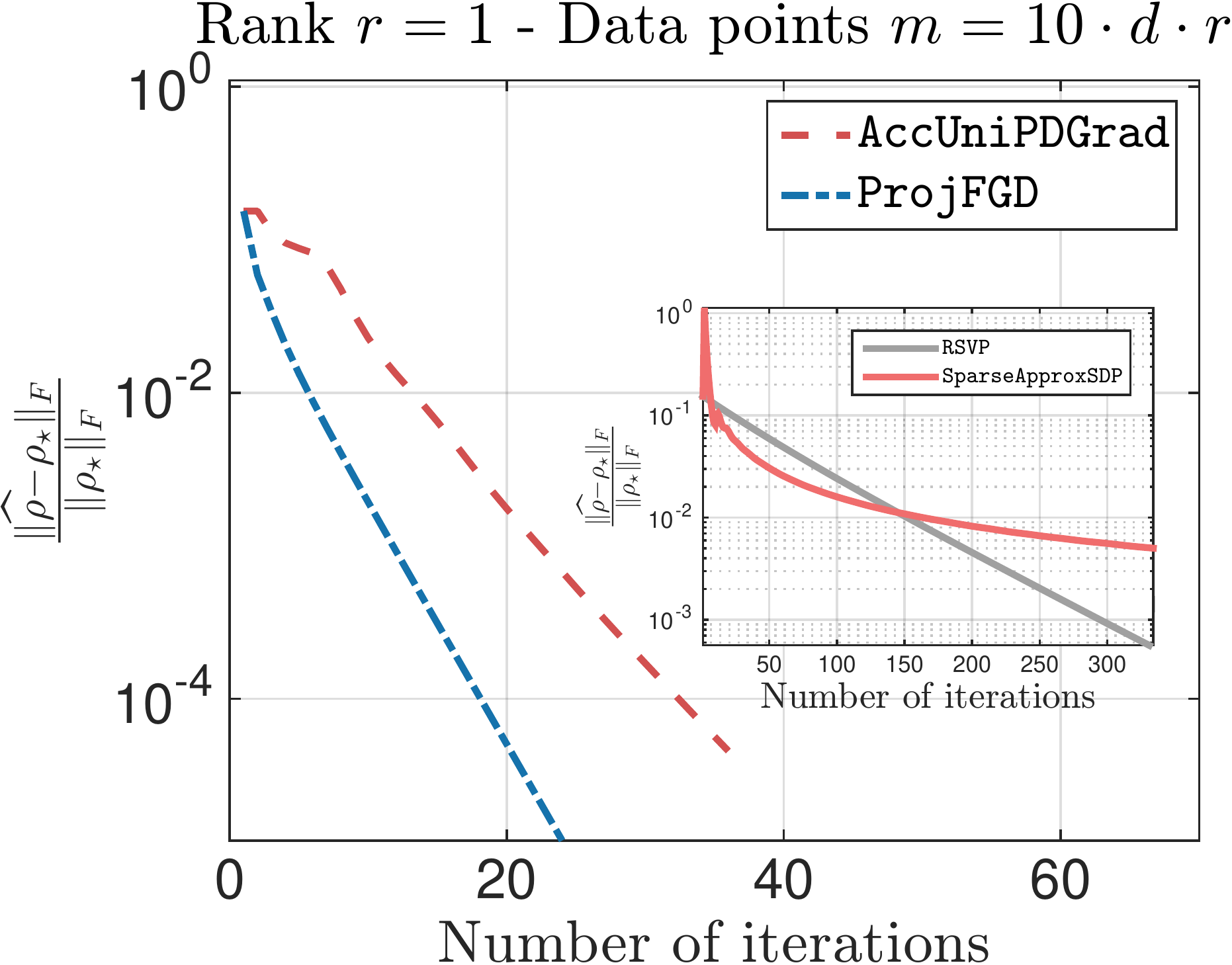} \\
\vspace{0.3cm}
\includegraphics[width=0.32\textwidth]{./pureQST_q_12_C_3_init_ours_XcurvesTime} 
\includegraphics[width=0.32\textwidth]{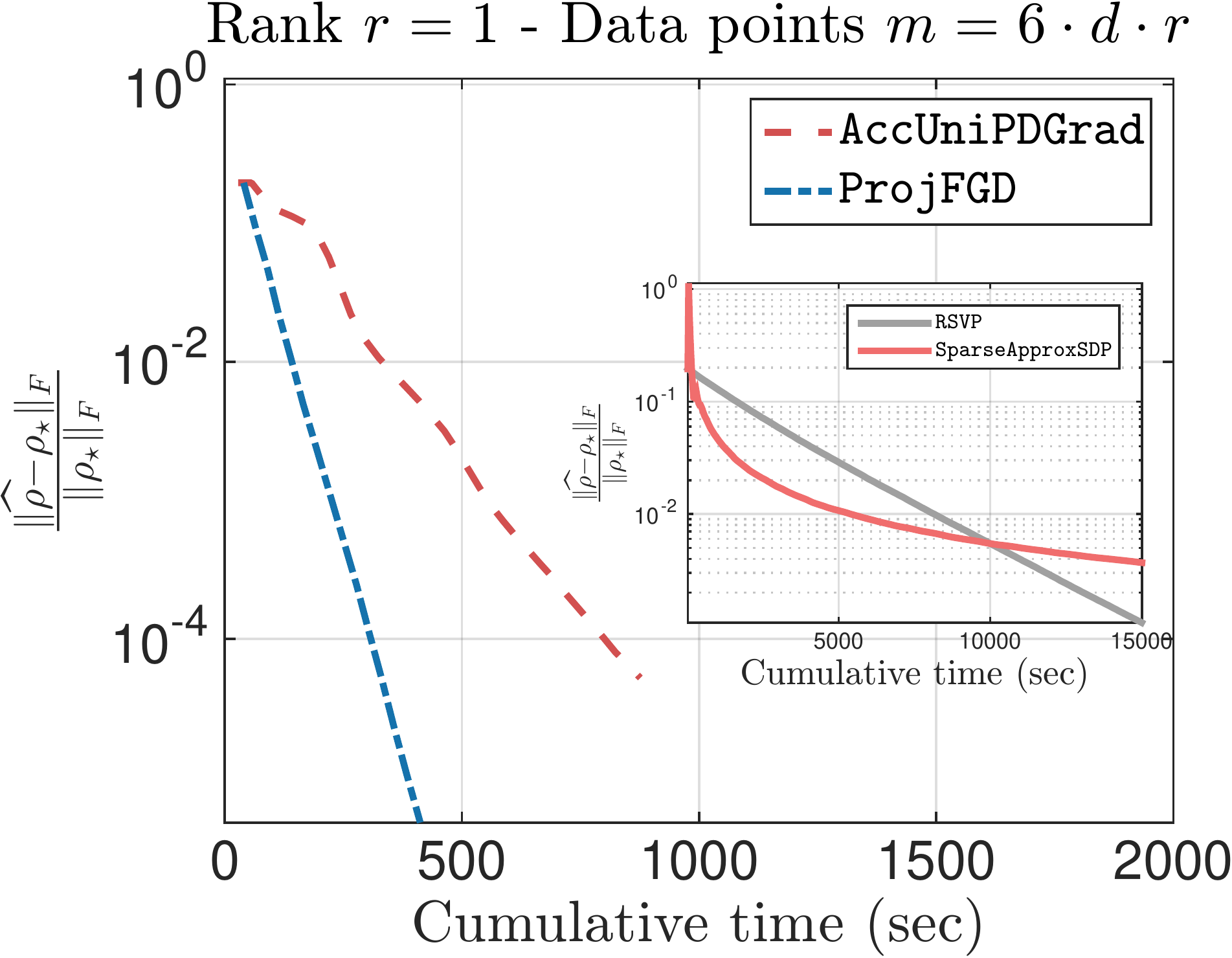} 
\includegraphics[width=0.32\textwidth]{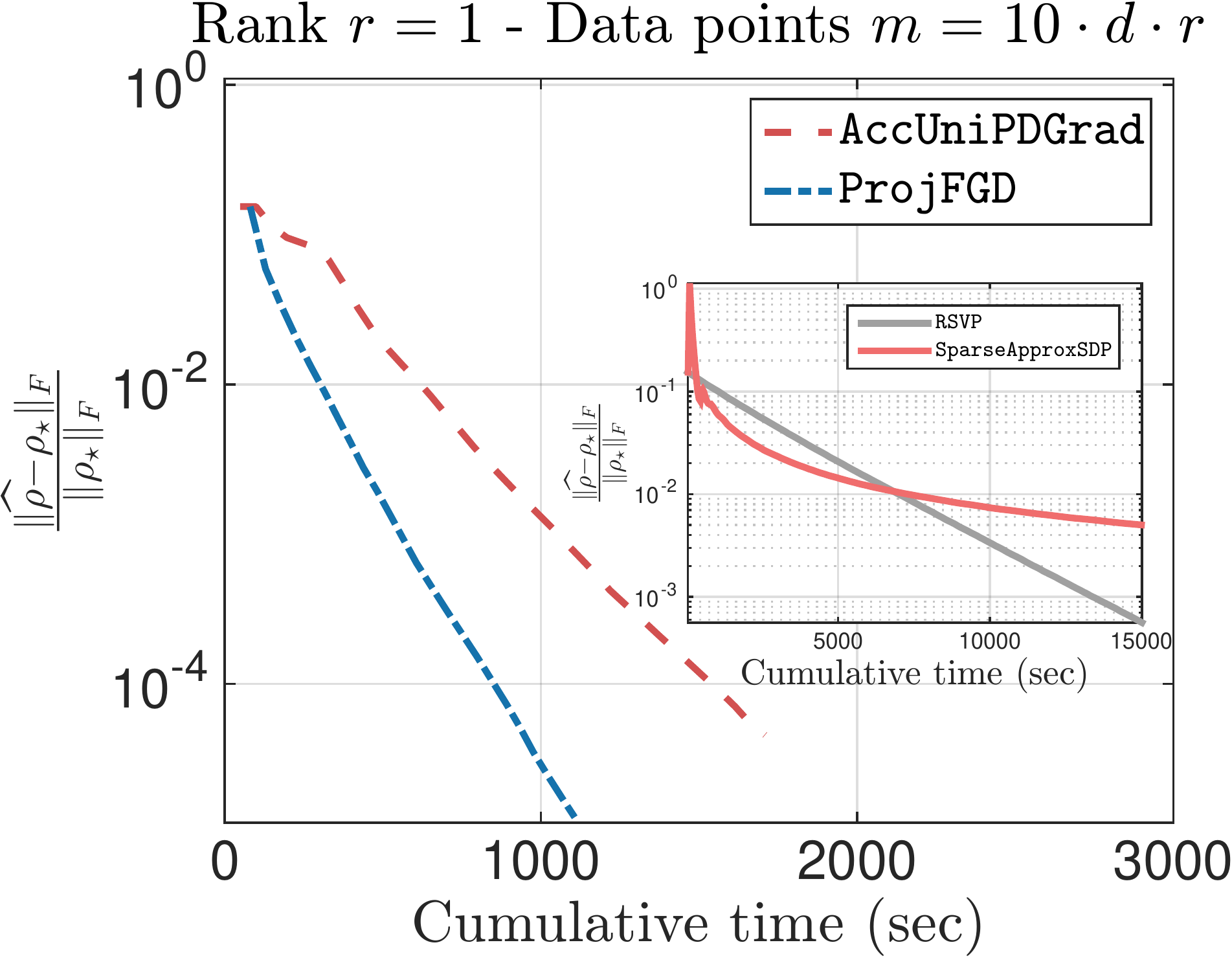} \\
\caption{Convergence performance of algorithms under comparison w.r.t. $\tfrac{\|\widehat{\rho} - \rhoo\|_F}{\|\rhoo\|_F}$ vs. $(i)$ the total number of iterations (top) and $(ii)$ the total execution time (bottom). First, second and third column corresponds to $C_{\rm sam} = 3, 6$ and $10$, respectively. For all cases, $r = 1$ (pure state setting) and $n = 12$. Initial point is $U_0 = \Pi_{\C}(\widetilde{\U_0})$ such that $\X_0 = \widetilde{\U}_0 \widetilde{\U}_0^\dagger$ where $\X_0 = \Pi_{+}\left(-\mathcal{A}^*(y)\right)$.
}
\label{fig:app_exp3}
\end{figure*}

\begin{figure*}[h!]
\centering
\includegraphics[width=0.32\textwidth]{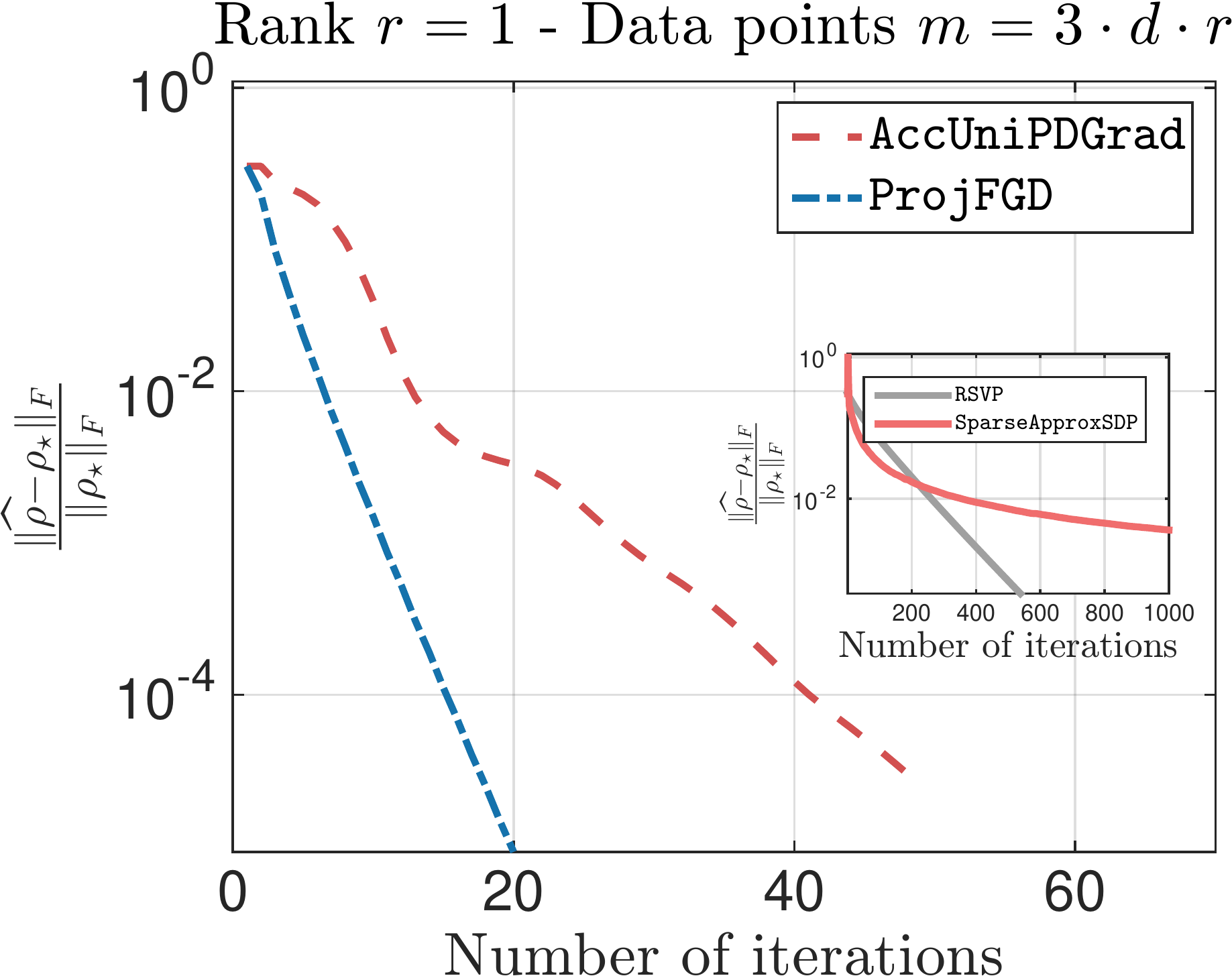} 
\includegraphics[width=0.32\textwidth]{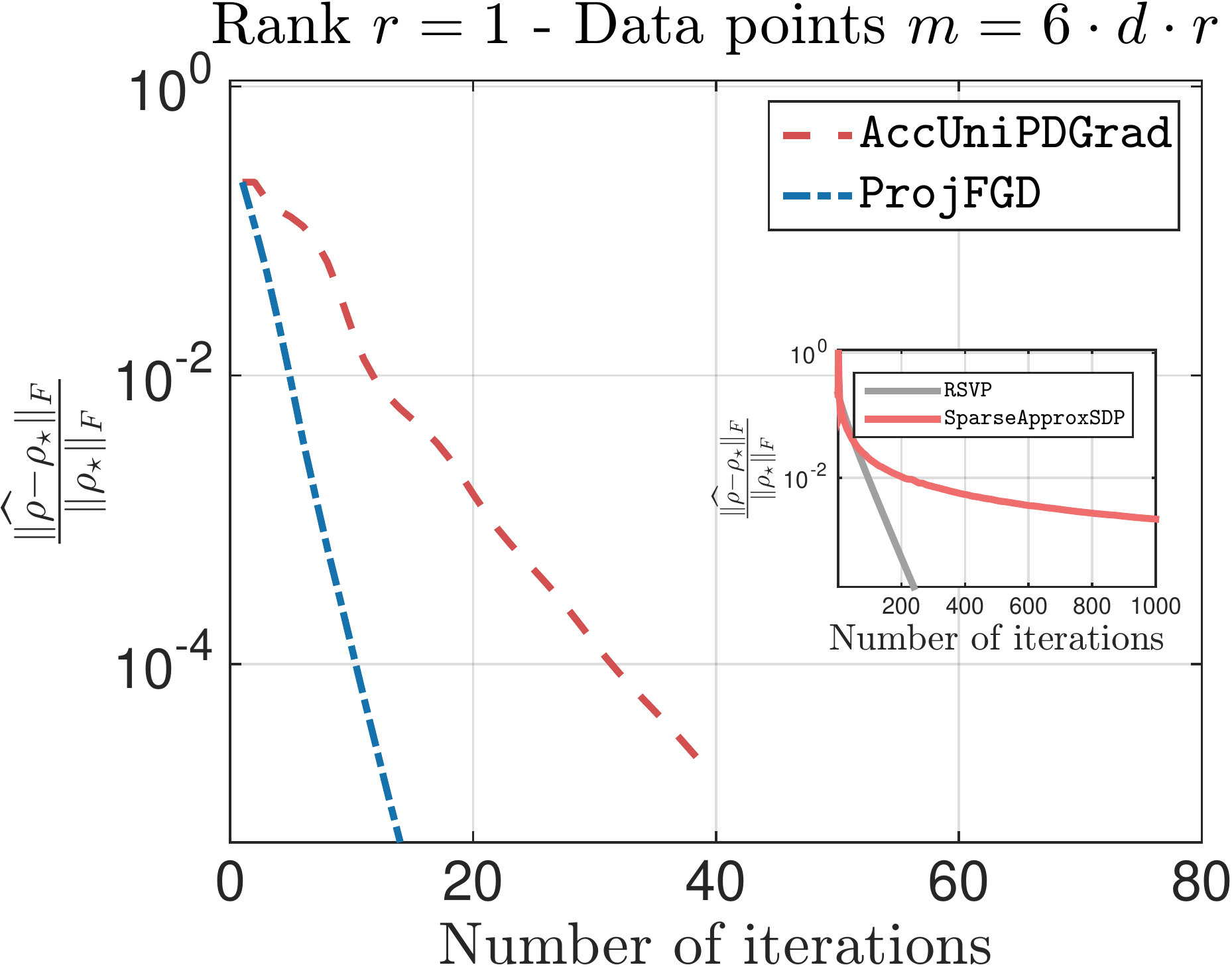} 
\includegraphics[width=0.32\textwidth]{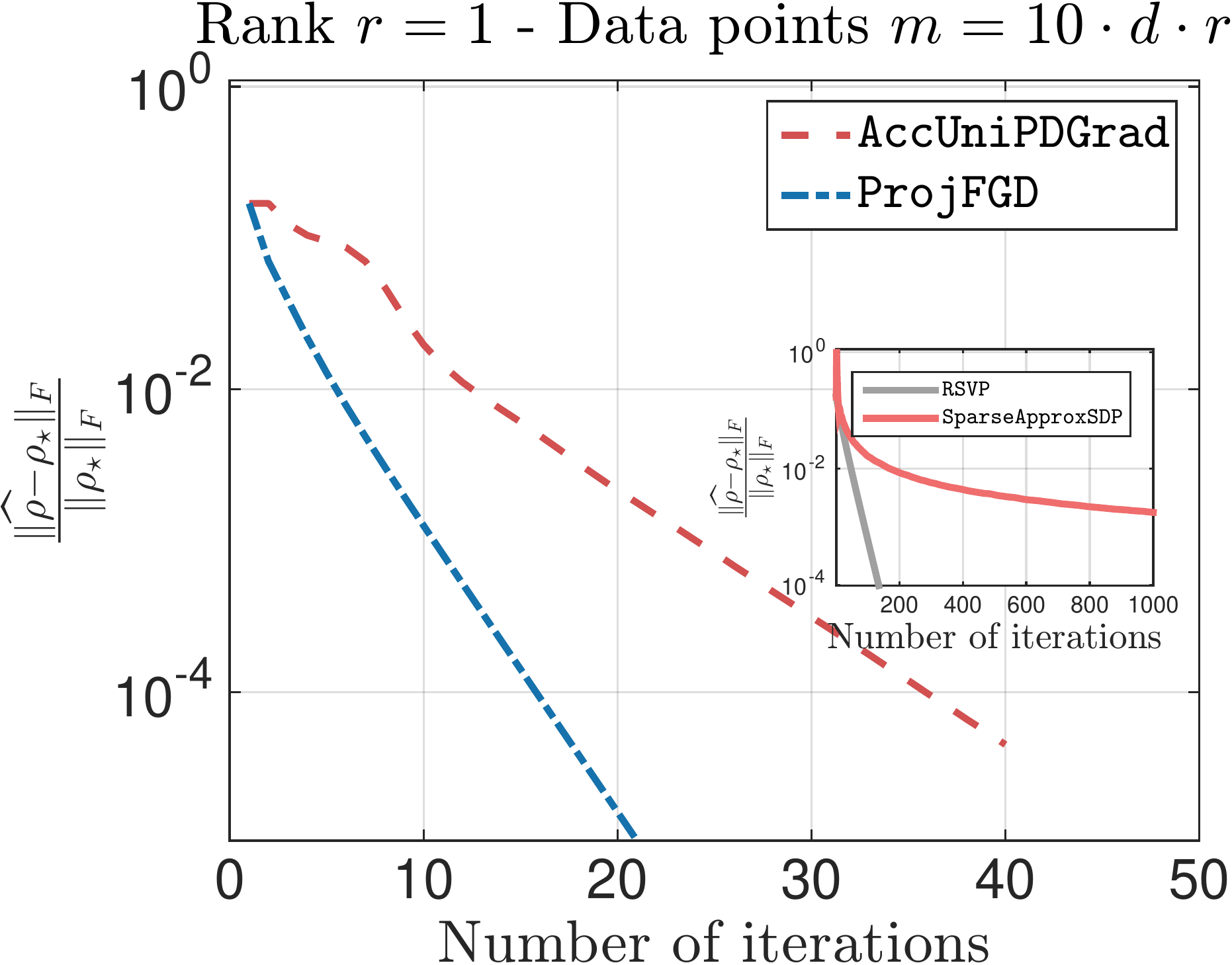} \\
\vspace{0.3cm}
\includegraphics[width=0.32\textwidth]{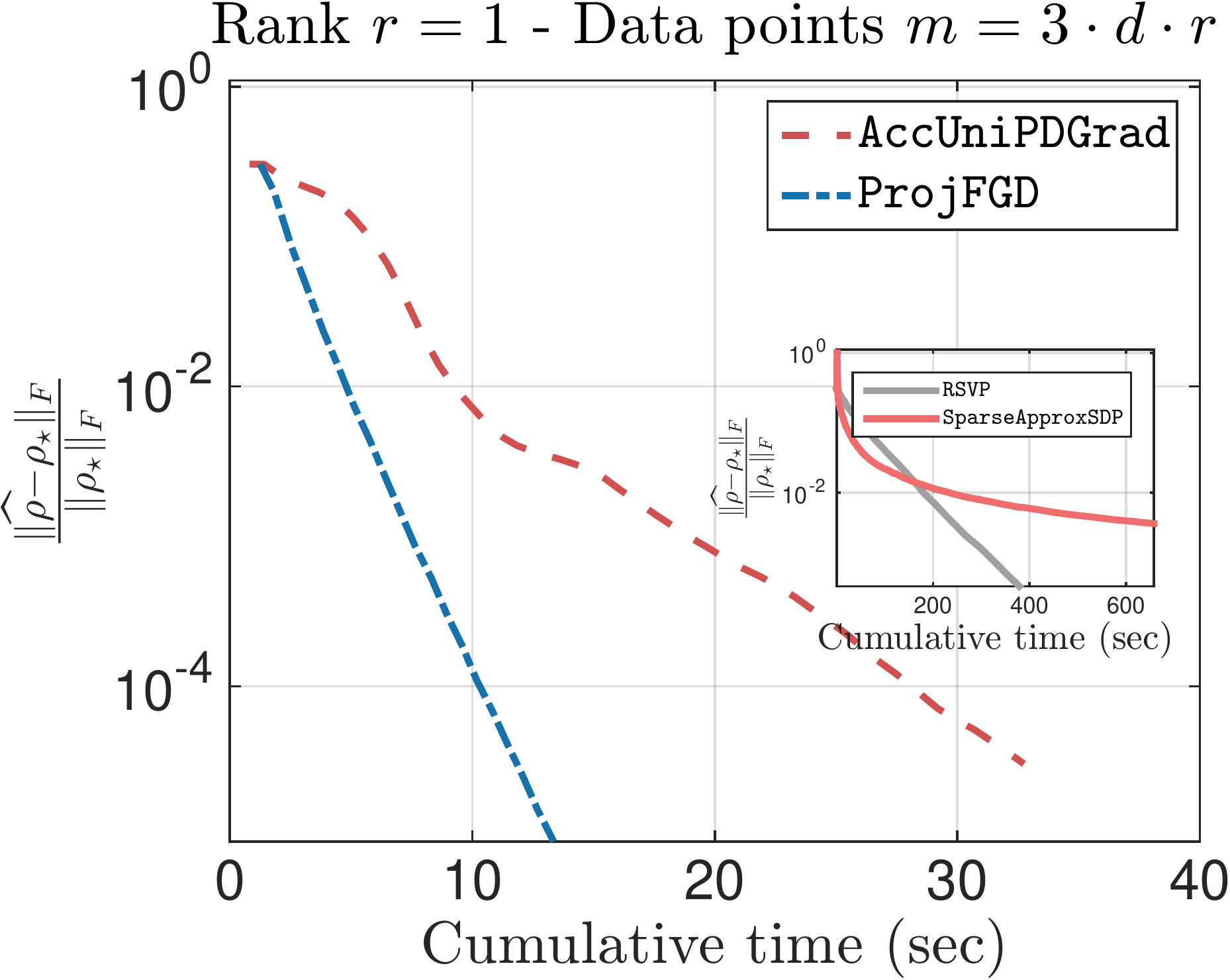} 
\includegraphics[width=0.32\textwidth]{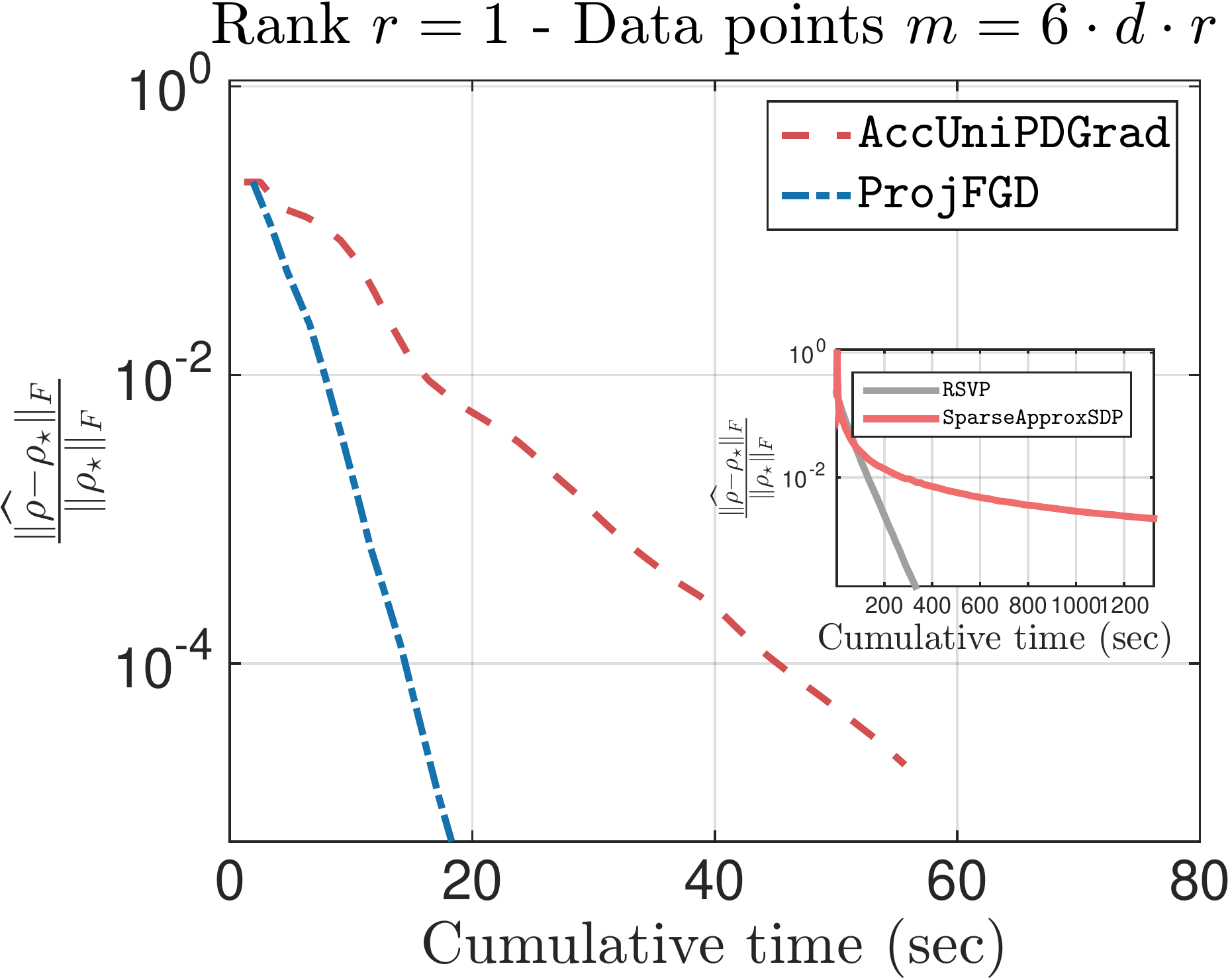} 
\includegraphics[width=0.32\textwidth]{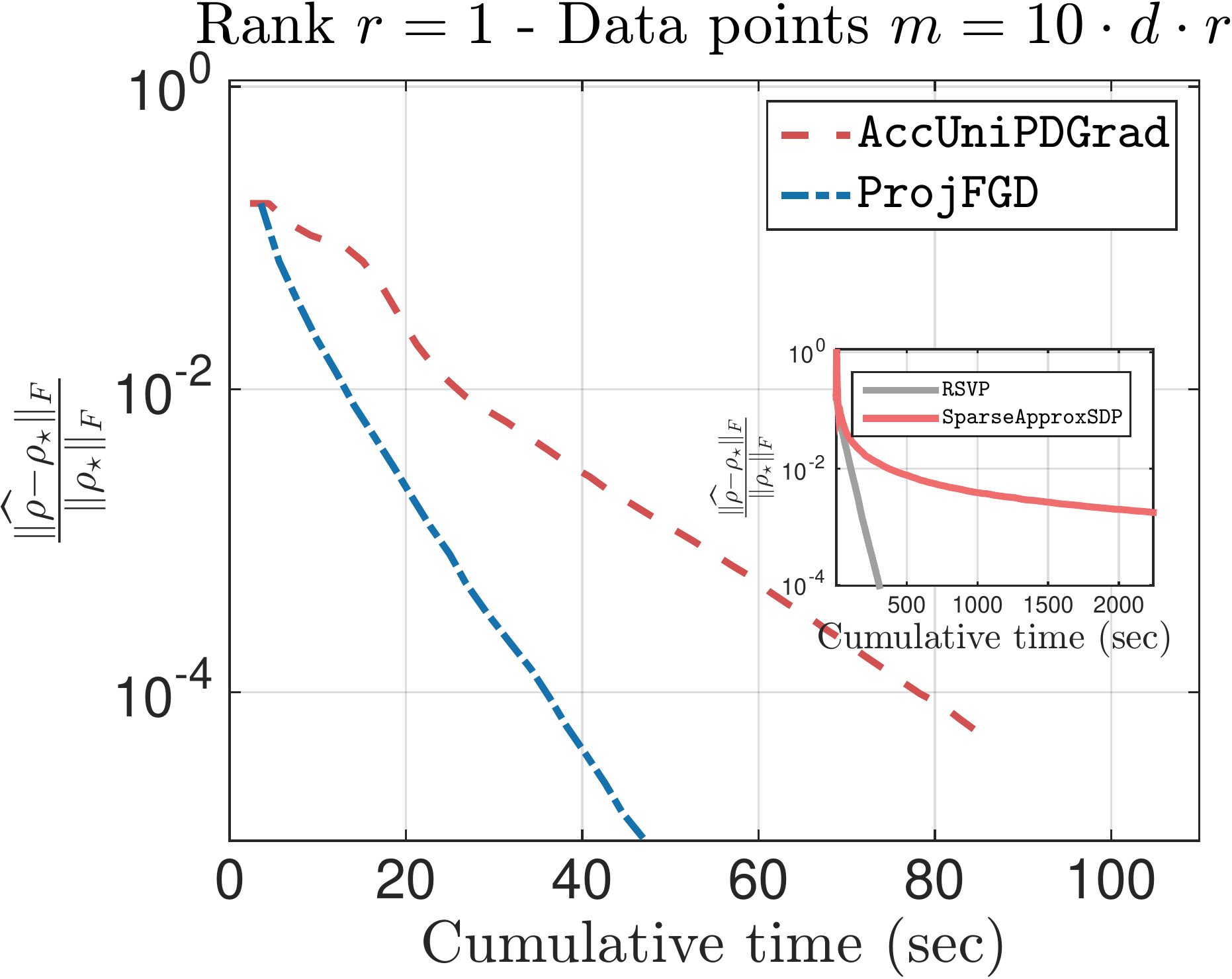} \\
\caption{Convergence performance of algorithms under comparison w.r.t. $\tfrac{\|\widehat{\rho} - \rhoo\|_F}{\|\rhoo\|_F}$ vs. $(i)$ the total number of iterations (top) and $(ii)$ the total execution time (bottom). \emph{We consider random initialization for all algorithms}. First, second and third column corresponds to $C_{\rm sam} = 3, 6$ and $10$, respectively. For all cases, $r = 1$ (pure state setting) and $n = 10$. 
}
\label{fig:app_exp4}
\end{figure*}

\begin{table*}[!h]
\centering
\begin{tabular}{c c c c c c c c c c c c c} \toprule
& \phantom{a} & \multicolumn{3}{c}{$n = 6$, $C_{\rm sam} = 3$.} & \phantom {a} & \multicolumn{3}{c}{$n = 6$, $C_{\rm sam} = 6$.} & \phantom {a} & \multicolumn{3}{c}{$n = 6$, $C_{\rm sam} = 10$.}\\
\cmidrule {3-5} \cmidrule{7-9} \cmidrule{11-13} 
Algorithm & \phantom{a} & $\tfrac{\|\widehat{\X} - \rhoo\|_F}{\|\rhoo\|_F}$ & \phantom{a}  & Total time 
			   & \phantom{a} & $\tfrac{\|\widehat{\X} - \rhoo\|_F}{\|\rhoo\|_F}$ & \phantom{a}  & Total time 
			   & \phantom{a} & $\tfrac{\|\widehat{\X} - \rhoo\|_F}{\|\rhoo\|_F}$ & \phantom{a}  & Total time  \\
\cmidrule{1-1} \cmidrule {3-3} \cmidrule{5-5} \cmidrule{7-7} \cmidrule{9-9} \cmidrule {11-11} \cmidrule{13-13}
\texttt{RSVP} & & 5.1496e-05 & & 0.7848 & & 1.8550e-05 & & 0.3791 & & 6.6328e-06 & & 0.1203 \\ 
\texttt{SparseApproxSDP} & & 4.6323e-03 & & 3.7404 & & 2.2469e-03 & & 4.3775 & & 1.4776e-03 & & 3.8536 \\ 
\texttt{AccUniPDGrad} & & 4.0388e-05 & & 0.3634 & & 2.4064e-05 & & 0.3311 & & 1.9032e-05 & & 0.4911\\ 
\texttt{ProjFGD} & & 2.4116e-05 & & 0.0599 & & 1.6052e-05 & & 0.0441 & & 1.1419e-05 & & 0.0446 \\ 
\midrule 
& \phantom{a} & \multicolumn{3}{c}{$n = 8$, $C_{\rm sam} = 3$.} & \phantom {a} & \multicolumn{3}{c}{$n = 8$, $C_{\rm sam} = 6$.} & \phantom {a} & \multicolumn{3}{c}{$n = 8$, $C_{\rm sam} = 10$.}\\
\cmidrule {3-5} \cmidrule{7-9} \cmidrule{11-13} 
 \texttt{RSVP} & & 1.5774e-04 & & 5.7347 & & 5.2470e-05 & & 3.8649 & & 2.9583e-05 & & 4.6548 \\ 
\texttt{SparseApproxSDP} & & 4.1639e-03 & & 16.1074 & & 2.2011e-03 & & 33.7608 & & 1.7631e-03 & & 85.0633 \\ 
\texttt{AccUniPDGrad} & & 3.5122e-05 & & 1.1006 & & 2.4634e-05 & & 1.8428 & & 1.7719e-05 & & 3.9440 \\ 
\texttt{ProjFGD} & & 2.4388e-05 & & 0.6918 & & 1.5431e-05 & & 0.8994 & & 1.0561e-05 & & 1.8804 \\ 
 \midrule 
 & \phantom{a} & \multicolumn{3}{c}{$n = 10$, $C_{\rm sam} = 3$.} & \phantom {a} & \multicolumn{3}{c}{$n = 10$, $C_{\rm sam} = 6$.} & \phantom {a} & \multicolumn{3}{c}{$n = 10$, $C_{\rm sam} = 10$.}\\
 \cmidrule {3-5} \cmidrule{7-9} \cmidrule{11-13} 
\texttt{RSVP} & & 4.6056e-04 & & 379.8635 & & 1.8017e-04 & & 331.1315 & & 9.7585e-05 & & 307.9554 \\ 
\texttt{SparseApproxSDP} & & 3.6310e-03 & & 658.7082 & & 2.1911e-03 & & 1326.5374 & & 1.7687e-03 & & 2245.2301 \\ 
\texttt{AccUniPDGrad} & & 3.0456e-05 & & 33.3585 & & 1.9931e-05 & & 56.9693 & & 4.5022e-05 & & 88.2965 \\ 
\texttt{ProjFGD} & & 9.2352e-06 & & 13.9547 & & 5.8515e-06 & & 19.3982 & & 1.0460e-05 & & 49.4528 \\ 
 \midrule 
 & \phantom{a} & \multicolumn{3}{c}{$n = 12$, $C_{\rm sam} = 3$.} & \phantom {a} & \multicolumn{3}{c}{$n = 12$, $C_{\rm sam} = 6$.} & \phantom {a} & \multicolumn{3}{c}{$n = 12$, $C_{\rm sam} = 10$.}\\
 \cmidrule {3-5} \cmidrule{7-9} \cmidrule{11-13} 
\texttt{RSVP} & & 4.7811e-03 & & 14029.1525 & & 1.0843e-03 & & 15028.2836 & & 5.6169e-04 & & 15067.7249 \\ 
\texttt{SparseApproxSDP} & & 3.1717e-03 & & 13635.4238 & & 3.6954e-03 & & 15041.6235 & & 5.0197e-03 & & 15051.4497 \\ 
\texttt{AccUniPDGrad} & & 8.8050e-05 & & 461.2084 & & 5.2367e-05 & & 904.0507 & & 4.5660e-05 & & 1759.6698  \\ 
\texttt{ProjFGD} & & 8.4761e-06 & & 266.8203 & & 4.7399e-06 & & 440.7193 & & 1.1871e-05 & & 1159.2885 \\ 
\bottomrule
\end{tabular}
\caption{Summary of comparison results for reconstruction and efficiency. As a stopping criterion, we used $\sfrac{\|\X_{t+1} - \X_{t}\|_2}{\|\X_{t+1}\|_2} \leq 5 \cdot 10^{-6}$, where $\X_i$ is the estimate at the $i$-th iteration. Time reported is in seconds. Initial point is $U_0 = \Pi_{\C}(\widetilde{\U_0})$ such that $\X_0 = \widetilde{\U}_0 \widetilde{\U}_0^\dagger$ where $\X_0 = \Pi_{+}\left(-\mathcal{A}^*(y)\right)$.} \label{tbl:Comp}
\end{table*}

\begin{figure*}[t]
\centering
\includegraphics[width=0.24\textwidth]{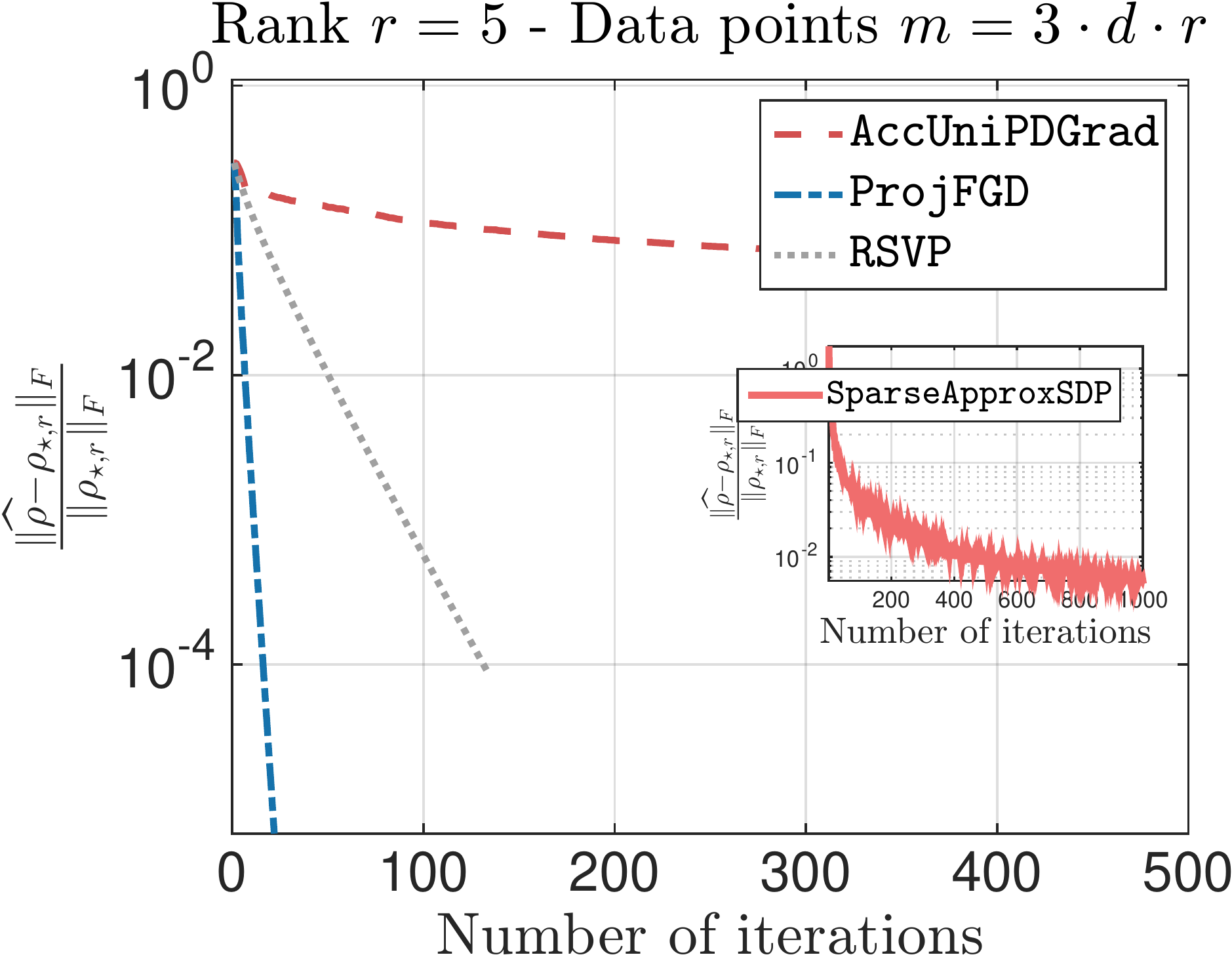} 
\includegraphics[width=0.23\textwidth]{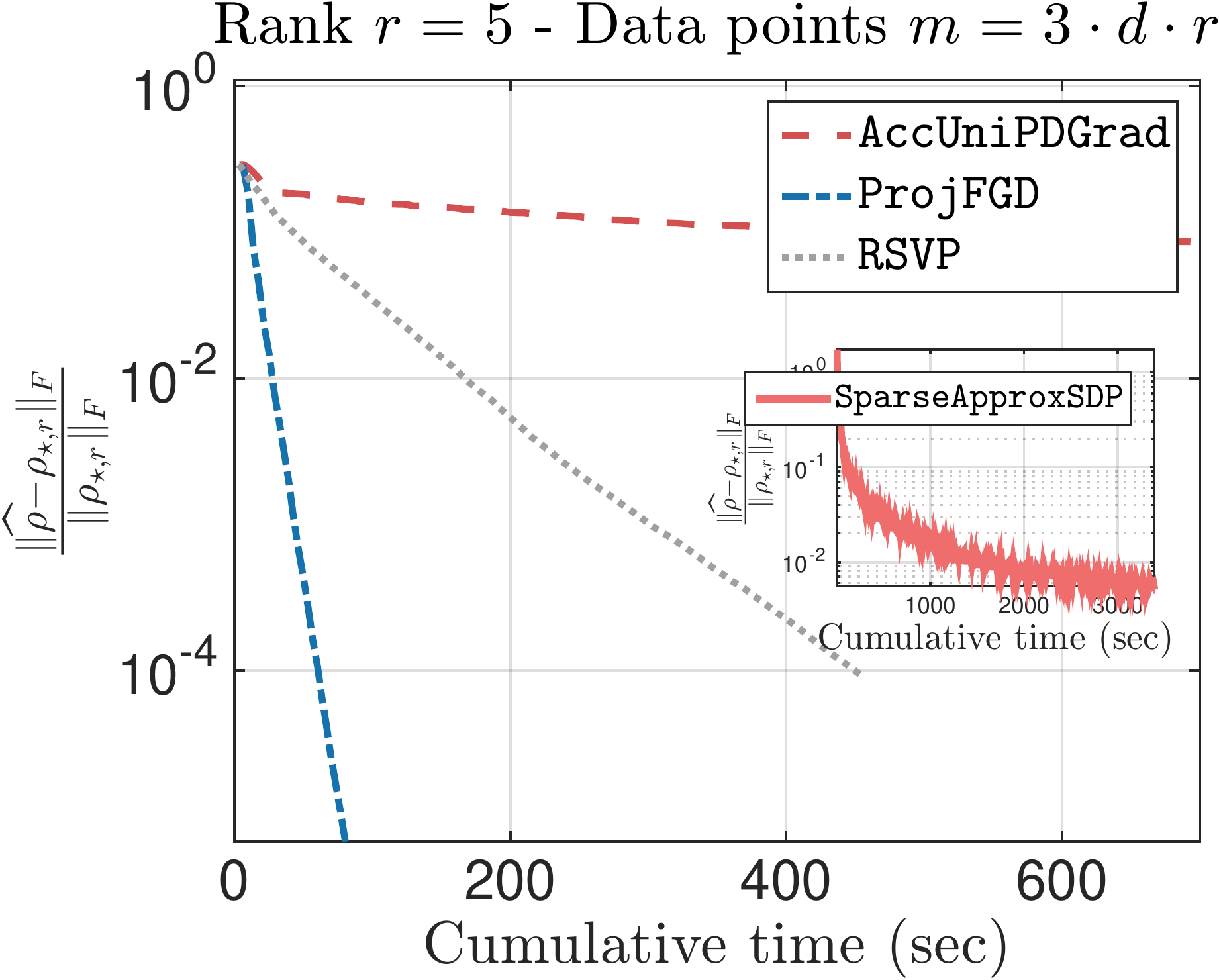}
\includegraphics[width=0.24\textwidth]{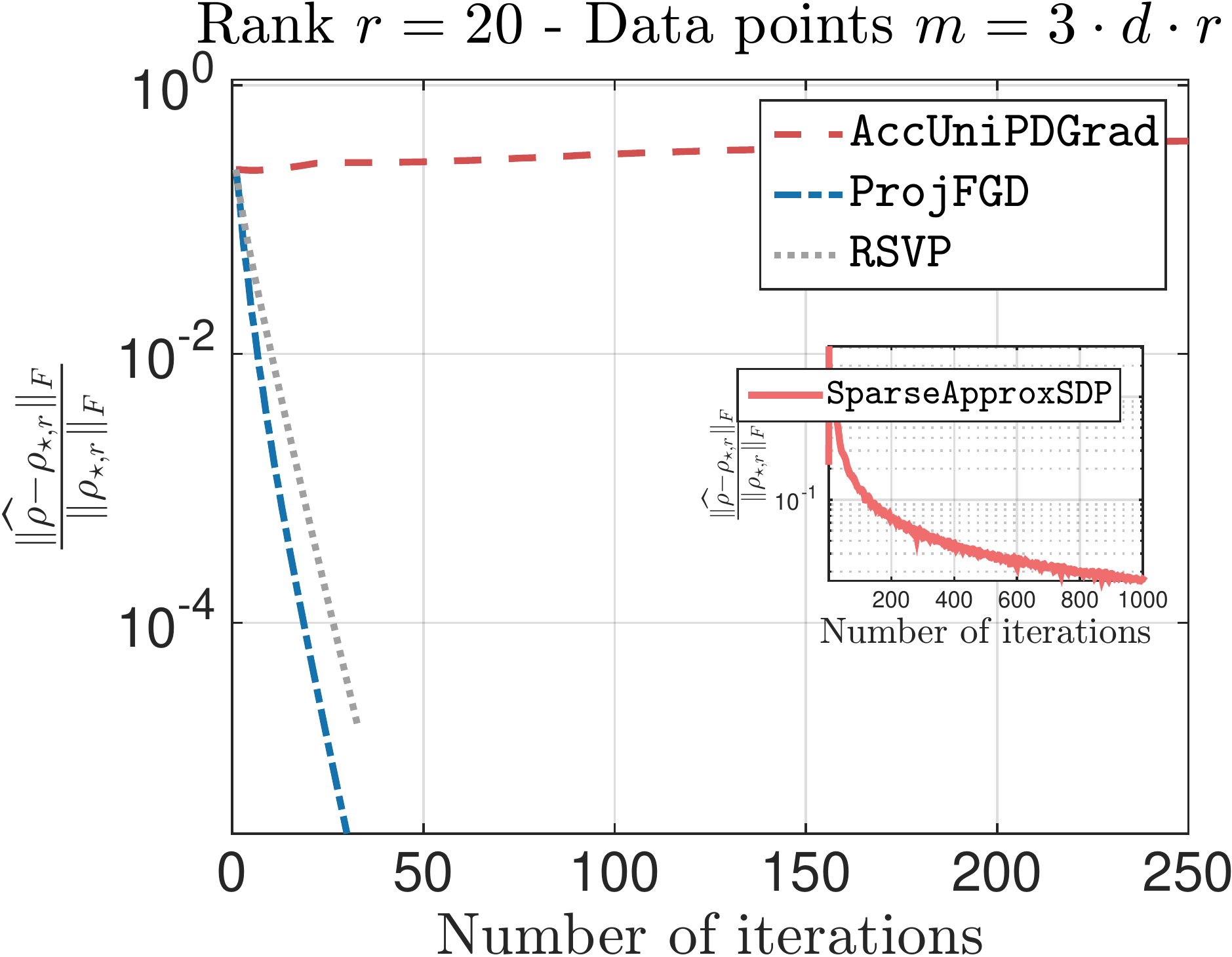} 
\includegraphics[width=0.24\textwidth]{./QST_r_20_1} 
\caption{Convergence performance of algorithms under comparison w.r.t. $\tfrac{\|\widehat{\rho} - \rhoo\|_F}{\|\rhoo\|_F}$ vs. $(i)$ the total number of iterations (left) and $(ii)$ the total execution time (right). The two left plots correspond to the case $r = 5$ and the two right plots to the case $r = 20$. In all cases $C_{\rm sam} = 3$ and $n = 10$. Initial point is $U_0 = \Pi_{\C}(\widetilde{\U_0})$ such that $\X_0 = \widetilde{\U}_0 \widetilde{\U}_0^\dagger$ where $\X_0 = \Pi_{+}\left(-\mathcal{A}^*(y)\right)$.
}
\label{fig:app_exp5}
\end{figure*}

\bibliography{nonConvexQST}
\end{document}